\documentclass[11pt,twoside,a4paper]{article}

\def\DONOTINSERTCOMMENTS{}

\usepackage{ifthen}

\newboolean{conferenceversion}
\setboolean{conferenceversion}{false}

\usepackage[letterpaper,margin=2.5cm]{geometry}
\usepackage{tgtermes}
\usepackage[scale=.92]{tgheros}
\usepackage{tgcursor}

\usepackage[noblocks]{authblk}

\usepackage[english]{babel}

\usepackage{fancyhdr}

\usepackage{subcaption}

\usepackage{bussproofs}

\usepackage{ifthen}

\usepackage{ifthen}

\provideboolean{detectedSTOC}
\provideboolean{detectedFOCS}
\provideboolean{detectedElsevier}
\provideboolean{detectedNOW}
\provideboolean{detectedLMCS}
\provideboolean{detectedIEEE}
\provideboolean{detectedPoster}
\provideboolean{detectedSIAM}
\provideboolean{detectedLNCS}
\provideboolean{detectedACM}
\provideboolean{detectedACMconf}
\provideboolean{detectedSigplanconf}
\provideboolean{detectedToC}
\provideboolean{detectedLIPIcs}
\provideboolean{detectedAAAI}
\provideboolean{detectedIJCAI}
\provideboolean{detectedCompCplx}
\provideboolean{detectedEasyChair}
\provideboolean{detectedArticle}
\provideboolean{detectedReport}
\provideboolean{detectedThesis}

\makeatletter

\@ifclassloaded{sig-alternate}
{\setboolean{detectedSTOC}{true}}
{\setboolean{detectedSTOC}{false}}

\@ifclassloaded{elsarticle}
{\setboolean{detectedElsevier}{true}}
{\setboolean{detectedElsevier}{false}}

\@ifclassloaded{now}
{\setboolean{detectedNOW}{true}}
{\setboolean{detectedNOW}{false}}

\@ifclassloaded{lmcs}
{\setboolean{detectedLMCS}{true}}
{\setboolean{detectedLMCS}{false}}

\@ifclassloaded{IEEEtran} {
  \setboolean{detectedIEEE}{true}
  \ifCLASSOPTIONconference {
    \setboolean{detectedFOCS}{true}
  }
  \else {
    \setboolean{detectedFOCS}{false}
  }
  \fi
}
{
  \setboolean{detectedFOCS}{false}
  \setboolean{detectedIEEE}{false}
}

\@ifclassloaded{siamart171218}
{\setboolean{detectedSIAM}{true}}
{\setboolean{detectedSIAM}{false}}

\@ifclassloaded{llncs}
{\setboolean{detectedLNCS}{true}}
{\setboolean{detectedLNCS}{false}}

\@ifclassloaded{acmsmall}
{\setboolean{detectedACM}{true}}
{\setboolean{detectedACM}{false}}

\@ifclassloaded{acmart}
{\setboolean{detectedACMconf}{true}}
{\setboolean{detectedACMconf}{false}}

\@ifclassloaded{sigplanconf}
{\setboolean{detectedSigplanconf}{true}}
{\setboolean{detectedSigplanconf}{false}}

\@ifclassloaded{toc}
{\setboolean{detectedToC}{true}}
{\setboolean{detectedToC}{false}}

\@ifclassloaded{lipics}
{\setboolean{detectedLIPIcs}{true}}
{\@ifclassloaded{lipics-v2019}
  {\setboolean{detectedLIPIcs}{true}}
  {\@ifclassloaded{oasics-v2019}
    {\setboolean{detectedLIPIcs}{true}}
    {\setboolean{detectedLIPIcs}{false}}
  }
}

\@ifclassloaded{cc}
{\setboolean{detectedCompCplx}{true}}
{\setboolean{detectedCompCplx}{false}}

\@ifclassloaded{easychair}
{\setboolean{detectedEasyChair}{true}}
{\setboolean{detectedEasyChair}{false}}

\@ifpackageloaded{aaai}
{\setboolean{detectedAAAI}{true}}        
{\@ifpackageloaded{aaai18}
  {\setboolean{detectedAAAI}{true}}
  {\@ifpackageloaded{aaai20}       
    {\setboolean{detectedAAAI}{true}}
    {\setboolean{detectedAAAI}{false}}}}

\@ifpackageloaded{ijcai18}
{\setboolean{detectedIJCAI}{true}}        
{\@ifpackageloaded{ijcai19}
  {\setboolean{detectedIJCAI}{true}}        
  {\setboolean{detectedIJCAI}{false}}}

\@ifclassloaded{sciposter}
{\setboolean{detectedPoster}{true}}
{\setboolean{detectedPoster}{false}}

\@ifclassloaded{article}
{\setboolean{detectedArticle}{true}}
{\setboolean{detectedArticle}{false}}

\makeatother

\ifthenelse{\not \isundefined{\examen} 
  \and \not \isundefined{\disputationsdatum} 
  \and \not \isundefined{\disputationslokal}}   
  {\setboolean{detectedThesis}{true}}
  {\setboolean{detectedThesis}{false}}

\ifthenelse{\boolean{detectedArticle} \or \boolean{detectedThesis}
  \or \boolean{detectedSTOC} \or \boolean{detectedFOCS}
  \or \boolean{detectedSIAM} \or \boolean{detectedIEEE}
  \or \boolean{detectedPoster}}
{\setboolean{detectedReport}{false}}
{\setboolean{detectedReport}{true}}

\usepackage[utf8]{inputenc}

\usepackage[cmex10]{amsmath}
\usepackage{amssymb}
\usepackage{amsfonts}

\usepackage{mathtools}

\usepackage{xspace}

\ifthenelse
{\boolean{detectedElsevier}    \or \boolean{detectedFOCS} \or \boolean{detectedSTOC}
  \or \boolean{detectedPoster} \or \boolean{detectedSIAM} \or \boolean{detectedLMCS}
  \or \boolean{detectedIEEE}   \or \boolean{detectedNOW}  \or \boolean{detectedToC}
  \or \boolean{detectedThesis} \or \boolean{detectedLNCS} \or \boolean{detectedACM}} 
{}
{\usepackage{nada-typography}}

\usepackage{ifthen}
\usepackage{xspace}
\ifthenelse
{\boolean{detectedElsevier}}
{}
{\usepackage{varioref}}

\DeclareMathAlphabet{\mathsfsl}{OT1}{cmss}{m}{sl}

\newcommand{\formatfunctiontoset}[1]{\mathit{#1}}

\newcommand{\introduceterm}[1]{{\emph{#1}}}

\newcommand{\eqperiod}{\enspace .}

\ifthenelse{\boolean{detectedToC}}{}
  {\newcommand{\eg}{for instance\xspace} %
    }

\ifthenelse{\boolean{detectedToC}}{}{\newcommand{\ie}{i.e.,\ }}

\ifthenelse{\boolean{detectedLIPIcs}}
{}
{}

\newcommand{\etal}{et al.\@\xspace}

\newcommand{\wolog}{without loss of generality\xspace}

\ifthenelse{\boolean{detectedIEEE}}{}{}

\newcommand{\Bigoh}[1]{\mathrm{O} \bigl( #1 \bigr)}
\newcommand{\bigoh}[1]{\mathrm{O} ( #1 )}

\newcommand{\bigtheta}[1]{\Theta ( #1 )}

\newcommand{\Bigomega}[1]{\Omega \bigl( #1 \bigr)}
\newcommand{\bigomega}[1]{\Omega ( #1 )}

\newcommand{\complclassformat}[1]%
        {\textrm{\upshape{\textsf{#1}}}\xspace}

\newcommand{\cocomplclass}[1]%
        {\textrm{\upshape{\textsf{co#1}}}\xspace}

\newcommand{\DTIMEadviceclass}[2]%
    {\ensuremath{\complclassformat{DTIME}\bigl(#1\bigr)/{#2}}}
\newcommand{\Pclass}{\complclassformat{P}}
\newcommand{\NP}{\complclassformat{NP}}

\newcommand{\PSPACE}{\complclassformat{PSPACE}}

\newcommand{\pspace}{\PSPACE}

\newcommand{\NC}{\complclassformat{NC}}

\newcommand{\refsec}[1]{Section~\ref{#1}}

\newcommand{\reffig}[1]{Figure~\ref{#1}}

\newcommand{\reftwofigs}[2]{Figures~\ref{#1} and~\ref{#2}}

\newcommand{\refth}[1]{Theorem~\ref{#1}}

\newcommand{\reflem}[1]{Lemma~\ref{#1}}
\newcommand{\reftwolems}[2]{Lemmas~\ref{#1} and~\ref{#2}}

\newcommand{\refpr}[1]{Proposition~\ref{#1}}

\newcommand{\refcor}[1]{Corollary~\ref{#1}}

\newcommand{\refdef}[1]{Definition~\ref{#1}}

\newcommand{\refalg}[1]{Algorithm~\ref{#1}}

\newcommand{\refitem}[1]{item~\ref{#1}}

\ifthenelse
{\isundefined{\refeq}}
{\newcommand{\refeq}[1]{\eqref{#1}}}
{\renewcommand{\refeq}[1]{\eqref{#1}}}

\newcommand{\defeq}{\,\overset{\mathrm{def}}{=}\,}

\newcommand{\ceiling}[1]{\lceil #1 \rceil}

\newcommand{\floor}[1]{\lfloor #1 \rfloor}

\ifthenelse{\boolean{detectedToC}}{}
{

  \newcommand{\N}         {\mathbb{N}}
  \newcommand{\Nplus}     {\mathbb{N}^{+}}

}

\newcommand{\MAXOFEXPR}[2][]{\max_{#1} \left\{ #2 \right\}}
\newcommand{\MINOFEXPR}[2][]{\min_{#1} \left\{ #2 \right\}}
\newcommand{\Maxofexpr}[2][]{\max_{#1} \bigl\{ #2 \bigr\}}
\newcommand{\Minofexpr}[2][]{\min_{#1} \bigl\{ #2 \bigr\}}
\newcommand{\maxofexpr}[2][]{\max_{#1} \{ #2 \}}

\newcommand{\MAXOFSET}[3][:]{\max \left\{ #2 #1 #3 \right\}}
\newcommand{\MINOFSET}[3][:]{\min \left\{ #2 #1 #3 \right\}}
\newcommand{\Maxofset}[3][:]{\max \bigl\{ #2 #1 #3 \bigr\}}
\newcommand{\Minofset}[3][:]{\min \bigl\{ #2 #1 #3 \bigr\}}

\newcommand{\minofset}[3][:]{\min \{ #2 #1 #3 \}}
 
\renewcommand{\MAXOFSET}[3][:]%
     {\ifthenelse{\equal{#1}{;}}%
     {\MAXOFEXPR{ #2 \,;\, #3 }}
     {\ifthenelse{\equal{#1}{:}}%
     {\MAXOFEXPR{ #2 \,:\, #3 }}
     {\max \twincommandJN{\left\{}{#2}{\left#1}{\right}{\,#3}{\right\}}}}}
\renewcommand{\MINOFSET}[3][:]%
     {\ifthenelse{\equal{#1}{;}}%
     {\MINOFEXPR{ #2 \,;\, #3 }}
     {\ifthenelse{\equal{#1}{:}}%
     {\MINOFEXPR{ #2 \,:\, #3 }}
     {\min \twincommandJN{\left\{}{#2}{\left#1}{\right}{\,#3}{\right\}}}}}

\renewcommand{\Maxofset}[3][:]%
     {\ifthenelse{\equal{#1}{;}}%
     {\Maxofexpr{ #2 \,;\, #3 }}
     {\ifthenelse{\equal{#1}{:}}%
     {\Maxofexpr{ #2 \,:\, #3 }}
     {\max \twincommandJN{\bigl\{}{#2}{\bigl#1}{\bigr}{\,#3}{\bigr\}}}}}
\renewcommand{\Minofset}[3][:]%
     {\ifthenelse{\equal{#1}{;}}%
     {\Minofexpr{ #2 \,;\, #3 }}
     {\ifthenelse{\equal{#1}{:}}%
     {\Minofexpr{ #2 \,:\, #3 }}
     {\min \twincommandJN{\bigl\{}{#2}{\bigl#1}{\bigr}{\,#3}{\bigr\}}}}}

\DeclareMathOperator{\Expop}{E}

\ifthenelse{\boolean{detectedLMCS}}
{}
{}

\newcommand{\twincommandJN}[6]%
    {#1#2#3\vphantom{#2#5}\mspace{-2.05mu}#4.#5#6}

\newcommand{\CondExp}[2]%
    {\Expop\twincommandJN{\bigl[}{#1}{\bigl|}{\bigr}{\,#2}{\bigr]}}
\newcommand{\CONDEXP}[2]%
     {\Expop\twincommandJN{\left[}{#1}{\left|}{\right}{\,#2}{\right]}}

\newcommand{\CondProb}[3][]%
    {\Pr_{#1}\twincommandJN{\bigl[}{#2}{\bigl|}{\bigr}{\,#3}{\bigr]}}
\newcommand{\CONDPROB}[3][]%
    {\Pr_{#1}\twincommandJN{\left[}{#2}{\left|}{\right}{\,#3}{\right]}}

\newcommand{\isdistras}[2]{\ensuremath{#1} \sim \ensuremath{#2}}
\newcommand{\funcdescr}[3]{\ensuremath{ #1 : #2 \to #3}}

\newcommand{\edges}[1]{E( #1 )}

\newcommand{\vertices}[1]{V( #1 )}

\newcommand{\succname}{\formatfunctiontoset{succ}}
\newcommand{\predname}{\formatfunctiontoset{pred}}

\newcommand{\ancname}{\formatfunctiontoset{anc}}
\newcommand{\descname}{\formatfunctiontoset{desc}}
\newcommand{\pancname}{\formatfunctiontoset{anc}^*}
\newcommand{\pdescname}{\formatfunctiontoset{desc}^*}

\newcommand{\succnode}[2][]{\succname_{#1}(#2)}
\newcommand{\prednode}[2][]{\predname_{#1}(#2)}

\newcommand{\ancnode}[2][]{\ancname_{#1}(#2)}
\newcommand{\pdescnode}[2][]{\pdescname_{#1}(#2)}
\newcommand{\pancnode}[2][]{\pancname_{#1}(#2)}
\ifthenelse{\isundefined{\descnode}}
{\newcommand{\descnode}[2][]{\descname_{#1}(#2)}}
{\renewcommand{\descnode}[2][]{\descname_{#1}(#2)}}

\newcommand{\Ancnode}[2][]{\ancname_{#1}\bigl(#2\bigr)}

\newcommand{\Pancnode}[2][]{\pancname_{#1}\bigl(#2\bigr)}

\newcommand{\setcompact}[1]{{\ensuremath{\bigl\{ #1 \bigr\}}}}

\newcommand{\setdescrcompact}[3][\mid]{{\setcompact{ #2 #1 #3 }}}

\newcommand{\set}[1]{\{ #1 \}}
\newcommand{\Set}[1]{\bigl\{ #1 \bigr\}}

\newcommand{\Setdescr}[3][|]%
     {\ifthenelse{\equal{#1}{;}}%
     {\Set{ #2 \,;\, #3 }}
     {\ifthenelse{\equal{#1}{:}}%
     {\Set{ #2 \,:\, #3 }}
     {\twincommandJN{\bigl\{}{#2\,}{\bigl#1}{\bigr}{\,#3}{\bigr\}}}}}
\newcommand{\SETDESCR}[3][|]%
     {\twincommandJN{\left\{}{#2\,}{\left#1}{\right}{\,#3}{\right\}}}

\newcommand{\Setdescrbrackets}[3][|]%
     {\twincommandJN{\bigl[}{#2}{\bigl#1}{\bigr}{\,#3}{\bigr]}}
\newcommand{\SETDESCRBRACKETS}[3][|]%
     {\twincommandJN{\left[}{#2}{\left#1}{\right}{\,#3}{\right]}}

\newcommand{\setsize}[1]{\lvert#1\rvert}

\newcommand{\intersection}{\cap}

\newcommand{\union}{\cup}
\newcommand{\Union}{\bigcup}

\newcommand{\unionSP}{\, \union \, }

\newcommand{\disjointunion}{\overset{.}{\cup}}

\newcommand{\DisjointunionInText}%
    {{\smash{\overset{\mbox{\boldmath{.}}}{\bigcup}}}\vphantom{\bigcup}}
\newcommand{\intnfirst}[1]{[{#1}]}

\newcommand{\Lor}{\bigvee}
\newcommand{\Land}{\bigwedge}

\newcommand{\olnot}[1]{\overline{#1}}
\newcommand{\stdnot}[1]{\olnot{#1}}

\newcommand{\false}{\bot}
\newcommand{\true}{\top}

\newcommand{\cnfform}{\cnfshort for\-mu\-la\xspace}

\newcommand{\cnfshort}{CNF\xspace}

\newcommand{\xcnf}[1]{\mbox{\ensuremath{#1}-CNF}\xspace}

\newcommand{\xcnfform}[1]{\mbox{\ensuremath{#1}-}\cnfform}
\newcommand{\kcnfform}{\xcnfform{\clwidth}}
\newcommand{\xclause}[1]{\mbox{\ensuremath{#1}-clause}\xspace}

\newcommand{\nvar}{n}

\newcommand{\nclause}{m}

\newcommand{\clwidth}{k}

\newcommand{\randkcnfnclwrepl}[3][\clwidth]%
        {\ensuremath{\mathcal{F}^{#2, #3}_{#1}}}
\newcommand{\randkcnfnclwreplstd}%
        {\randkcnfnclwrepl{\clwidth}{\nvar}{\nclause}}

\newcommand{\israndkcnfnclwrepl}[4]%
  {\isdistras{#1}{\randkcnfnclwrepl[#2]{#3}{#4}}}

\newcommand{\randkcnfprobcl}[3]%
        {\ensuremath{\mathcal{F}^{#2}_{#1} \bigl(#3 \bigr)}}

\newcommand{\pcfor}[4][to]{for #2 := #3 #1 #4 do}
\newcommand{\pcformath}[4][to]%
    {\pcfor[#1]{\ensuremath{#2}}{\ensuremath{#3}}{\ensuremath{#4}}}

\newcommand{\pcassigncompact}[2]{#1 := #2}
\newcommand{\pcassignmathcompact}[2]%
        {\pcassigncompact{\ensuremath{#1}}{\ensuremath{#2}}}

\newcommand{\inductionformat}[1]{\textit{#1}}

\newcommand{\BASE}[1][]
        {\inductionformat
                {%
                        \ifthenelse{\equal{#1}{}}%
                                {Base case: }%
                                {Base case (#1):}%
                }%
        }

\ifthenelse{\isundefined{\DONOTLOADHYPERREFPACKAGE}
  \and \not \boolean{detectedSTOC}     \and \not \boolean{detectedFOCS}
  \and \not \boolean{detectedElsevier} \and \not \boolean{detectedPoster} 
  \and \not \boolean{detectedSIAM}     \and \not \boolean{detectedACM}
  \and \not \boolean{detectedIEEE}     \and \not \boolean{detectedNOW}
  \and \not \boolean{detectedToC}      \and \not \boolean{detectedThesis}
  \and \not \boolean{detectedLNCS}}
{\usepackage[colorlinks=true,citecolor=blue]{hyperref}}
{}

\ifthenelse    
{\boolean{detectedSTOC} \or \boolean{detectedSIAM} \or \boolean{detectedLMCS}
 \or \boolean{detectedNOW} \or \boolean{detectedLNCS} \or \boolean{detectedACM}
 \or \boolean{detectedToC}}
{}
{
  \usepackage{amsthm}
  \usepackage{amsthm-boldfont}
}

\ifthenelse{\boolean{detectedSTOC}}                  
  {\input{definitions-STOC.tex}}
  {\ifthenelse{\boolean{detectedSIAM}}
    {\input{definitions-SIAM.tex}}
    {\ifthenelse{\boolean{detectedNOW}}
      {\input{definitions-NOW.tex}}
      {\ifthenelse{\boolean{detectedLMCS}}
        {\input{definitions-lmcs.tex}}
        {\ifthenelse{\boolean{detectedIEEE}}
          {\input{definitions-IEEE.tex}}
          {\ifthenelse{\boolean{detectedLNCS}}
            {\input{definitions-LNCS.tex}}
            {\ifthenelse{\boolean{detectedACM}}
              {\input{definitions-ACM.tex}}
              {\ifthenelse{\boolean{detectedToC}}
                {\input{definitions-ToC.tex}}
                {\ifthenelse{\boolean{detectedArticle} \or \boolean{detectedElsevier}}
                  {%

\numberwithin{equation}{section}

   \usepackage{thmtools}
   \usepackage[nameinlink, capitalize]{cleveref} %

   \makeatletter
   \def\thmt@refnamewithcomma #1#2#3,#4,#5\@nil{%
     \@xa\def\csname\thmt@envname #1utorefname\endcsname{#3}%
     \ifcsname #2refname\endcsname
       \csname #2refname\expandafter\endcsname\expandafter{\thmt@envname}{#3}{#4}%
     \fi
   }
   \makeatother

\declaretheorem[name=Theorem,parent=section]{theorem}
\declaretheorem[name=Lemma,sibling=theorem]{lemma}
\declaretheorem[name=Claim,sibling=theorem,Refname={Claim,Claims}]{claim}
\declaretheorem[name=Corollary,sibling=theorem]{corollary}
\declaretheorem[name=Proposition,sibling=theorem]{proposition}
\declaretheorem[name=Definition,style=definition,sibling=theorem]{definition}

\declaretheorem[name=Construction,style=definition,sibling=theorem,
Refname={Construction,Constructions}]{construction}

\declaretheorem[name=Algorithm,sibling=theorem]{algorithm}

\Crefname{section}{Section}{Sections}
\Crefname{inequality}{Inequality}{Inequalities}

}
                  {\input{definitions-report.tex}}}}}}}}}}

\ifthenelse{\boolean{detectedThesis}}
{}
{\usepackage{ifpdf}}

\ifthenelse{\boolean{detectedToC}}
{}
{\usepackage{graphicx}}

\ifpdf         
\fi

\ifthenelse
{\boolean{detectedSTOC}  \or \boolean{detectedThesis} \or 
 \boolean{detectedLNCS}  \or \boolean{detectedToC}}    
{}
{
  \setcounter{secnumdepth}{3}
  \setcounter{tocdepth}{3}
}

\newcount\hours
\newcount\minutes
\def\SetTime{\hours=\time
\global\divide\hours by 60
\minutes=\hours
\multiply\minutes by 60
\advance\minutes by-\time
\global\multiply\minutes by-1 }
\SetTime
\def\now{\number\hours:\ifnum\minutes<10 0\fi\number\minutes}

\newcommand{\formuladots}{\cdots}

\DeclareFontFamily{OT1}{pzc}{}
\DeclareFontShape{OT1}{pzc}{m}{it}{<-> s * [1.200] pzcmi7t}{}
\DeclareMathAlphabet{\mathpzc}{OT1}{pzc}{m}{it}

\newcommand{\deriveswithall}%
        {\vdash_{\!\!\!{\scriptscriptstyle \forall}}} 
\newcommand{\notderiveswithall}%
        {\nvdash_{\!\!\!{\scriptscriptstyle \forall}}} 

\newcommand{\clcfgtransitioncrammed}[2]%
        {\ensuremath{#1 \!\rightsquigarrow\! #2}}

\newcommand{\fstd}{{\ensuremath{F}}}

\newcommand{\varx}{\ensuremath{x}}

\newcommand{\lita}{\ensuremath{a}}

\newcommand{\clc}{\ensuremath{C}}

\newcommand{\setsofvarsorlitlarge}[2]%
        {\mathit{#1}\left({#2}\right)}
\newcommand{\setsofvarsorlit}[2]%
        {\mathit{#1}({#2})}
\newcommand{\setsofvarsorlitcompact}[2]%
        {\mathit{#1}\bigl({#2}\bigr)}

\newcommand{\setsofvarsorlitsup}[3]%
        {\mathit{#1}^{#2}({#3})}
\newcommand{\setsofvarsorlitsuplarge}[3]%
        {\mathit{#1}^{#2}\left({#3}\right)}
\newcommand{\setsofvarsorlitsupcompact}[3]%
        {\mathit{#1}^{#2}\bigl({#3}\bigr)}

\newcommand{\restrict}[2]{#1\!\!\upharpoonright_{#2}}

\newcommand{\derivabbrev}[2]{\bigl( #1 \vdash #2 \bigr)}
\newcommand{\derivabbrevsmall}[2]{( #1 \vdash #2 )}
\newcommand{\derivabbrevcompact}[2]{\bigl( #1 \vdash #2 \bigr)}

\newcommand{\refutabbrevsmall}[1]{\derivabbrevsmall{#1}{\falsenum}}
\newcommand{\refutabbrevcompact}[1]{\derivabbrevcompact{#1}{\falsenum}}

\renewcommand{\refutabbrevsmall}[1]{\derivabbrevsmall{#1}{\!\emptycl}}
\renewcommand{\refutabbrevcompact}[1]{\derivabbrevcompact{#1}{\!\emptycl}}

\renewcommand{\refutabbrevsmall}[1]{\derivabbrevsmall{#1}{\!\bot}}
\renewcommand{\refutabbrevcompact}[1]{\derivabbrevcompact{#1}{\!\bot}}

\newcommand{\genericrefsmall}[3]%
    {{\mathit{#1}}_{#2}\refutabbrevsmall{#3}}
\newcommand{\genericrefcompact}[3]%
    {{\mathit{#1}}_{#2}\refutabbrevcompact{#3}}
\newcommand{\genericderiv}[4]%
    {{\mathit{#1}}_{#2}\derivabbrev{#3}{#4}}
\newcommand{\genericderivsmall}[4]%
    {{\mathit{#1}}_{#2}\derivabbrevsmall{#3}{#4}}
\newcommand{\genericderivcompact}[4]%
    {{\mathit{#1}}_{#2}\derivabbrevcompact{#3}{#4}}
\newcommand{\generictaut}[3]%
    {{\mathit{#1}}_{#2}\derivabbrev{}{#3}}
\newcommand{\generictautcompact}[3]%
    {{\mathit{#1}}_{#2}\derivabbrevcompact{}{#3}}
\newcommand{\generictautsmall}[3]%
    {{\mathit{#1}}_{#2}\derivabbrevsmall{}{#3}}

\newcommand{\formulaformat}[1]{\ensuremath{\mathit{#1}}}
\renewcommand{\formulaformat}[1]{\mathit{#1}}

\newcommand{\formatconfiguration}[1]{\mathbb{#1}}

\newcommand{\formatpebblingstrategy}[1]{\mathcal{#1}}

\newcommand{\transitionarrow}{\rightsquigarrow}
\newcommand{\pebcfgtransition}[2]%
    {\ensuremath{#1 \transitionarrow #2}}
\newcommand{\pebcfgtransitionsqueeze}[2]%
    {#1 \! \transitionarrow \! #2}

\newcommand{\pebbling}[1][P]{\formatpebblingstrategy{#1}}

\newcommand{\pconf}[1][P]{\formatconfiguration{#1}}

\newcommand{\formatpebblingprice}[1]{\textsl{\textsf{#1}}}

\newcommand{\Pebblingprice}[1]%
    {\formatpebblingprice{Peb}\bigl(#1\bigr)}
\newcommand{\pebblingpricecompact}[1]%
    {\formatpebblingprice{Peb}\bigl(#1\bigr)}

\newcommand{\Bwpebblingprice}[1]%
    {\formatpebblingprice{BW-Peb}\bigl(#1\bigr)}
\newcommand{\bwpebblingpricecompact}[1]%
    {\formatpebblingprice{BW-Peb}\bigl(#1\bigr)}

\newcommand{\pebpersistentsymbol}{\bullet}
\newcommand{\pebvisitingsymbol}{\emptyset}

\newcommand{\bwpebpricepersistent}[1]%
    {\formatpebblingprice{BW-Peb}^{\pebpersistentsymbol}(#1)}
\newcommand{\Bwpebpricepersistent}[1]%
    {\formatpebblingprice{BW-Peb}^{\pebpersistentsymbol}\bigl(#1\bigr)}
\newcommand{\bwpebpricevisiting}[1]%
    {\formatpebblingprice{BW-Peb}^{\pebvisitingsymbol}(#1)}
\newcommand{\Bwpebpricevisiting}[1]%
    {\formatpebblingprice{BW-Peb}^{\pebvisitingsymbol}\bigl(#1\bigr)}

\newcommand{\pebpricepersistent}[1]%
    {\formatpebblingprice{Peb}^{\pebpersistentsymbol}(#1)}
\newcommand{\Pebpricepersistent}[1]%
    {\formatpebblingprice{Peb}^{\pebpersistentsymbol}\bigl(#1\bigr)}
\newcommand{\pebpricevisiting}[1]%
    {\formatpebblingprice{Peb}^{\pebvisitingsymbol}(#1)}
\newcommand{\Pebpricevisiting}[1]%
    {\formatpebblingprice{Peb}^{\pebvisitingsymbol}\bigl(#1\bigr)}

\newcommand{\bwpebblingpriceempty}[1]%
    {\formatpebblingprice{BW-Peb}^{\pebvisitingsymbol}(#1)}
\newcommand{\bwpebblingpriceemptycompact}[1]%
    {\formatpebblingprice{BW-Peb}^{\pebvisitingsymbol}\bigl(#1\bigr)}

\newcommand{\stoptime}{\tau}

\newcommand{\pebspace}[1]{\formatpebblingprice{space} ( #1 )}

\newcommand{\pebtime}[1]{\formatpebblingprice{time} ( #1 )}

\newcommand{\pebdeg}{\ensuremath{d}}

\newcommand{\pebaxcompact}[2]%
        [\pebdeg]{\ensuremath{\formulaformat{Ax}^{#1} \bigl(#2 \bigr)}}

\newcommand{\pqrxvar}[6]%
    {\ensuremath{\stdnot{\varx({#1})}_{#2} \lor \stdnot{\varx({#3})}_{#4} \lor %
    \sourceclausexvar[#6]{#5}}}

\newcommand{\pqr}[6]%
    {\ensuremath{\stdnot{#1}_{#2} \lor \stdnot{#3}_{#4} \lor %
    \sourceclausenodisplay[#6]{#5}}}
\newcommand{\pqrstd}{\pqr{p}{i}{q}{j}{r}{l}}
\newcommand{\pqrall}[6]%
        {\setdescrcompact
        {\pqr{#1}{#2}{#3}{#4}{#5}{#6}}{#2,#4 \in \intnfirst{\pebdeg}}}
\newcommand{\pqrallstd}%
        {\setdescrcompact{\pqrstd}{i,j \in \intnfirst{\pebdeg}}}

\newcommand{\sourceclausexvar}[2][n]%
        {\Lor_{#1 = 1}^{\pebdeg} \varx({#2})_{#1}}
\newcommand{\subsourceclausexvar}[3][n]%
        {\Lor_{#1 = {#2}}^{\pebdeg} \varx({#3})_{#1}}

\newcommand{\sourceclausexvarnodisplay}[2][n]%
        {\textstyle \Lor_{#1 = 1}^{\pebdeg} \varx({#2})_{#1}}

\newcommand{\sourceclausenodisplay}[2][n]%
        {\textstyle \Lor_{#1 = 1}^{\pebdeg} #2_{#1}}

\newcommand{\relativisation}[1]%
    {\ensuremath{\formulaformat{Rel}\bigl(#1 \bigr)}}

\newcommand{\extPHPnot}[2]
    {\ensuremath{\extendedversion{\formulaformat{PHP}}^{#1}_{#2}}}

\newcommand{\GraphOntoPHPnot}[1][G]%
    {\text{$\formulaformat{Onto}$-$\formulaformat{PHP}$}_{#1}}

\renewcommand{\extPHPnot}[2]%
    {\ephpnot{#1}{#2}}

\newcommand{\ephpnot}[2]%
    {\vphantom{\extendedversion{\formulaformat{PHP}}}
      {\smash{\extendedversion{\formulaformat{PHP}}}
        \vphantom{\formulaformat{PHP}}}^{#1}_{#2}}

\newcommand{\efphpnot}[2]%
    {\vphantom{\extendedversion{\formulaformat{FPHP}}}
      {\smash{\extendedversion{\formulaformat{FPHP}}}
        \vphantom{\formulaformat{FPHP}}}^{#1}_{#2}}

\newcommand{\ontophpnot}[2]%
    {\formulaformat{Onto}\text{-}\formulaformat{PHP}^{#1}_{#2}}

\newcommand{\ontofphpnot}[2]%
    {\formulaformat{Onto}\text{-}\formulaformat{FPHP}^{#1}_{#2}}

\newcommand{\extendedversion}[1]{\widetilde{#1}}

\usepackage{stmaryrd} 

\newcommand{\formatfunctiontosubconfiguration}[1]{\mathsf{#1}}
\newcommand{\formatfunctiontomulti}[1]{\mathcal{#1}}

\newcommand{\pebconditional}{conditional\xspace}

\newcommand{\pebunconditional}{unconditional\xspace}

\newcommand{\pebcomplete}{complete\xspace}

\DeclareMathOperator{\dummystar}{*}
\newcommand{\pebblingcontrNT}[2][G]%
 {\ensuremath{\dummystar\!\!\formulaformat{Peb}^{#2}_{#1}}}

\newcommand{\somenodetrueclausedeg}[2]{\formulaformat{All}_{#1}^{+}({#2})}

\newcommand{\slashedstrickenletter}[1]{{\backslash\mkern-9mu #1}}
            
\newcommand{\strikethroughcommand}[1]{\slashedstrickenletter{#1}}

\newcommand{\abovevertices}[2][G]%
    {{#1}_{#2}^{\hspace{-0.2 pt}\triangledown}}
\newcommand{\aboveverticesNR}[2][G]%
    {{#1}_{\strikethroughcommand{#2}}^{\hspace{-0.3 pt}\triangledown}}
\newcommand{\belowvertices}[2][G]%
    {{#1}^{#2}_{\hspace{-0.6 pt}\vartriangle}}
\newcommand{\belowverticesNR}[2][G]%
    {{#1}^{\strikethroughcommand{#2}}_{\hspace{-0.6 pt}\vartriangle}}

\newcommand{\lpebblingpricecompact}[1]%
    {\formatpebblingprice{L-Peb}\bigl(#1\bigr)}

\newcommand{\scnot}[2]{#1 \langle #2 \rangle}
\newcommand{\scnotcompact}[2]{#1 \bigl\langle #2 \bigr\rangle}

\newcommand{\spcanonconfcompact}[1]%
        {\formatfunctiontosubconfiguration{canon}\bigl({#1}\bigr)}

\newcommand{\spprojsubsub}[4]%
    {\formatfunctiontosubconfiguration{proj}_{\scnot{#1}{#2}}(\scnot{#3}{#4})}
\newcommand{\spprojsubsubcompact}[4]%
    {\formatfunctiontosubconfiguration{proj}_{\scnot{#1}{#2}}%
    \bigl(\scnot{#3}{#4}\bigr)}
\newcommand{\spprojsubconf}[3]%
    {\formatfunctiontosubconfiguration{proj}_{\scnot{#1}{#2}}({#3})}
\newcommand{\spprojsubconfcompact}[3]%
    {\formatfunctiontosubconfiguration{proj}_{\scnot{#1}{#2}}\bigl({#3}\bigr)}
\newcommand{\spprojconfsub}[3]%
    {\formatfunctiontosubconfiguration{proj}_{#1}(\scnot{#2}{#3})}
\newcommand{\spprojconfsubcompact}[3]%
    {\formatfunctiontosubconfiguration{proj}_{#1}\bigl(\scnot{#2}{#3}\bigr)}
\newcommand{\spprojconfconf}[2]%
    {\formatfunctiontosubconfiguration{proj}_{#1}({#2})}
\newcommand{\spprojconfconfcompact}[2]%
    {\formatfunctiontosubconfiguration{proj}_{#1}\bigl({#2}\bigr)}

\newcommand{\spclossubcompact}[2]%
        {\formatfunctiontoset{cl}\bigl(\scnotcompact{#1}{#2}\bigr)}

\newcommand{\spintersubcompact}[2]%
        {\formatfunctiontoset{int}\bigl(\scnotcompact{#1}{#2}\bigr)}

\newcommand{\spcoversubcompact}[2]%
        {\formatfunctiontoset{cover}\bigl(\scnotcompact{#1}{#2}\bigr)}

\newcommand{\spcoverconfcompact}[1]%
        {\formatfunctiontoset{cover}\bigl({#1}\bigr)}

\newcommand{\spinducedblack}[1]%
    {\formatfunctiontoset{Bl} (#1)}
\newcommand{\spinducedwhite}[1]%
    {\formatfunctiontoset{Wh} (#1)}
\newcommand{\spinducedblackcompact}[1]%
    {\formatfunctiontoset{Bl} \bigl(#1 \bigr)}
\newcommand{\spinducedwhitecompact}[1]%
    {\formatfunctiontoset{Wh} \bigl(#1 \bigr)}

\newcommand{\pathclausedeg}[2][\pebdeg]%
    {\somenodetrueclausedeg[#1]{\vertexpath{#2}}}
\newcommand{\pathclauseNRdeg}[2][\pebdeg]%
    {\somenodetrueclausedeg[#1]{\vertexpathNR{#2}}}

\newcommand{\blacktruthdegexplicit}[4]%
        {\setdescrcompact
        {{\textstyle \Lor_{#2 = 1}^{#3} {#1}_{#2}}}
        {{#1} \in {#4}}}

\newcommand{\binsubtree}[1]{T^{#1}}

\newcommand{\vertexpath}[1]{{P}^{#1}}
\newcommand{\vertexpathNR}[1]{{P}_{*}^{#1}}

\newcommand{\unrelatedNP}[1]%
        {T \setminus \bigl(\binsubtree{#1} \unionSP \vertexpath{#1} \bigr)}
\newcommand{\unrelatedsmallNP}[1]%
        {T \setminus (\binsubtree{#1} \unionSP \vertexpath{#1} )}

\newcommand{\abovelevelblockerminsizecompact}%
    [2]{L_{\succeq{#1}}\bigl({#2}\bigr)}

\newcommand{\necessaryhidingvert}[2]%
{{#1}{\scriptstyle{\llfloor {#2} \rrfloor}}}

\newcommand{\Klawepropertyprefix}{Limited hiding-cardinality\xspace}

\newcommand{\klawepropacronym}{LHC property\xspace}

\newcommand{\nongenklaweprop}%
{non-generalized \Klawepropertyprefix property\xspace}
\newcommand{\nongenklawepropacronym}%
{non-generalized \klawepropacronym}
\newcommand{\nongenklawepropacronymWithParam}%
{(non-generalized) \klawepropacronym}

\newcommand{\siblingnonreachabiblitypropertynoref}%
{Sibling non-reachability property\xspace}
\newcommand{\Siblingnonreachabiblitypropertynoref}%
{Sibling non-reachability property\xspace}
\newcommand{\siblingnonreachabiblityproperty}%
{\siblingnonreachabiblitypropertynoref~%
\ref{property:sibling-non-reachability-property}\xspace}
\newcommand{\Siblingnonreachabiblityproperty}%
{\Siblingnonreachabiblitypropertynoref~%
\ref{property:sibling-non-reachability-property}\xspace}

\newcommand{\introducetermanmpctext}%
    {a \introduceterm{\mpctext{}}\xspace}

\newcommand{\introducetermamultipebblingtext}%
  {a \introduceterm{\multipebblingtext{}}\xspace}

\newcommand{\blobpebblingtext}{blob-pebbling\xspace}

\newcommand{\multipebblingtext}{\blobpebblingtext}

\newcommand{\mpcostblack}[1]%
        {\formatpebblingprice{cost}_{\mpcblacks}( #1 )}
\newcommand{\mpcostwhite}[1]%
        {\formatpebblingprice{cost}_{\mpcwhites}( #1 )}

\newcommand{\blobpebblingpricecompact}[1]%
    {\formatpebblingprice{Blob-Peb}\bigl(#1\bigr)}

\newcommand{\multipebblingpricecompact}[1]%
    {\formatpebblingprice{Blob-Peb}\bigl(#1\bigr)}

\newcommand{\mpcblacks}{\formatfunctiontomulti{B}}
\newcommand{\mpcwhites}{\formatfunctiontomulti{W}}

\newcommand{\mpscnotcompact}[2]%
        {\big[ {#1} \big] \bigl\langle {#2} \bigr\rangle}

\newcommand{\mpctext}{\blobpebblingtext con\-fig\-u\-ra\-tion\xspace}

\newcommand{\chargeablevertices}[1]%
{\formatfunctiontoset{chargeable}({#1}) }
\newcommand{\chargeableverticescompact}[1]%
{\formatfunctiontoset{chargeable}\bigl({#1}\bigr) }

\newcommand{\blackschargedfor}[1][]%
    {\mpcblacks_{#1}}
\newcommand{\whiteschargedfor}[1][]%
    {\mpcwhites_{#1}^{\hspace{-0.3 pt}\vartriangle}}

\newcommand{\whitesbelowjustblocked}%
    {\mpcwhites_{B}^{\hspace{-0.3 pt}\vartriangle}}
\newcommand{\whitesbelowhidden}%
    {\mpcwhites_{H}^{\hspace{-0.3 pt}\vartriangle}}

\newcommand{\whitestight}%
    {\mpcwhites_{T}^{\hspace{-0.3 pt}\vartriangle}}

\usepackage[shortlabels]{enumitem}
\usepackage{upgreek}

\newcommand{\logstar}{\log^*}

\newcommand{\pconfbefore}{\pconf}
\newcommand{\pconfafter}{\pconf'}

\newcommand{\blowupformatting}[1]{\mathsf{#1}}
\renewcommand{\blowupformatting}[1]{\mathfrak{#1}}
\renewcommand{\blowupformatting}[1]{\mathcal{#1}}

\newcommand{\pathnvert}{\ell}

\newcommand{\pebsp}{s}
\newcommand{\pebsepconst}{K}
\newcommand{\pebshiftconst}{K_p}

\newcommand{\timet}{t}

\newcommand{\pstandard}{standard\xspace}
\newcommand{\Pstandard}{Standard\xspace}

\newcommand{\singlesinkdagtext}{single-sink DAG\xspace}

\newcommand{\graphstd}{G}

\newcommand{\refclaim}[1]{Claim~\ref{#1}}
\newcommand{\refcon}[1]{Construction~\ref{#1}}

\newcommand{\efficiently}{polytime\xspace} %
\renewcommand{\efficiently}{polynomial-time\xspace} %

\renewcommand{\defeq}{\mathrel{\mathop{:}}=}

\renewcommand{\th}{^{\text{\tiny th}}}
\newcommand{\where}{:} %
\renewcommand{\funcdescr}[3]{\ensuremath{ #1 \colon #2 \to #3}}

\renewcommand{\true}{\textsc{True}\xspace}
\renewcommand{\false}{\textsc{False}\xspace}
\renewcommand{\true}{\textsc{true}\xspace}
\renewcommand{\false}{\textsc{false}\xspace}
\newcommand{\tvundef}{\ensuremath{*}\xspace} %

\newcommand{\ib}[1]{\ibleft#1\ibright} 
\newcommand{\ibleft}{\llbracket}
\newcommand{\ibright}{\rrbracket}

\newcommand{\intersect}{\cap}
\newcommand{\disjunion}{\amalg}
\renewcommand{\disjunion}{\disjointunion}

\DeclareMathOperator{\symmdiff}{\triangle}
\newcommand{\pyrg}{g}
\newcommand{\pyrx}{\Delta}

\newcommand{\mold}{M} %
\newcommand{\inmold}{_{\mathrm{in}}} %
\newcommand{\outmold}{_{\mathrm{out}}} %
\newcommand{\anymold}{_{\mathrm{any}}} %

\newcommand{\outmoldproj}{\projectionFunc\outmold}
\newcommand{\anymoldproj}{\projectionFunc\anymold}
\newcommand{\outmoldprojTime}[1]{\outmoldproj(\pconf_{#1}')}
\newcommand{\anymoldprojTime}[1]{\anymoldproj(\pconf_{#1}')}

\newcommand{\projectionFunc}{\formatfunctiontoset{proj}}
\newcommand{\projection}[1]{\projectionFunc(#1)}

\newcommand{\sink}{z} %

\newcommand{\propancnode}{\pancnode}

\newcommand{\DAGNamePath}{P}

\newcommand{\interedges}[1][]{E^{#1}_{\mathrm{inter}}}
\newcommand{\intraedges}[1][]{E^{#1}_{\mathrm{intra}}}

\newcommand{\visitprice}[2][]{\formatpebblingprice{Visit}_{#1}(#2)} 
\newcommand{\surroundprice}[2][]{\formatpebblingprice{Surr}_{#1}(#2)}
\newcommand{\persistentprice}[2][]{\formatpebblingprice{Per}_{#1}(#2)}
\newcommand{\blackprice}[2][]{\formatpebblingprice{Peb}_{#1}(#2)}
\newcommand{\dtprice}[2][]{\formatpebblingprice{DT}_{#1}(#2)}
\newcommand{\rmprice}[2][]{\formatpebblingprice{RM}_{#1}(#2)}
\renewcommand{\visitprice}[2][]{\formatpebblingprice{RPeb}^V_{#1}%
  \ifthenelse{\equal{#1}{}}{\!}{}(#2)} 
\renewcommand{\surroundprice}[2][]{\formatpebblingprice{RPeb}^S_{#1}%
  \ifthenelse{\equal{#1}{}}{\!}{}(#2)} 
\renewcommand{\persistentprice}[2][]{\formatpebblingprice{RPeb}_{#1}(#2)} 
\renewcommand{\blackprice}[2][]{\formatpebblingprice{Peb}_{#1}(#2)}
\renewcommand{\dtprice}[2][]{\formatpebblingprice{DT}_{#1}(#2)}
\renewcommand{\rmprice}[2][]{\formatpebblingprice{RM}_{#1}(#2)}

\newcommand{\Visitprice}[2][]{\formatpebblingprice{RPeb}^V_{#1}%
  \ifthenelse{\equal{#1}{}}{\!}{}\bigl(#2\bigr)} 
\newcommand{\Surroundprice}[2][]{\formatpebblingprice{RPeb}^S_{#1}%
  \ifthenelse{\equal{#1}{}}{\!}{}\bigl(#2\bigr)} 
\newcommand{\Persistentprice}[2][]{\formatpebblingprice{RPeb}_{#1}\bigl(#2\bigr)}

\newcommand{\dtnames}{Dymond--Tompa\xspace}
\newcommand{\dtgametext}{\dtnames game\xspace}

\newcommand{\dtpricetext}{\dtnames price\xspace}
\newcommand{\dtnrounds}{r}

\newcommand{\dtpebbler}{Pebbler\xspace}
\newcommand{\dtchallenger}{Challenger\xspace}
\newcommand{\dtjump}{jump\xspace}
\newcommand{\dtjumps}{jumps\xspace}
\newcommand{\dtjumping}{jumping\xspace}
\newcommand{\dtstay}{stay\xspace}
\newcommand{\dtstays}{stays\xspace}
\newcommand{\dtstaying}{staying\xspace}

\newcommand{\rmnames}{Raz--McKenzie\xspace}
\newcommand{\rmgametext}{\rmnames game\xspace}
\newcommand{\rmpricetext}{\rmnames price\xspace}
\newcommand{\rmnrounds}{r}
\newcommand{\rmpebbler}{Pebbler\xspace}
\newcommand{\rmcolourer}{Colourer\xspace}
\newcommand{\rmred}{red\xspace}
\newcommand{\rmblue}{blue\xspace}

\newcommand{\sP}{\pebbling} %

\newcommand{\pike}{road\xspace}
\newcommand{\Pike}{Road\xspace}
\newcommand{\semiopen}{semiopen\xspace}

\newcommand{\highway}{turnpike\xspace}
\newcommand{\Highway}{Turnpike\xspace}
\newcommand{\highwaynot}{T}

\newcommand{\literal}{\ell}
\newcommand{\variable}{x}
\newcommand{\clause}{C}
\newcommand{\CNF}{\Gamma}
\newcommand{\quantifier}{Q}
\newcommand{\QBF}{\phi}

\newcommand{\partassign}{\rho} %

\newcommand{\gadget}{G} %
\newcommand{\gadgetConst}[1]{\gadget(#1)} %
\newcommand{\condition}{q} %
\newcommand{\bank}{r} %

\newcommand{\Literal}{L} %
\newcommand{\literalConst}[2]{\Literal_{#1}(#2)} %
\newcommand{\canonical}{S} %
\newcommand{\Conj}{\Lambda} %
\newcommand{\conjConst}[3]{\Conj_{#1}(#2,#3)} %
\newcommand{\ConjConst}[3]{\Conj_{#1}\bigl(#2,#3\bigr)} %
\newcommand{\Region}{\check R} %

\newcommand{\perlocked}{p-locked\xspace} %
\newcommand{\vislocked}{v-locked\xspace} %
\newcommand{\srclocked}{s-locked\xspace} %
\newcommand{\perlockedfull}{persistent-locked\xspace} %
\newcommand{\vislockedfull}{visiting-locked\xspace} %
\newcommand{\blowup}{\blowupformatting{R}} %
\newcommand{\blowupConst}[2]{\blowup(#1,#2)} %
\newcommand{\outblowup}{_{\mathrm{out}}} %
\newcommand{\extblowup}{_{\mathrm{ext}}} %
\newcommand{\intblowup}{_{\mathrm{int}}} %
\newcommand{\anyblowup}{_{\mathrm{any}}} %

\newcommand{\block}[1]{^{#1}}
\newcommand{\pebblingBlock}[1]{\pebbling\block{#1}}
\newcommand{\pconfBlock}[1]{\pconf\block{#1}}

\newcommand{\pconfBlowUp}{\pconf'}
\newcommand{\pebblingBlowUp}{\pebbling'}
\newcommand{\pconfTime}[1]{\pconf(#1)}
\newcommand{\pconfBlockTime}[2]{\pconfBlock{#1}(#2)}

\newcommand{\outproj}[2]{\projectionFunc\block{#1}\outblowup(#2)}
\newcommand{\extproj}[2]{\projectionFunc\block{#1}\extblowup(#2)}
\newcommand{\intproj}[2]{\projectionFunc\block{#1}\intblowup(#2)}
\newcommand{\anyproj}[2]{\projectionFunc\block{#1}\anyblowup(#2)}
\newcommand{\perproj}[1]{\projectionFunc^P(#1)}
\newcommand{\perprojText}{p-projection\xspace}
\newcommand{\visproj}[1]{\projectionFunc^V(#1)}
\newcommand{\visprojText}{v-projection\xspace}
\newcommand{\pebClosure}[1]{\formatfunctiontoset{clos}(#1)}

\newcommand{\pconfB}{\pconf'}
\newcommand{\pconfBTime}[1]{\pconfB_{#1}}
\newcommand{\pconfBTimeFunc}[1]{#1}
\newcommand{\pconfBlocksTime}[1]{%
  \{ \pconfBlockTime{v_1}{#1} \}_{v_1 \in V(G_1)}}

\newcommand{\pebblingBlocks}{%
  \{ \pebblingBlock{v_1} \}_{v_1 \in V(G_1)}}

\newcommand{\unsurrounded}[2]{%
  \formatfunctiontoset{unsur}_{#1}(#2)}

\newcommand{\Unsurrounded}[2]{%
  \formatfunctiontoset{unsur}_{#1}\bigl(#2\bigr)}

\newcommand{\RoundInvariant}[3]{%
  #1 \union \allowbreak \Bigl( \VisitPart{#1}{#2}{#3} \Bigr)}

\newcommand{\IntInvariant}[2]{%
  \RoundInvariant{%
    \outproj{#1}{#2}}{%
    \intproj{#1}{#2}}{%
    G_2}%
}
\newcommand{\IntInvariantFunc}[2]{%
  \formatfunctiontoset{I}\block{#1}\bigl(#2\bigr)%
}
\newcommand{\IntInvariantSurroundedFunc}[2]{%
  \IntInvariantFunc{#1}{\pconfBTimeFunc{#2}}%
}

\newcommand{\IntInvariantUnsurroundedFunc}[2]{%
  \IntInvariantFunc{#1}{\pebClosure{\pconfBTimeFunc{#2}}}%
}

\newcommand{\ExtInvariant}[1]{%
  \RoundInvariant{\perproj{#1}}{\visproj{#1}}{G_1}%
}
\newcommand{\extInvariantFunc}[1]{%
  \formatfunctiontoset{Ext}(#1)%
}

\newcommand{\visitPart}[3]{%
  #2 \intersect \allowbreak \unsurrounded{#3}{#1}}
\newcommand{\VisitPart}[3]{%
  #2 \intersect \allowbreak \Unsurrounded{#3}{#1}}

\newcommand{\VisitPartInterior}[2]{%
  \VisitPart{%
    \outproj{#1}{#2}}{%
    \intproj{#1}{#2}}{%
    G_2}}
\newcommand{\VisitPartInteriorTime}[2]{%
  \VisitPartInterior{#1}{\pconfBTimeFunc{#2}}%
}

\newcommand{\visitPartExterior}[1]{%
  \visitPart{\perproj{#1}}{\visproj{#1}}{G_1}%
}

\newcommand{\UnsurroundedOutput}[2]{%
  \Unsurrounded{G_2}{\outproj{#1}{\pconfBTimeFunc{#2}}}}

\newcommand{\toPebble}{T_+}
\newcommand{\toUnpebble}{T_-}
\newcommand{\preToPebble}{D_+}
\newcommand{\preToUnpebble}{D_-}
\newcommand{\prePebbleUnpebble}{D}
\newcommand{\timeBottleneck}{b}
\newcommand{\pconfBottleneck}{\pconf_{\timeBottleneck}}
\newcommand{\pconfBlowUpBottleneck}{\pconf_{\beta}'}

\newcommand{\intmappingBlock}[1]{\formatfunctiontoset{Int}\block{#1}}
\newcommand{\Reasonablepebbling}{Reasonable pebbling\xspace}
\newcommand{\reasonablepebbling}{reasonable pebbling\xspace}

\newcommand{\blowupNaiveConst}[2]{\blowupformatting{N}(#1,#2)} 
\newcommand{\blowupExpensive}{\blowupformatting{T}}
\newcommand{\blowupExpensiveConst}[2]{\blowupformatting{T}(#1,#2)} 
\newcommand{\pebblingBlowUpExpensive}{\pebbling^\blowupExpensive}

\newcommand{\blowupBlack}{\blowupformatting{S}} %
\newcommand{\pconfBlowUpBlack}{\pconf^{\blowupBlack}}
\newcommand{\pebblingBlowUpBlack}{\pebbling^\blowupBlack}
\newcommand{\pebblingchunk}[1]{\pebbling^{[#1]}}

\newcommand{\blowupBlackConst}[2]{\blowupBlack(#1,#2)} %

\newcommand{\pointygraph}[1]{\hat{#1}}

\newcommand{\theauthorSMC}{the first author\xspace}
\newcommand{\TheauthorSMC}{The first author\xspace}

\newcommand{\TheauthorJN}{The third author\xspace}

\ifthenelse{\boolean{conferenceversion}}
{

}
{

}

\ifthenelse{\boolean{conferenceversion}}
{}
{
\newtheoremstyle{metacommenttheoremstyle}%
    {3pt}%
    {3pt}%
    {\sffamily \itshape \scriptsize
    }%
    {}%
    {\bfseries \scshape \footnotesize }%
    {:}%
    { }%
    {}%

\theoremstyle{metacommenttheoremstyle}
}
\newtheorem{jncommentcontainer}{Jakob's comment}
\newtheorem{mlcommentcontainer}{Massimo's comment}
\newtheorem{smccommentcontainer}{Siu Man's comment}
\newtheorem{mvcommentcontainer}{Marc's comment}

\ifthenelse{\isundefined{\DONOTINSERTCOMMENTS}}
{ %
  \usepackage[usenames,dvipsnames]{xcolor}

  \newcommand{\jncomment}[1]%
  {\begin{jncommentcontainer} \textcolor{blue}{#1} \end{jncommentcontainer}}

  \newcommand{\mlcomment}[1]%
  {\begin{mlcommentcontainer} \textcolor{OliveGreen}{#1} \end{mlcommentcontainer}}

  \newcommand{\smccomment}[1]%
  {\begin{smccommentcontainer} \textcolor{magenta}{#1} \end{smccommentcontainer}}

  \newcommand{\mvcomment}[1]%
  {\begin{mvcommentcontainer} \textcolor{orange}{#1} \end{mvcommentcontainer}}

}
{ %
  \newcommand{\jncomment}[1]{}
  \newcommand{\mlcomment}[1]{}
  \newcommand{\smccomment}[1]{}
  \newcommand{\mvcomment}[1]{}
}

\usepackage{tikz}
\usepackage{adjustbox}

\usetikzlibrary{shapes.geometric} %
\usetikzlibrary{positioning} %
\tikzset{%
  Node/.style={draw,circle,thick},
  Arc/.style={draw,->,thick},
  Bank/.style={draw,rectangle},
  Dashed/.style={draw,rectangle,dashed,rounded corners},
  Cra/.style={draw,<-,thick},
  Attachment/.style={draw,-,thick},
  Mold/.style={draw,
    shape border rotate=30,
    regular polygon,
    regular polygon sides=3},
}

\newcommand{\GadgetBankShift}{.7}

\newcommand{\MoldLegendOpt}[6]{%
  \node[Mold,#4={0cm of #2},label={#5:$#3$},#6] (#1-#2) {};
}

\newcommand{\MoldLegend}[3]{%
  \MoldLegendOpt{#1}{#2}{#3}{left}{center}{}
}

\newcommand{\GadgetAttachment}[3]{%
  \node[Node] (#1) at (0,0) {};
  \node[Node] (#2) at (0,1.5) {};
  
  \node[label={below:$#1$}] at (#1) {};
  \node[label={below:$#2$}] at (#2) {};
  \MoldLegend{#1}{#2}{#3}
  \draw[Attachment] (#1) to (#1-#2);
}

\newcommand{\GadgetAttachmentUp}[3]{%
  \node[Node] (#1) at (0,0) {};
  \node[Node] (#2) at (0,1.5) {};
  
  \node[label={right:$#1$}] at (#1) {};
  \node[label={right:$#2$}] at (#2) {};
  \MoldLegendOpt{#1}{#2}{#3}{below}{center}{shape border rotate=0}
  \draw[Attachment] (#1) to (#1-#2);
}

\newcommand{\GadgetLiteralFilled}[5]{%
  \node[Node,fill=#4] (#2') at (0,0) {};
  \node[Node,fill=#5] (#2) at (0,1) {};

  \node[label={#3:$#1'$}] at (#2') {};
  \node[label={#3:$#1$ }] at (#2) {};

  \draw[Arc] (#2') to (#2);
}

\newcommand{\GadgetLiteralWithSourcesFilled}[6]{%
  \GadgetLiteralFilled{#1}{#2}{#4}{#5}{#6}

  \node[Bank] (s#2) at (0,-\GadgetBankShift) {$#3$};
  \draw (s#2.north west) to (#2'.south west);
  \draw (s#2.north east) to (#2'.south east);
}

\newcommand{\GadgetLiteralWithSources}[4]{%
  \GadgetLiteralWithSourcesFilled{#1}{#2}{#3}{#4}{none}{none}
}

\newcommand{\GadgetVariableWithSourcesFilled}[4]{%
  \begin{scope}[xshift=-.5cm]
    \GadgetLiteralWithSourcesFilled{\variable_{#1}}{\variable#1}{#2}{left}{#3}{#4}
  \end{scope}
  \begin{scope}[xshift=.5cm]
    \GadgetLiteralWithSourcesFilled{\bar\variable_{#1}}{b\variable#1}{#2}{right}{#4}{#3}
  \end{scope}
}

\newcommand{\GadgetVariableWithSources}[2]{%
  \GadgetVariableWithSourcesFilled{#1}{#2}{none}{none}
}
\newcommand{\GadgetVariableWithSourcesFalse}[2]{%
  \GadgetVariableWithSourcesFilled{#1}{#2}{black}{none}
}
\newcommand{\GadgetVariableWithSourcesTrue}[2]{%
  \GadgetVariableWithSourcesFilled{#1}{#2}{none}{black}
}

\newcommand{\GadgetThreeOR}[1]{%
  \node[Node] (a#1) at (-1,0) {};
  \node[Node] (b#1) at ( 0,0) {};
  \node[Node] (c#1) at ( 1,0) {};
  \node[Node] (u#1) at (-.5,1) {};
  \node[Node] (v#1) at (0.5,1) {};
  \node[Node] (p#1) at ( 0,2) {};

  \node[label={below:$a_{#1}$}] at (a#1) {};
  \node[label={below:$b_{#1}$}] at (b#1) {};
  \node[label={below:$c_{#1}$}] at (c#1) {};
  \node[label={below:$u_{#1}$}] at (u#1) {};
  \node[label={below:$v_{#1}$}] at (v#1) {};
  \node[label={below:$p_{#1}$}] at (p#1) {};

  \foreach \u / \v in {a#1/u#1, b#1/u#1, b#1/v#1,
    c#1/v#1, u#1/p#1, v#1/p#1} {%
    \draw[Arc] (\u) to (\v);
  }
}

\newcommand{\GadgetClauseLiterals}[2]{%
  \node[Node] (l#11') at (-1,-3) {};
  \node[Node] (l#12') at (0,-3) {};
  \node[Node] (l#13') at (1,-3) {};
  \node[Node] (l#11) at (-1,-2) {};
  \node[Node] (l#12) at (0,-2) {};
  \node[Node] (l#13) at (1,-2) {};

  \node[label={below:$\literal_{#1,1}'$}] at (l#11') {};
  \node[label={below:$\literal_{#1,2}'$}] at (l#12') {};
  \node[label={below:$\literal_{#1,3}'$}] at (l#13') {};
  \node[label={below:$\literal_{#1,1}$}] at (l#11) {};
  \node[label={below:$\literal_{#1,2}$}] at (l#12) {};
  \node[label={below:$\literal_{#1,3}$}] at (l#13) {};

  \foreach \u / \v in {l#11'/l#11, l#12'/l#12, l#13'/l#13} {%
    \draw[Arc] (\u) to (\v);
  }

  \foreach \u / \v in {l#11/a#1, l#12/b#1, l#13/c#1} {%
    \MoldLegendOpt{\u}{\v}{#2}{left}{center}{minimum size=.9cm}
    \draw[Attachment] (\u) to (\u-\v);
  }
}

\newcommand{\GadgetClause}[2]{%
  \GadgetThreeOR{#1}
  \GadgetClauseLiterals{#1}{#2}
}

\newcommand{\GadgetConjunction}[3]{%
  \node[Node] (t1) at (-.7, 0) {};
  \node[Node] (t2) at (0.7, 0) {};
  \node[Node] (d1) at (-.7, 2) {};
  \node[Node] (d2) at (0.7, 2) {};
  \node[Node] (d3) at (  0, 3) {};
  \node[Node] (d4) at (  0, 5) {};
  \node[Node] (e) at (  0, 6) {};

  \node[label={below:$\sink_1$}] at (t1) {};
  \node[label={below:$\sink_2$}] at (t2) {};
  \node[label={below:$d_1$}] at (d1) {};
  \node[label={below:$d_2$}] at (d2) {};
  \node[label={below:$d_3$}] at (d3) {};
  \node[label={below:$d_4$}] at (d4) {};
  \node[label={below:$e$}] at (e) {};

  \foreach \u / \v in {d1/d3, d2/d3, d4/e} {%
    \draw[Arc] (\u) to (\v);
  }

  \foreach \u / \v /\w in {t1/d1/#1, t2/d2/#1-1, d3/d4/#1-2} {%
    \MoldLegendOpt{\u}{\v}{\w}{left}{#2}{#3}
    \draw[Attachment] (\u) to (\u-\v);
  }
}

\newcommand{\GadgetExistentialQuantifier}[2]{%
  \GadgetVariableWithSources{#1}{#2}

  \node[Node] (f#1) at (-2,2) {};
  \node[label={left:$f_{#1}$}] at (f#1) {};

  \node[Node] (g#1) at (0,2) {};
  \node[label={right:$g_{#1}$}] at (g#1) {};

  \node[Node] (q#1) at (-2,3) {};
  \node[label={left:$q_{#1}$}] at (q#1) {};

  \foreach \u / \v in {\variable#1/g#1, b\variable#1/g#1, f#1/q#1, g#1/q#1} {%
    \draw[Arc] (\u) to (\v);
  }
}

\newcommand{\GadgetExistentialQuantifierConnection}[5]{%
  \MoldLegendOpt{q#2}{f#1}{#3}{below}{#4}{shape border rotate=0,#5}
  \draw[Attachment] (q#2) to (q#2-f#1);
}

\newcommand{\GadgetExistentialQuantifierSingle}[6]{%
  \GadgetExistentialQuantifier{#1}{#3}
  \node[Node] (q#2) at (-2,-.5) {};
  \node[label={left:$q_{#2}$}] at (q#2) {};

  \GadgetExistentialQuantifierConnection{#1}{#2}{#4}{#5}{#6}
}

\newcommand{\GadgetUniversalQuantifier}[6]{%
  \GadgetVariableWithSources{#1}{#2}

  \node[Node] (f#1') at (-.5,2) {};
  \node[label={#6:$f_{#1}'$}] at (f#1') {};

  \node[Node] (bf#1') at (0.5,2) {};
  \node[label={#6:$\bar f_{#1}'$}] at (bf#1') {};

  \node[Node] (f#1) at (-.5,4) {};
  \node[label={#6:$f_{#1}$}] at (f#1) {};

  \node[Node] (bf#1) at (0.5,4) {};
  \node[label={#6:$\bar f_{#1}$}] at (bf#1) {};

  \node[Node] (g#1) at (-2.5,3) {};
  \node[label={left:$g_{#1}$}] at (g#1) {};

  \node[Node] (bg#1) at (-1.5,3) {};
  \node[label={left:$\bar g_{#1}$}] at (bg#1) {};

  \node[Node] (h#1) at (-2.5,4) {};
  \node[label={left:$h_{#1}$}] at (h#1) {};

  \node[Node] (bh#1) at (-1.5,4) {};
  \node[label={left:$\bar h_{#1}$}] at (bh#1) {};

  \node[Node] (q#1) at (-2,5) {};
  \node[label={left:$q_{#1}$}] at (q#1) {};

  \foreach \u / \v in {\variable#1/f#1', b\variable#1'/f#1', g#1/h#1,
    b\variable#1/bf#1', \variable#1'/bf#1', bg#1/bh#1,
    h#1/q#1, bh#1/q#1} {%
    \draw[Arc] (\u) to (\v);
  }

  \foreach \u / \v in {f#1/h#1, bf#1/bh#1} {%
    \draw[Arc,bend right=40] (\u) to (\v);
  }

  \foreach \u / \v in {f#1'/f#1, bf#1'/bf#1} {%
    \MoldLegendOpt{\u}{\v}{#3}{below}{#4}{shape border rotate=0,#5}
    \draw[Attachment] (\u) to (\u-\v);
  }
}

\newcommand{\GadgetUniversalQuantifierStd}[5]{
  \GadgetUniversalQuantifier{#1}{#2}{#3}{#4}{#5}{right}
}

\newcommand{\GadgetUniversalQuantifierConnection}[5]{%
  \MoldLegendOpt{q#2}{g#1}{#3}{below}{#4}{shape border rotate=0,#5}
  \draw[Attachment] (q#2) to (q#2-g#1);

  \MoldLegendOpt{q#2}{bg#1}{#3}{below}{#4}{shape border rotate=0,#5}
  \draw[Attachment] (q#2) to (q#2-bg#1);
}

\newcommand{\GadgetUniversalQuantifierSingle}[7]{%
  \GadgetUniversalQuantifierStd{#1}{#3}{#4}{#6}{#7}
  \node[Node] (q#2) at (-2,-.5) {};
  \node[label={left:$q_{#2}$}] at (q#2) {};

  \GadgetUniversalQuantifierConnection{#1}{#2}{#5}{#6}{#7}
}

\newcommand{\GadgetClauseConnection}[8]{%
  \foreach \u / \v in {#5\variable#2/a#1, #6\variable#3/b#1, #7\variable#4/c#1} {%
    \MoldLegendOpt{\u}{\v}{#8}{left}{center}{minimum size=.9cm}
    \draw[Attachment] (\u) to [out=120,in=-90,out looseness=1.5,in looseness=0.5] (\u-\v.west);
  }
}

\newcommand{\GadgetConjunctionX}[6]{%
  \node[Node] (d#11) at (-.7, 2.5) {};
  \node[Node] (d#12) at (0.7, 2.5) {};
  \node[Node] (d#13) at (  0, 3.5) {};
  \node[Node] (d#14) at (  0, 5) {};
  \node[Node] (e#1) at (  0, 6) {};

  \node[label={above right:$d_{#1,1}$}] at (d#11) {};
  \node[label={above right:$d_{#1,2}$}] at (d#12) {};
  \node[label={above right:$d_{#1,3}$}] at (d#13) {};
  \node[label={above right:$d_{#1,4}$}] at (d#14) {};
  \node[label={above right:$e_#1$}] at (e#1) {};

  \foreach \u / \v in {d#11/d#13, d#12/d#13, d#14/e#1} {%
    \draw[Arc] (\u) to (\v);
  }

  \foreach \u / \v /\w in {e#2/d#11/#4, p#3/d#12/#4-1, d#13/d#14/#4-2} {%
    \pgfmathparse{\w}
    \MoldLegendOpt{\u}{\v}{\pgfmathprintnumber{\pgfmathresult}}{above left}{#5}{shape border rotate=105,#6}
    \draw[Attachment] (\u.south east) to [out=90,in=-90,looseness=1.5] (\u-\v.north west);
  }
}

\newcommand{\ExampleQBF}{
  \path[use as bounding box] (-4, -3) rectangle (12, 16);
  \begin{scope}
    \GadgetUniversalQuantifier{1}{15}{12}{center}{minimum size=.9cm}{left}
  \end{scope}
  \begin{scope}[yshift=6cm]
    \GadgetExistentialQuantifier{2}{19}{center}{minimum size=.9cm}
    \GadgetExistentialQuantifierConnection{2}{1}{16}{center}{minimum size=.9cm}
  \end{scope}
  \begin{scope}[yshift=10cm]
    \GadgetUniversalQuantifier{3}{23}{20}{center}{minimum size=.9cm}{left}
    \GadgetUniversalQuantifierConnection{3}{2}{19}{center}{minimum size=.9cm}
  \end{scope}
  \begin{scope}[xshift=5cm,yshift=12cm,rotate=-90]
    \GadgetThreeOR{1}
    \GadgetClauseConnection{1}{3}{2}{1}{}{}{}{2}
  \end{scope}
  \begin{scope}[xshift=5cm,yshift=8cm,rotate=-90]
    \GadgetThreeOR{2}
    \GadgetClauseConnection{2}{3}{2}{1}{}{}{b}{4}
  \end{scope}
  \begin{scope}[xshift=5cm,yshift=4cm,rotate=-90]
    \GadgetThreeOR{3}
    \GadgetClauseConnection{3}{3}{2}{1}{b}{b}{b}{6}
  \end{scope}
  \begin{scope}[xshift=8cm,yshift=13cm]
    \node[Bank] (s0) at (0,\GadgetBankShift) {7};
    \node[Node] (e0) at (0,0) {};
    \draw (s0.south west) to (e0.north west);
    \draw (s0.south east) to (e0.north east);
  \end{scope}
  \begin{scope}[xshift=7cm,yshift=13cm,rotate=-135]
    \GadgetConjunctionX{1}{0}{1}{2}{center}{minimum size=.9cm}
  \end{scope}
  \begin{scope}[xshift=7cm,yshift=9cm,rotate=-135]
    \GadgetConjunctionX{2}{1}{2}{4}{center}{minimum size=.9cm}
  \end{scope}
  \begin{scope}[xshift=7cm,yshift=5cm,rotate=-135]
    \GadgetConjunctionX{3}{2}{3}{6}{center}{minimum size=.9cm}
  \end{scope}
  \MoldLegendOpt{e3}{g1}{11}{below}{center}{shape border rotate=0, minimum size=.9cm};
  \draw[Attachment] (e3) to [out=180,in=-90] (-2,0) to [out=90,in=-90] (e3-g1);
  \MoldLegendOpt{e3}{bg1}{11}{below}{center}{shape border rotate=0, minimum size=.9cm};
  \draw[Attachment] (e3) to [out=180,in=-90] (-2,0) to [out=90,in=-90] (e3-bg1);
}

\newcommand{\ExampleMolding}{
  \node[Node] (i1) at (-1,0) {};
  \node[Node] (o1) at (0,0) {};
  \node[Node] (i2) at (-1,1) {};
  \node[Node] (o2) at (0,1) {};
  \node[Node] (i3) at (1,0.5) {};
  \node[Node] (o3) at (2,0.5) {};
  \node[Node] (s) at (-2,2) {};
  
  \node[label={below:$s$}] at (s) {};

  \foreach \u / \v in {i1/o1, i2/o2, i3/o3, o1/i3, o2/i3} {%
    \draw[Arc] (\u) to (\v);
  }
  \foreach \u / \v in {s/o1, s/o2, s/o3} {%
    \draw[Arc] (\u) to [bend left=30](\v);
  }
}

\newcommand{\BlowUpCell}[2]{
  \node[Node] (e#1#2) at (-.5,0) {};
  \node[Node] (i#1#2) at (.5,0) {};
  \node[Node] (o#1#2) at (0,1) {};
  \foreach \u / \v in {e#1#2/o#1#2, i#1#2/o#1#2} {%
    \draw[Arc] (\u) to (\v);
  }
}

\newcommand{\BlowUpBlock}[1]{
  \begin{scope}[yshift=-2cm]
    \BlowUpCell{#1}{1}
  \end{scope}
  \begin{scope}[xshift=-1.5cm]
    \BlowUpCell{#1}{2}
  \end{scope}
  \begin{scope}[xshift=1.5cm]
    \BlowUpCell{#1}{3}
  \end{scope}
  \begin{scope}[yshift=2cm]
    \BlowUpCell{#1}{4}
  \end{scope}
  \foreach \u / \v in {1/2, 1/3, 2/4, 3/4} {%
    \draw[Arc] (o#1\u) to [bend right=30] (i#1\v);
  }
}

\newcommand{\BlowUpBadBlock}[1]{
  \node[Node] (w#1) at (-.7,0) {};
  \node[Node] (e#1) at (.7,0) {};
  \node[Node] (n#1) at (0,1) {};
  \node[Node] (s#1) at (0,-1) {};
  \foreach \u / \v in {s/w, s/e, w/n, e/n} {%
    \draw[Arc] (\u#1) to [bend right=15] (\v#1);
  }
}
\newcommand{\BlowUpBad}{
\begin{scope}[xshift=-1.5cm]
  \BlowUpBadBlock{1}
\end{scope}
\begin{scope}[xshift=1.5cm]
  \BlowUpBadBlock{2}
\end{scope}
\begin{scope}[yshift=3cm]
  \BlowUpBadBlock{3}
\end{scope}
  \foreach \u / \v in {1/3, 2/3} {
    \draw[Arc] (n\u) to [bend left=30] (s\v);
  }
}

\newcommand{\BlowUpExpensive}{
\begin{scope}[xshift=-1.5cm]
  \BlowUpBadBlock{4}
\end{scope}
\begin{scope}[xshift=1.5cm]
  \BlowUpBadBlock{5}
\end{scope}
\begin{scope}[yshift=3cm]
  \BlowUpBadBlock{6}
\end{scope}
  \draw[Arc] (n4) to [bend left=30] (s6);
  \draw[Arc] (n4) to [bend left=30] (w6);
  \draw[Arc] (n4) to [bend left=30] (e6);
  \draw[Arc] (n4) to [bend left=50] (n6);
  \draw[Arc] (n5) to [bend left=30] (s6);
  \draw[Arc] (n5) to [bend left=50] (w6);
  \draw[Arc] (n5) to [bend left=30] (e6);
  \draw[Arc] (n5) to [bend left=30] (n6);
}

\newcommand{\BlowUpRhombus}{
  \BlowUpBadBlock{0}
}

\newcommand{\BlowUpPyr}{
  \node[Node] (a) at (-.7,0) {};
  \node[Node] (b) at (.7,0) {};
  \node[Node] (c) at (0,1) {};
  \foreach \u in {a,b} {
    \draw[Arc] (\u) to [bend left=15] (c);
  }
}

\newcommand{\BlowUpReversible}{
\begin{scope}[xshift=-3cm]
  \BlowUpBlock{1}
\end{scope}
\begin{scope}[xshift=3cm]
  \BlowUpBlock{2}
\end{scope}
\begin{scope}[yshift=6cm]
  \BlowUpBlock{3}
\end{scope}
  \draw[Arc] (o14) to [bend left=30] (e31);
  \draw[Arc] (o14) to [bend left=30] (e32);
  \draw[Arc] (o14) to [bend left=30] (e33);
  \draw[Arc] (o14) to [bend left=50] (e34);
  \draw[Arc] (o24) to [bend left=30] (e31);
  \draw[Arc] (o24) to [bend left=50] (e32);
  \draw[Arc] (o24) to [bend left=30] (e33);
  \draw[Arc] (o24) to [bend left=30] (e34);
}

\newcommand{\ExampleBlowUp}{
\begin{scope}[xshift=0cm,yshift=7cm]
  \BlowUpRhombus
  \node at (-1.5,0.5) {$G_{2}=$};
\end{scope}
\begin{scope}[xshift=-5.75cm,yshift=7cm]
  \BlowUpPyr
  \node at (-1.5,0.5) {$G_{1}=$};
\end{scope}
\begin{scope}[xshift=-5.75cm]
  \BlowUpBad
  \node at (-0.75,4.75) {$\blowupNaiveConst{G_1}{G_2}=$};
\end{scope}
\begin{scope}[xshift=0cm]
  \BlowUpExpensive
  \node at (-0.75,4.75) {$\blowupExpensiveConst{G_1}{G_2}=$};
\end{scope}
\begin{scope}[xshift=8.5cm]
  \BlowUpReversible
  \node at (-4,8) {$\blowupConst{G_1}{G_2}=$};
\end{scope}
}

\newcommand{\ExampleBlowUpReversible}{
\begin{scope}
  \BlowUpRhombus
  \node at (-2,0) {$G_{2}=$};
\end{scope}
\begin{scope}[yshift=6cm]
  \BlowUpPyr
  \node at (-2,0) {$G_{1}=$};
\end{scope}
\begin{scope}[xshift=10cm]
  \BlowUpReversible
  \node at (-4,8) {$\blowupConst{G_1}{G_2}=$};
\end{scope}
}

\newcommand{\BlowUpBlackBlock}[1]{
  \foreach \i in {0,1,2,3} {
    \node[Node] (p{#1}\i) at (0,\i) {};
  }
  \foreach \i in {1,2,3} {
    \pgfmathtruncatemacro{\prevlayer}{\i-1}
    \node[Node] (b{#1}\i) at (-1,\prevlayer) {};
    \draw[Arc] (b{#1}\i) to (p{#1}\i);  
    \draw[Arc] (p{#1}\prevlayer) to (p{#1}\i);  
  }
  \node[Node] (g{#1}1) at (0,-2) {};
  \node[Node] (g{#1}2) at (-1,-1) {};
  \node[Node] (g{#1}3) at (1,-1) {};
  \draw[Arc] (g{#1}1) to [bend right=15] (g{#1}2);
  \draw[Arc] (g{#1}1) to [bend right=15] (g{#1}3);
  \draw[Arc] (g{#1}2) to [bend right=15] (p{#1}0);
  \draw[Arc] (g{#1}3) to [bend right=15] (p{#1}0);
}

\newcommand{\BlowUpBlack}{
\begin{scope}[xshift=-3cm]
  \BlowUpBlackBlock{1}
\end{scope}
\begin{scope}[xshift=3cm]
  \BlowUpBlackBlock{2}
\end{scope}
\begin{scope}[yshift=6cm]
  \BlowUpBlackBlock{3}
\end{scope}
  \foreach \u / \v in {1/3, 2/3} {%
    \foreach \z in {1,2,3} {%
      \draw[Arc] (p{\u}3) to [bend left=27] (b{\v}\z);
    }
  }
}

\newcommand{\BlowUpBlackExample}{
\begin{scope}
  \BlowUpRhombus
  \node at (-2,0) {$G_{2}=$};
\end{scope}
\begin{scope}[yshift=6cm]
  \BlowUpPyr
  \node at (-2,0) {$G_{1}=$};
\end{scope}
\begin{scope}[xshift=10cm]
  \BlowUpBlack
  \node at (-4,8) {$\blowupBlackConst{G_1}{G_2}=$};
\end{scope}
}

\newcommand{\JoinLayers}[2]{
  \pgfmathtruncatemacro{\nextlayer}{#1+1}
  \foreach \column in {#2} {
    \pgfmathtruncatemacro{\nextcol}{Mod(\column,3)+1}
    \draw[Arc] (x#1\column) to (x\nextlayer\column);
    \draw[Arc] (x#1\nextcol) to (x\nextlayer\column);
  }
}
\newcommand{\ExamplePike}{
  \foreach \layer in {1,2,3,4,5,6,7,8} {
    \foreach \column in {1,2,3} {
      \node[Node] (x\layer\column) at (\column,\layer) {};
    }
  }
  \node[Node] (x91) at (1,9) {};
  \node[Node] (x92) at (2,9) {};
  \node[Node] (x101) at (1,10) {};
  \foreach \layer in {1,2,3,4,5,6,7} {
    \JoinLayers{\layer}{1,2,3}
  }
  \JoinLayers{8}{1,2}
  \JoinLayers{9}{1}
}

\newcommand{\JoinLayersUnbounded}[3]{
  \pgfmathtruncatemacro{\nextlayer}{#1+1}
  \foreach \column in {#2} {
    \foreach \nextcol in {#3} {
      \draw[Arc] (x#1\column) to (x\nextlayer\nextcol);
    }
  }
}
\newcommand{\ExamplePikeUnbounded}{
  \foreach \layer in {1,2,3,4,5,6,7,8,9} {
    \foreach \column in {1,2,3} {
      \node[Node] (x\layer\column) at (\column,\layer) {};
    }
  }
  \node[Node] (x101) at (2,10) {};
  \foreach \layer in {1,2,3,4,5,6,7,8} {
    \JoinLayersUnbounded{\layer}{1,2,3}{1,2,3}
  }
  \JoinLayersUnbounded{9}{1,2,3}{1}
}

\newcommand{\christmaslayer}[4]{%
  \begin{scope}[yshift=-#1cm]
     \draw (0.577*#2,0) -- (-0.577*#2,0) -- (0,#2) -- cycle;
     \node at (0,0.25*#2) {$G_{#3}$};
     \node at (3.5,0.25*#2) {Layer $#4$};
   \end{scope}
}
\newcommand{\ExampleChristmasTree}{
\christmaslayer011\bank
\draw[thin] (0,-0.5) .. controls (0.5,-0.25) .. (0,0);
\draw[thin] (0,-0.5) .. controls (0.5,-0.25) .. (0.25,0);
\draw[thin] (0,-0.5) .. controls (0.5,-0.25) .. (-0.25,0);
\christmaslayer{1.5}11{\bank-1}
\draw[thin] (0,-2) .. controls (-1.5,-1.25) .. (0,0);
\draw[thin] (0,-2) .. controls (-1.5,-1.25) .. (0.25,0);
\draw[thin] (0,-2) .. controls (-1.5,-1.25) .. (-0.25,0);
\draw[thin] (0,-2) .. controls (0.5,-1.75) .. (0,-1.5);
\draw[thin] (0,-2) .. controls (0.5,-1.75) .. (0.25,-1.5);
\draw[thin] (0,-2) .. controls (0.5,-1.75) .. (-0.25,-1.5);
\christmaslayer{3.5}{1.5}2{\bank-2}
\node at (0,-4) {$\vdots$};
\christmaslayer{7.5}{2.5}{\bank-1}1
}

\numberwithin{equation}{section}

\begin{document}

\title{Hardness of Approximation in PSPACE and \\
  Separation Results for Pebble Games}

\author[ ]{Siu Man Chan\!}
\author{Massimo Lauria}
\affil{Sapienza --- Università di Roma, Italy}
\author{Jakob Nordstr\"{o}m}
\affil{University of Copenhagen, Denmark, and Lund University, Sweden}
\author{Marc Vinyals}
\affil{University of Auckland, New Zealand}

\title{%
  Hardness of Approximation in PSPACE and \\
  Separation Results for Pebble Games%
  \thanks{This is the full-length version of the paper with the same
    title that appeared in
    \emph{Proceedings of the 56th IEEE Symposium on Foundations of
                  Computer Science (FOCS~'15)}.}}

\date{\today}

\maketitle

\thispagestyle{empty}

\pagestyle{fancy}
\fancyhead{}
\fancyfoot{}
\renewcommand{\headrulewidth}{0pt}
\renewcommand{\footrulewidth}{0pt}

\fancyhead[CE]{\slshape 
  HARDNESS OF APPROXIMATION IN PSPACE AND SEPARATION FOR PEBBLE GAMES}
\fancyhead[CO]{\slshape \nouppercase{\leftmark}}
\fancyfoot[C]{\thepage}

\setlength{\headheight}{13.6pt}

\begin{abstract}
  We consider the pebble game on DAGs with bounded fan-in introduced
  in [Paterson and Hewitt~'70] and the reversible version of this game
  in [Bennett~'89], and study the question of how hard it is to decide
  exactly or approximately the number of pebbles needed for a
  given DAG in these games.

  We prove that the problem of deciding whether $s$~pebbles suffice to
  reversibly pebble a DAG~$G$ is \PSPACE-complete, as was previously
  shown for the \pstandard pebble game in [Gilbert, Lengauer and Tarjan~'80].  
  Via two different graph product constructions we then strengthen
  these results to establish that both \pstandard and reversible
  pebbling space are \PSPACE-hard to approximate to within any
  additive constant.  To the best of our knowledge, these are the
  first hardness of approximation results for pebble games in an
  unrestricted setting (even for polynomial time). Also, since [Chan~'13]
  proved that reversible pebbling is equivalent to the games in
  [Dymond and Tompa~'85] and [Raz and McKenzie~'99], our results apply
  to the Dymond--Tompa and Raz--McKenzie games as well, and from the
  same paper it follows that resolution depth is \PSPACE-hard to
  determine up to any additive constant.

  We also obtain a multiplicative logarithmic separation between
  reversible and standard pebbling space. This improves on the
  additive logarithmic separation previously known and could plausibly
  be tight, although we are not able to prove this.

  We leave as an interesting open problem whether our additive
  hardness of approximation result could be strengthened to a
  multiplicative bound if the computational resources are decreased
  from polynomial space to the more common setting of polynomial time.
\end{abstract}

\section{Introduction}
\label{sec:intro}

In the \introduceterm{pebble game} first studied by Paterson and
Hewitt~\cite{PH70Comparative}, one starts with an empty directed
acyclic graph (DAG)~$G$ with bounded fan-in (and which in this paper
in addition will always have a single sink) and places pebbles on the
vertices according to the following rules:
\begin{itemize}
\item
  If all (immediate) predecessors of an empty vertex~$v$ contain
  pebbles, a pebble may be placed on~$v$.
  
\item

  A pebble may be removed from any vertex at any time.
\end{itemize}
The goal is to get a pebble on the sink vertex of~$G$ with all other
vertices being empty, and to do so while minimizing the total number
of pebbles on~$G$ at any given time (the 
\introduceterm{pebbling price} 
of~$G$).
This game models computations with execution 
independent of the actual 
input. A pebble 
on a vertex indicates that the corresponding value is currently kept
in memory and the objective is to perform the computation with the
minimum amount of memory.

The pebble game has been used to study
flowcharts and recursive schemata~\cite{PH70Comparative},
register allocation~\cite{Sethi75CompleteRegisterAllocation},
time and space as Turing-machine
resources~\cite{Cook74ObservationTimeStorageTradeOff,HPV77TimeVsSpace}, 
and algorithmic time-space trade-offs
\cite{Chandra73Efficient, 
  SS77SpaceTimeFFT, 
  SS79SpaceTimeOblivious,
  SS83SpaceTimeLinear,
  Tompa78TimeSpaceComputing}.
In the last 10--15 years, there has been a renewed interest in pebbling in
the context of proof complexity as discussed in the 
survey~\cite{Nordstrom13SurveyLMCS}
(although in this context one is often interested also in the slightly
more general \introduceterm{black-white pebble game} introduced
in~\cite{CS76Storage}), and pebbling has also turned out to be useful for
applications in cryptography~\cite{DNW05Pebbling,AS15HighParallel}.   
An excellent overview of pebbling up to ca.\ 1980 is given 
in~\cite{Pippenger80Pebbling} and some more recent developments are
covered in the upcoming survey~\cite{Nordstrom09PebblingSurveyFTTCS}.

Bennett~\cite{Bennett89TimeSpaceReversible} introduced the
\introduceterm{reversible pebble game} as part of a broader
program~\cite{Bennett73LogicalReversibility} to investigate
possibilities to eliminate (or significantly reduce) energy
dissipation in logical computation.  Another reason reversible
computation is of interest is that observation-free quantum
computation is inherently reversible.  In the reversible pebble game,
the moves performed in reverse order should also constitute a legal
pebbling, which means that the rules for pebble placement and removal
become symmetric as follows:
\begin{itemize}

\item
  If all predecessors of an
  empty vertex~$v$   contain pebbles,
  a pebble may  be placed on~$v$. 

\item
  If all predecessors of a pebbled vertex~$v$ 
  contain pebbles, 
  the pebble on~$v$ may be removed.
  
\end{itemize}
Reversible pebblings 
of DAGs 
have been studied in
\cite{LV96Reversibility,Kralovic04TimeSpaceReversible}
and have been employed to shed light on 
time-space trade-offs in
reversible simulation of irreversible computation in
\cite{LTV98Reversible,LMT00Reversible,Williams00Reversible,BTV01Reversible}.  
In a different line of work
Potechin~\cite{Potechin10Monotone} 
implicitly used the reversible pebble game for proving lower bounds on monotone
space complexity, with the connection made
explicit in the follow-up works \cite{CP14Monotone,FPRC13Average}.

Another pebble game on DAGs that will be of interest in this paper is the 
\introduceterm{\dtgametext{}}~\cite{DT85Speedups}
played on a DAG~$G$ by a
\introduceterm{\dtpebbler{}}
and a 
\introduceterm{\dtchallenger{}}.
This game is played in rounds, with both players starting at the sink
in the first round.
In the following rounds, \dtpebbler places a pebble on some vertex
of~$G$ after which 
\dtchallenger either stays at the current vertex or moves to the newly
pebbled vertex. This repeats until at the end of a round \dtchallenger
is standing on a vertex with all (immediate) predecessors pebbled (or
on a source, in which case the condition vacuously holds),
at which point the game ends.
Intuitively, \dtchallenger is challenging \dtpebbler to ``catch me if
you can'' and wants to play for as many rounds as possible, whereas
\dtpebbler wants to ``surround'' \dtchallenger as quickly as possible.
The
\introduceterm{\dtpricetext{}}
of~$G$
is the smallest number~$\dtnrounds$
such that \dtpebbler can always finish the game in at
most~$\dtnrounds$ rounds.
The \dtgametext has been used to establish that for parallel time a
speed-up by a logarithmic factor is always possible~\cite{DT85Speedups},
and  in~\cite{VT89NewPebbleGame} it was shown that 
a slightly modified variant of this game exactly characterizes
parallelism in complexity classes like $\complclassformat{AC}^i$,
$\NC$, and~$\Pclass$,
and can be used to re-derive, for instance, Savitch's theorem.
Furthermore, collapses or separations of these classes can in
principle be recast (or discovered) as bounds on \dtpricetext{}.
Interestingly, this characterization of parallelism 
extends to proof complexity as well as discussed
in~\cite{Chan13JustAPebble}. 

A final game with pebbles that we want to just mention 
without going into any details
is the 
\introduceterm{\rmgametext{}} introduced in~\cite{RM99Separation}
to study the depth complexity of decision
trees solving search problems.
The reason for bringing up the
\dtnames and \rmgametext{}s is that it was shown 
in~\cite{Chan13JustAPebble}
that both games are actually equivalent to the reversible
pebble game. Hence, any bounds derived for the reversible pebble
game also hold for
\dtpricetext and \rmpricetext.

The main focus of this paper is to study how hard it is to decide
exactly or approximately the pebbling
price of a DAG. For the \pstandard pebble game
Gilbert \etal~\cite{GLT80PebblingProblemComplete}
showed that given a DAG~$G$ and a positive integer~$\pebsp$ it is
\PSPACE-complete to determine whether space~$\pebsp$ 
is sufficient to pebble~$G$ or not.
It would seem natural to suspect that reversible pebbling price should
be \PSPACE-complete as well, but the construction
in~\cite{GLT80PebblingProblemComplete} cannot be used to show this.

Given that pebbling price is hard to determine exactly, an even more
interesting question is if anything can be said regarding the hardness
of approximating pebbling price. Although this seems like a very
natural and appealing question, apparently next to nothing has been
known about this.

Wu \etal~\cite{WAPL14Inapproximability}
showed that ``one-shot'' \pstandard pebbling price is hard to
approximate to within any multiplicative constant assuming the
so-called
Small Set Expansion (SSE) hypothesis.
In a one-shot pebbling one is only allowed to pebble each vertex once,
however, and this is a major restriction since the complexity now drops
from \PSPACE-complete to 
\mbox{\NP-complete} \cite{Sethi75CompleteRegisterAllocation}.
Note that containment in \NP is easy to see since any one-shot
pebbling can be described concisely just by listing the order in which
the vertices should be pebbled (and it is easy to compute when a pebble is
no longer needed and can be removed). In contrast, in the general case
pebbling strategies that are optimal with respect to space 
can sometimes provably
require exponential time.

One can also go in the other direction and study more general pebble
games, such as the AND/OR pebble game introduced by
Lingas~\cite{Lingas78PSPACE} in one of the works leading up
to~\cite{GLT80PebblingProblemComplete}. 
Here every vertex is labelled AND or OR. For AND-vertices we have
the usual pebbling rule, but for
OR-vertices it is sufficient to just have one pebble on some
predecessor in order to be allowed to pebble the vertex.
This game has a relatively straightforward reduction from hitting
set~\cite{FNPW10Unpublished}, which shows that it is hard to
approximate to within a logarithmic factor, but the reduction
crucially depends on the OR-nodes.

We remark that hardness of approximation in \PSPACE for other problems
has been studied in~\cite{CFLS95Probabilistically}, but those
techniques seem hard to adapt to pebble games since the reduction from
QBF to pebbling is inherently unable to preserve gaps. 

Another problem that we study 
in the current paper
is the relation between \pstandard
pebbling price and reversible pebbling price.  Clearly, the space
needed to reversibly pebble a graph is at least the space required in the
\pstandard pebble game. It is also not hard to see that there are
graphs that require strictly more pebbles in a reversible setting: for
a directed path on $n$~vertices only $2$~pebbles are needed in the
\pstandard game, while it is relatively straightforward to show that
the reversible pebbling space is
$\bigtheta{\log n}$~\cite{Bennett89TimeSpaceReversible,LV96Reversibility}.
However, for ``classic'' graphs studied in the pebbling literature,
such as binary trees, pyramids, certain superconcentrators, and the
worst-case graphs in~\cite{PTC76SpaceBounds},
the reversible and \pstandard pebbling prices coincide asymptotically,
and are sometimes markedly closer than an additive logarithm apart.

This raises the question whether reversible and \pstandard pebbling
can be asymptotically separated with respect to space.  It might be
worth pointing out in this context that for Turing machines
it was proven in~\cite{LMT00Reversible} that any computation can be
simulated reversibly in exactly the same space. 
In the more restricted pebbling model,  it was shown
in~\cite{Kralovic04TimeSpaceReversible} 
that if the \pstandard pebbling price of a DAG~$G$ on $n$~vertices
is~$\pebsp$, then $G$ can be reversibly pebbled with at most
$\pebsp^2 \log n$~pebbles. 
Thus, if there is not only an additive but also a multiplicative
separation between \pstandard and reversible  pebbling price, such a
separation cannot be too large.

\subsection{Our Results}

We obtain the following results:
\begin{enumerate}
\item 
  We establish an asymptotic separation between \pstandard and
  reversible pebbling by exhibiting
  families of graphs
  $\set{G_n}_{n=1}^{\infty}$ of size~$\bigtheta{n}$
  with a single sink and fan-in~$2$ 
  which have standard pebbling price~$\pebsp(n)$
  and reversible pebbling price $\bigomega{\pebsp(n) \log n}$. 
  This construction works for
  any $\pebsp(n) = \Bigoh{n^{1/2 - \epsilon}}$
  with $\epsilon > 0$ constant,
  where the constant hidden in the asymptotic notation in the lower
  bound has a (mild)   dependence on~$\epsilon$. 

\item 
  \label{item:pspace-hardness-result}
  We prove that determining reversible pebbling price
  is \PSPACE-complete. That is, 
  given a single-sink DAG~$G$ of fan-in~$2$ and a parameter~$\pebsp$, it
  is \PSPACE-complete to decide whether $G$ can be reversibly
  pebbled in space~$\pebsp$ or not.

\item 
  Finally, we present two different graph products
  (for \pstandard and reversible pebbling, respectively)
  that take DAGs $G_i$ of size~$n_i$
  with pebbling price~$\pebsp_i$ for $i=1,2$
  and yield  a DAG
  of size~$\Bigoh{(n_1 + n_2)^2}$
  with pebbling price
  $s_1 + s_2 + \pebshiftconst$
  (for $\pebshiftconst = \pm 1$ depending on the flavour of the pebble
  game).
  Combining these graph products with the \PSPACE-completeness results
  for \pstandard pebbling in~\cite{GLT80PebblingProblemComplete} and
  reversible pebbling in  \refitem{item:pspace-hardness-result},
  we deduce that
  for any fixed~$\pebsepconst$ 
  the promise problem of deciding for a DAG $G$ (with a single sink
  and fan-in~$2$) whether it can be pebbled in space~$\pebsp$ or
  requires space~$\pebsp + \pebsepconst$ is \PSPACE-hard in both the
  \pstandard and the reversible pebble game.
\end{enumerate}

We need to provide more formal definitions before going into a
detailed discussion of techniques, but want to stress right away that
a key feature of the above results is the bounded fan-in
condition. This is the standard setting for pebble games in the
literature and is also crucial in most of the applications mentioned
above.  Without this constraint it would be much easier, but also much
less interesting, to prove our
results.\footnote{The reason to emphasize this is that for
  unbounded fan-in \theauthorSMC proved a \PSPACE-completeness result
  for reversible pebbling in~\cite{Chan13Thesis}, but this
  result uses simpler constructions and techniques that do not
  transfer to the bounded fan-in setting. Another, somewhat related,
  example is that deciding space in the black-white pebble game
  has also been shown to be \PSPACE-complete for unbounded indegree
  in~\cite{HP10PspaceCompleteness}, but there the unbounded
  fan-in can be used to eliminate the white pebbles completely, 
  and again the techniques fail to transfer to the bounded indegree
  case.} 

Another aspect worth pointing out is that although the reversible
pebble game is weaker than the 
\pstandard pebble game, it is technically much more challenging to
analyze. The reason for this is that a \pstandard pebbling will always
progress in a ``forward sweep'' through the graph in topological order,
and so one can often assume without loss of generality that once one
has pebbled through some subgraph the pebbling will never touch this
subgraph again.  For a reversible pebbling this is not so,
since any pebble placed on any descendant of vertices in the subgraph
will also have to be removed at some later time, and this has to be
done in reverse topological order. Therefore, in any reversible
pebbling there will be alternating phases of ``forward sweeps''  and
``reverse sweeps,'' and these phases can also be interleaved at
various levels. For this reason, controlling the progress of a reversible
pebbling is substantially more complicated.
Despite the additional technical
difficulties,
however, we consider the reversible pebble game to be at least as
interesting to study as the \pstandard and black-white pebble games
in view of its tight connection with parallelism in circuit and proof
complexity as described in~\cite{Chan13JustAPebble}.

\subsection{Follow-up Work}

Our hardness of approximation result for the standard pebble game was improved
by Demaine and Liu~\cite{DL17Inapproximability}, who proved that it is
$\pspace$-hard to approximate the standard space of a graph of size
$n$ within an additive $n^{1/3-\epsilon}$ term.

\subsection{Organization of This Paper}
\label{sec:intro-org}

\ifthenelse{\boolean{conferenceversion}}
{%

We present the necessary preliminaries in
\refsec{sec:preliminaries}
and then give a detailed overview of our results in
\refsec{sec:overview}.
Due to space constraints we can only sketch the proofs 
in this extended abstract
and therefore refer to the upcoming full-length version for all
missing details.  Some concluding remarks are presented in
\refsec{sec:conclusion}.
}
{%
We present the necessary preliminaries in
\refsec{sec:preliminaries}
and then give a detailed overview of our results in
\refsec{sec:overview}.
We prove an asymptotic separation between \pstandard and reversible pebbling in
\refsec{sec:separation-standard-reversible}.
In \refsec{sec:tight-bounds} we compute the exact price of some classic graphs, trees and pyramids, that we use in \refsec{sec:technical-constructions} to construct technical gadgets. These play a key role in \refsec{sec:pspace-completeness-reversible}, where we show that reversible pebbling is \pspace-complete.
We detail the graph product for reversible pebbling in \refsec{sec:blow-up-reversible} and its counterpart for \pstandard pebbling in \refsec{sec:blow-up-standard}.
Some concluding remarks are presented in
\refsec{sec:conclusion}.
}

\section{Preliminaries}
\label{sec:preliminaries}
 
All logarithms in this paper are base $2$ unless otherwise specified.
For a positive integer~$n$ we write $\intnfirst{n}$ to  denote the set
of integers $\set{1, 2, \ldots, n}$.
We use Iverson bracket notation
\begin{equation}
  \ib{B} =
  \begin{cases}
    1 & \text{if the Boolean expression $B$ is true;} \\
    0 & \text{otherwise;}
  \end{cases}
\end{equation}
to convert Boolean values to integer values.
 
\subsection{Boolean Formula Notation and Terminology}
\label{sec:prelim-boolean-notation}

A \introduceterm{literal}
$\lita$
over a Boolean variable $\varx$ is either
the variable $\varx$ itself or 
its negation~%
$\olnot{\varx}$ 
(a \introduceterm{positive} or \introduceterm{negative} literal,
respectively). A \introduceterm{clause}   
$\clc = \lita_1 \lor \formuladots \lor \lita_{\clwidth}$ 
is a disjunction of literals. 
A \introduceterm{$\clwidth$\nobreakdash-clause} is a clause that
contains at most $\clwidth$~literals. A  formula~$\fstd$ in
\introduceterm{conjunctive normal form (CNF)}
is a conjunction of
clauses
$\fstd = \clc_1 \land \formuladots \land \clc_m$.  
\mbox{A \introduceterm{\kcnfform{}}} is a CNF formula consisting of
\xclause{\clwidth}{}s. We think of clauses and CNF formulas as sets, so that
the order of elements is irrelevant and there are no repetitions.

A
\introduceterm{quantified Boolean formula (QBF)}
is a formula
$\QBF = Q_1 x_1 \, Q_2 x_2 \, \ldots \, Q_n x_n \, \, \fstd$,
where
$\fstd$ is a CNF formula over variables $x_1, \ldots, x_n$
and 
$Q_i \in \set{\forall, \exists}$
are universal or existential quantifiers
(\ie 
the formula is in prenex normal form with all
variables bound by quantifiers).
It was shown in~\cite{SM73WordProblems} that it is \PSPACE-complete to
decide whether a QBF is true or not 
(where we can assume \wolog that $\fstd$ is a \xcnfform{3}).

\subsection{Graph Notation and Terminology}
\label{sec:prelim-graph-notation}

We write
$G = (V,E)$
to denote a graph with vertices~$\vertices{G}=V$ and edges
$\edges{G}=E$. All graphs in this paper are directed acyclic graphs
(DAGs).  
An edge
$(u,v) \in \edges{G}$
is an \introduceterm{outgoing edge} of~$u$ and an
\introduceterm{incoming edge} of~$v$, and we say that $u$ is a
\introduceterm{predecessor} of $v$ and that $v$ is a
\introduceterm{successor} of~$u$.
We write $\prednode[G]{v}$ to denote the set of all predecessors of~$v$
in~$G$ and $\succnode[G]{v}$ to denote all its successors.
Vertices with no incoming edges are called \introduceterm{sources} and
vertices with no outgoing edges are called \introduceterm{sinks}.
For brevity, we will sometimes refer to a DAG with a unique
sink as a 
\introduceterm{\singlesinkdagtext{}},
and this sink will usually be denoted~$\sink$.

Taking the transitive closures of the predecessor and successor
relations, we define the 
\introduceterm{ancestors}
$\ancnode[G]{v}$ of~$v$ to be the set of vertices that have a path to~$v$
and the 
\introduceterm{descendants}
$\descnode[G]{v}$
to be the set of vertices on some path from~$v$.
By convention, $v$ is an ancestor and descendant of itself.
We write
$\pancnode[G]{v} = \ancnode[G]{v} \setminus \set{v}$
and
$\pdescnode[G]{v} = \descnode[G]{v} \setminus \set{v}$
to denote the
\introduceterm{proper ancestors} and 
\introduceterm{proper descendants} of~$v$, respectively.
These concepts are extended to sets of pairwise incomparable
vertices by taking unions so that 
$\ancnode[G]{U} = \Union_{u \in U} \ancnode[G]{u}$,
$\pancnode[G]{U} = \Union_{u \in U} \pancnode[G]{u}$,
et cetera, 
where we say that
the vertices in $U$ are pairwise incomparable when no vertex in the
set is an ancestor of any other vertex in the set.
When the graph~$G$ is clear from context we will sometimes drop it
from the notation.

\subsection{\Pstandard and Reversible Pebble Games}
\label{sec:prelim-pebbling}

A
\introduceterm{pebble configuration}
on a DAG
$G = (V,E)$
is a subset of vertices
$\pconf \subseteq V$.
We consider the following three rules for manipulating pebble
configurations: 
\begin{enumerate}
\item 
  \label{item:placement}
  $\pconfafter=\pconfbefore\unionSP \set{v}$ 
  for
  $v \notin \pconfbefore$
  such that
  $\prednode[G]{v} \subseteq \pconfbefore$
  (a \introduceterm{pebble placement} on~$v$).
\item
  \label{item:removal}
  $\pconfafter = \pconfbefore\setminus \set{v} $ for
  $ v\in \pconfbefore$
  (a \introduceterm{pebble removal} from~$v$).
\item
  \label{item:reversible-removal}
  $\pconfafter = \pconfbefore\setminus \set{v} $ for
  $ v\in \pconfbefore$
  such that
  $\prednode[G]{v} \subseteq \pconfbefore$
  (a \introduceterm{reversible pebble removal} from~$v$).
\end{enumerate}
A
\introduceterm{\pstandard pebbling} 
from
$\pconf_{0}$
to~$\pconf_{\stoptime}$
is a   sequence of pebble configurations
$
\pebbling =
(\pconf_{0}, \pconf_{1}, \ldots, \pconf_{\stoptime})
$
where each configuration is obtained from the preceding one by the 
rules
\ref{item:placement}
and~\ref{item:removal}
while in a 
\introduceterm{reversible pebbling} 
rules
\ref{item:placement}
and~\ref{item:reversible-removal}
should be used.
The \introduceterm{time} of a pebbling
$\pebbling = (\pconf_0, \ldots, \pconf_{\stoptime})$
is $\pebtime{\pebbling} = \stoptime$,
and the \introduceterm{space} is
$\pebspace{\pebbling} =
\maxofexpr[0 \leq t \leq {\stoptime}]{\setsize{\pconf_t}}$.

We say that a pebbling is
\introduceterm{\pebunconditional{}} 
if $\pconf_0 = \emptyset$ and
\introduceterm{\pebconditional{}}
otherwise.
The
\introduceterm{pebbling price}
$\blackprice[G]{\pconf}$
of a pebble configuration~$\pconf$
is the minimum space of any \pebunconditional
\pstandard pebbling on~$G$  ending 
\mbox{in~$\pconf_\stoptime = \pconf$},
and we define the 
\introduceterm{reversible pebbling price}
$\persistentprice[G]{\pconf}$
by taking the minimum over all \pebunconditional reversible pebblings
reaching~$\pconf$.
The pebbling price of a
\singlesinkdagtext $G$ with sink~$\sink$
is
$\blackprice{G} = \blackprice[G]{\set{\sink}}$,
and the reversible pebbling price of~$G$ is
$\persistentprice{G} = \persistentprice[G]{\set{\sink}}$.
We refer to such pebblings as
\introduceterm{(\pebcomplete{}) pebblings of~$G$}
or
\introduceterm{pebbling strategies for~$G$}.
Again, when $G$ is clear from context we can drop it from the
notation, and from now on we will usually abuse notation by omitting
the curly brackets around singleton vertex sets.

For technical reasons, we will often be interested in distinguishing
particular flavours of reversible pebblings.
Suppose that $v$ is a vertex in~$G$ and that
$\pebbling = (\pconf_0 = \emptyset, \pconf_1,\ldots, \pconf_{\stoptime})$
is a reversible pebbling. We will use the
following terminology and notation:
\begin{itemize}
\item 
  $\pebbling$ is a \introduceterm{visiting pebbling} of $v$ if
  $v \in \pconf_{\stoptime}$.
  The 
  \introduceterm{visiting price} 
  $\visitprice{v}$
  of~$v$
  is the minimal space of any such pebbling.
  
\item 
  $\pebbling$ is a 
  \introduceterm{surrounding pebbling} of~$v$ if
  $\prednode{v} \subseteq \pconf_{\stoptime}$
  and the 
  \introduceterm{surrounding price} 
  $\surroundprice{v}$
  is the minimal space of any such pebbling.
  
\item 
  $\pebbling$ is a \introduceterm{persistent pebbling} of $v$ if 
  it is a reversible pebbling of~$v$ in the sense defined before, \ie
  such that
  $P_{\stoptime}=\set{v}$.
  We will sometimes refer to 
  $\persistentprice{v}$
  as the
  \introduceterm{persistent price}
  of~$v$ to distinguish it from the visiting and surrounding prices.
\end{itemize}
We also define the
visiting price for a \singlesinkdagtext~$G$ with sink~$\sink$ as
$\visitprice{G}= \visitprice[G]{\sink}$
and the surrounding price as 
$\surroundprice{G} = \surroundprice[G]{\sink}$. 

Note that because of reversibility we could obtain exactly the same
visiting space measure by defining a visiting pebbling of~$v$ to be a
pebbling
$\pebbling = (\pconf_0, \pconf_1,\ldots, \pconf_{\stoptime})$
such that
$\pconf_0 = \pconf_{\stoptime} = \emptyset$
and
$v \in \Union_{0\leq t \leq \stoptime} \pconf_t$,
and let the visiting price be the minimal space of any such pebbling.
This is because once we have reached a configuration
containing~$v$ we can simply run the pebbling backwards
(because of reversibility) until we reach the empty configuration again.
We can therefore think of a pebbling as \introduceterm{visiting}~$v$ if
there is a pebble on $v$ at some point but this pebble does not stay
on~$v$ until the end of the pebbling.
In a \introduceterm{persistent} pebbling the pebble remains on~$v$ until all
other pebbles have been removed.
A \introduceterm{surrounding} pebbling,
finally, is a pebbling that reaches exactly the point where a pebble
could be placed on~$v$, since all its predecessors are covered by
pebbles (\ie $v$ is ``surrounded'' by pebbles),
but where $v$ is not necessarily pebbled.

It is not hard to see that for a \singlesinkdagtext~$G$ we have the
inequalities 
\begin{equation}
  \label{eq:standard-leq-visiting}
  \blackprice{G} \leq \visitprice{G}
\end{equation}
and
\begin{equation}
  \label{eq:surrounding-leq-visiting-leq-persistent}
  \surroundprice{G} \leq \visitprice{G} \leq \persistentprice{G}
  \eqperiod
\end{equation}
Perhaps slightly less obviously,
we also have the following useful equality.

\begin{proposition}
  \label{lem:surrounding-eq-persistent}
  For any vertex $v$ in a DAG $G$ it holds that
  $\surroundprice{v} = \persistentprice{v} - 1$.
\end{proposition}

\ifthenelse{\boolean{conferenceversion}}
{\begin{IEEEproof}}
{\begin{proof}}
  To see that $\persistentprice{v} \leq \surroundprice{v}+1 $
  consider a surrounding pebbling $\pebbling^{S}$ of space
  $\surroundprice{v}$.
  Let~$\pebbling^{*}     $ be the pebbling which
  first runs $\pebbling^{S}$ to surround $v$,
  then puts a pebble on $v$, and
  finally runs the reverse of $\pebbling^{S}$ to ``unsurround'' $v$ (while
  keeping the pebble on~$v$).
  Since $\pebbling^{*}$ is a persistent pebbling of space
  $\surroundprice{v}+1$, the inequality follows.
  
  We now prove that $\surroundprice{v} \leq \persistentprice{v}-1 $.
  Consider a persistent pebbling $\pebbling$ for $v$ of space
  $\persistentprice{v}$.
  Let $\timet$ be the last time that a pebble is put on $v$.
  Then vertex $v$ is surrounded at time $\timet$, and
  there is a pebble on $v$ since time $\timet$.
  Let $\pebbling_{\geq \timet}$ 
  be the \pebconditional pebbling obtained from $\pebbling$
  after
  time $\timet$, with the modification that vertex $v$ has no pebble
  throughout $\pebbling_{\geq \timet}$,
  and let $\pebbling^{R}_{\geq \timet}$ be this pebbling run in reverse.
  Then $\pebbling^{R}_{\geq \timet}$
  is a surrounding pebbling in space at
  most $\persistentprice{v}-1$, and the inequality follows.
\ifthenelse{\boolean{conferenceversion}}
{\end{IEEEproof}}
{\end{proof}}

\subsection{The Dymond--Tompa and Raz--McKenzie Games}
\label{sec:DT-and-RM}

As described above, the \introduceterm{\dtgametext{}} 
on a \singlesinkdagtext~$G$
is played in rounds by  two players 
\introduceterm{\dtpebbler{}} and  \introduceterm{\dtchallenger{}}.
In the first round \dtpebbler places a pebble on the sink~$\sink$
and \dtchallenger challenges this vertex.
In all subsequent rounds, \dtpebbler places a pebble on an arbitrary
empty vertex and \dtchallenger chooses to either challenge this new
vertex 
(which we refer to as
\introduceterm{\dtjumping{}})
or to re-challenge the previously challenged vertex
(referred to as \introduceterm{\dtstaying{}}).  
The game ends when at the end of a round all the (immediate)
predecessors of the currently challenged vertex are covered by 
pebbles.\footnote{We remark that our description 
  follows~\cite{Chan13JustAPebble} and thus
  differs slightly from the original definition in~\cite{DT85Speedups},
  but the two versions are equivalent 
  for all practical purposes.}
The 
\introduceterm{\dtpricetext{}}
$\dtprice{G}$ of~$G$
is the maximal number of pebbles~$\dtnrounds$ needed for \dtpebbler to
finish the game, or expressed differently the smallest
number~$\dtnrounds$ such that \dtpebbler has a strategy to make the
game end in at most~$\dtnrounds$ rounds regardless of how
\dtchallenger plays.

Let us also for completeness describe the
\introduceterm{Raz--McKenzie game},
which is also played on a \singlesinkdagtext~$G$
by  two players
\introduceterm{\rmpebbler{}} and  \introduceterm{\rmcolourer{}}.
In the first round \rmpebbler places a pebble on the sink~$\sink$
and \rmcolourer colours it \rmred.
In all subsequent rounds, \rmpebbler places a pebble on an arbitrary
empty vertex and \rmcolourer then colours this new pebble either \rmred
or \rmblue.
The game ends when there is a vertex with a \rmred pebble, while all
its predecessors in the graph have \rmblue pebbles.
The 
\introduceterm{\rmpricetext{}}
$\rmprice{G}$ of~$G$   
is the smallest number~$\rmnrounds$
such that \rmpebbler has a strategy to make the game end in at
most~$\rmnrounds$ rounds regardless of how \rmcolourer plays.

The intuition for this game is that the vertices on the graphs have
assigned values true (\rmblue) or false (\rmred), with the condition
that each vertex has value equal to the conjunction of the values of
its  predecessors. \rmcolourer claims that the sink is false, but the above
condition vacuously implies that all source vertices must be true.
\rmcolourer loses when \rmpebbler discovers a violation of the condition.
\rmpebbler wants to find the violation as soon as possible, while
\rmcolourer wants to fool \rmpebbler for as long as possible.

In \cite{Chan13JustAPebble}
\theauthorSMC proved that the equalities
\begin{equation}
  \label{eq:chan13-equalities}  
  \dtprice{G} = \rmprice{G} = \persistentprice{G}
\end{equation}
hold for any \singlesinkdagtext~$G$, \ie
that the reversible pebbling price, the \dtpricetext and the
\mbox{\rmpricetext} all coincide.
Thus, any result we prove for one of these games is also
guaranteed to hold for the other games. 
The above equalities are very convenient in that they allow us to
switch back and forth between the reversible pebble game and the
\dtgametext (or \rmgametext) when proving upper and 
lower bounds, depending on which
perspective is more suitable at any given time.
In particular, when proving lower bounds for reversible pebblings it
is often helpful 
to do so by devising good \dtchallenger strategies in the
\dtgametext. One final technical remark in this context is that in all 
such strategies that we construct it holds that \dtchallenger will
either \dtstay or \dtjump to an ancestor of the currently challenged
vertex. Because of this we can assume without loss of generality that
\dtpebbler only pebbles vertices in the subgraph consisting of
ancestors of the currently challenged vertex. If \dtpebbler pebbles
some vertex outside of this subgraph \dtchallenger will just stay put
on the current vertex, and so \dtpebbler just wastes a round.

\section{Overview of Results and Sketches of Proofs}
\label{sec:overview}

In this section we give a detailed overview of our results and also
sketch some of the main ideas in the proofs.
\ifthenelse{\boolean{conferenceversion}}           
{We refer to the upcoming full-length version 
  for all formal proofs and missing technical details.
}
{In the rest of the paper, we then provide all the missing technical
  definitions and present the actual formal proofs.
}

\subsection{Separation Between \Pstandard and Reversible Pebbling}
\label{sec:separation-overview}

As mentioned in \refsec{sec:intro}, the strongest separation hitherto
known between \pstandard and reversible pebbling is for the
length-$\pathnvert$ path on
vertices $\set{v_{1}, v_{2},  \ldots, v_{\pathnvert + 1}}$
with edges
$(v_{i}, v_{i+1})$ for all $i \in \intnfirst{\pathnvert}$,
which has a \pstandard pebbling with $2$~pebbles
whereas reversible pebblings require space
$\bigtheta{\log\ell}$~\cite{Bennett89TimeSpaceReversible,LV96Reversibility}.
We give a simple construction improving this to a multiplicative
logarithmic separation.

\begin{theorem}
  \label{th:separation-standard-reversible-overview}
  For any function $\pebsp(n) = \Bigoh{n^{1/2 - \epsilon}}$
  for $\epsilon > 0$ constant there are DAGs
  $\set{G_n}_{n=1}^{\infty}$ of size~$\bigtheta{n}$
  with a single sink and fan-in~$2$
  such that
  $\blackprice{G} = \bigoh{\pebsp(n)}$
  and 
  $\persistentprice{G} = \bigomega{\pebsp(n) \log n}$
  (where the hidden constant depends linearly on~$\epsilon$). 
\end{theorem}

A first observation is that if we did not have the bounded fan-in
restriction, 
\refth{th:separation-standard-reversible-overview}
would be very easy.
In such a case
we could just take the path of length~$\pathnvert$,
blow up every vertex~$v_{i}$ to $\pebsp$ vertices
$v_{i}^{1}, \ldots, v_{i}^{\pebsp}$,
and add edges
$\bigl( v_{i}^{j} , v_{i+1}^{j'}\bigr)$
for all $j, j' \in \intnfirst{\pebsp}$, 
so that we get a sequence of
complete bipartite graphs~$K_{\pebsp,\pebsp}$
glued together
as shown in \reffig{fig:pike-unbounded-overview}.
It is not hard to show that any reversible pebbling of this DAG would
have to do
$\pebsp$~parallel, synchronized pebblings of the paths
$\Set{ v_{1}^{j},  v_{2}^{j}, \ldots, v_{\pathnvert + 1}^{j}}$
for $j \in \intnfirst{\pebsp}$,
which would require space
$\bigomega{\pebsp \log \pathnvert}$,
whereas a \pstandard pebbling would clearly only need space~$\bigoh{\pebsp}$.

\begin{figure}[tp]
  \begin{center}
    \subcaptionbox{Path blown up to sequence of $K_{3,3}$-graphs.\label{fig:pike-unbounded-overview}}{%
      \begin{tikzpicture}[scale=.75,rotate=-90]
        \ExamplePikeUnbounded
      \end{tikzpicture}
    }
    \qquad
    \subcaptionbox{\Pike graph of length $9$ and width~$3$.\label{fig:pike-overview}}{%
      \begin{tikzpicture}[scale=.75,rotate=-90]
        \ExamplePike
      \end{tikzpicture}
    }
    \caption{Modifications of path graphs to amplify difference between
      reversible and \pstandard pebbling price.}
    \label{fig:path-graphs}
  \end{center}
\end{figure}
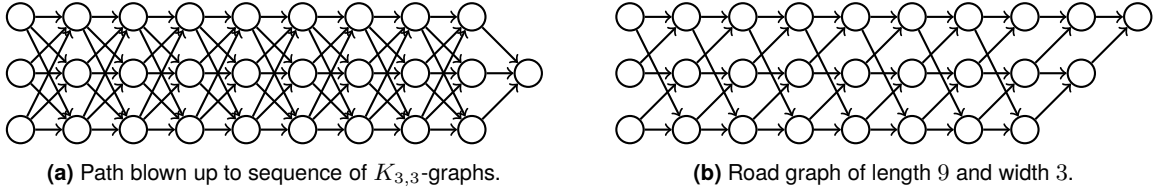

For bounded indegree it is not a priori clear what to do, however, or
indeed whether there should even be a multiplicative separation. But
it turns out that one can actually simulate a lower bound proof along
the same lines as above by considering a
layered
graph as in \reffig{fig:pike-overview}, 
with  $\pebsp$~parallel paths of length up to~$\pathnvert$ and
with every path having an extra edge
fanning out to its ``neighbour path'' above (or at the bottom for the
top row) at each level. We will refer to this construction as a
\introduceterm{\pike graph}
of length~$\pathnvert$ and width~$\pebsp$
(where a path is a 
maximally narrow
\pike of width~$1$).
It is easy to verify that the \pstandard pebbling price of
a \pike of width $\pebsp \geq 2$
is $\pebsp + 2$. We claim that the reversible pebbling price is 
$\Bigomega{\pebsp \log ( \pathnvert / \pebsp)}$, 
from which
\refth{th:separation-standard-reversible-overview}
follows.

To prove the reversible pebbling lower bound it is convenient to think
instead in terms of \dtchallenger strategies in the \dtgametext. The
idea is that \dtchallenger will stay put on the sink until \dtpebbler
has pebbled enough vertices so that there are no pebble-free paths
from any source vertex to the sink.
Intuitively,  the cheapest way for  \dtpebbler to disconnect the
graph is with a straight cut over some layer. 
When this happens, \dtchallenger looks at the latest pebbled vertex
and compares the subgraph between the sources and the cut with the
subgraph between the cut and the sink. If more than half of the graph
is before the cut, \dtchallenger jumps to the latest pebbled
vertex. If not, \dtchallenger stays on the sink. This strategy is then
repeated on a graph of at least half the length. Since every cut by
\dtpebbler requires $\pebsp$~pebbles, \dtchallenger can survive for
roughly 
$\pebsp \log \pathnvert$
rounds (except that the rigorous argument is not quite this simple,
and the slightly smaller factor 
$\log ( \ell / \pebsp)$
in the formal statement of the theorem is in fact inherent).
\subsection{PSPACE-Completeness of Reversible Pebbling}
\label{sec:pspace-completeness-overview}

Moving on to technically more challenging material, let us next discuss
our \PSPACE-completeness result for reversible pebbling, which we
restate here more formally for the record.

\begin{theorem}
  \label{th:pspace-construction}
  Given a single-sink DAG~$G$ of fan-in~$2$ and a parameter~$\pebsp$, it
  is \PSPACE-complete to decide whether~$G$ can be reversibly
  pebbled in space~$\pebsp$ or not. 
  In more detail, 
  given a QBF
  $\QBF = Q_1 x_1 \, Q_2 x_2 \, \ldots \, Q_n x_n \, \, \fstd$,
  where $\fstd$ is a \xcnfform{3}
  over variables $x_1, \ldots, x_n$,
  there is a \efficiently
  constructible 
  single-sink
  graph~$\gadgetConst{\QBF}$ 
  of fan-in~$2$
  and a \efficiently
  computable number $\gamma(\QBF)$ such that
  $\Persistentprice{\gadgetConst{\QBF}}=\gamma(\QBF)+\ib{\QBF \text{\
      is \false}}$.
\end{theorem}

At a high level, our proof is similar to that 
in~\cite{GLT80PebblingProblemComplete}
for \pstandard pebbling:
we build gadgets for 
variables,
clauses, and universal and existential quantifiers, and then glue them
together in the right way so that pebbling through the gadgets
corresponds to verifying satisfying assignments for
universally and existentially quantified subformulas of
the QBF~$\QBF$.
However, the execution of this simple idea
is highly nontrivial even in~\cite{GLT80PebblingProblemComplete}, and
we run into several additional technical difficulties when we want to 
do an analogous reduction for reversible pebbling.

For starters, since the difference in pebbling price
for graphs~$\gadgetConst{\QBF}$ 
obtained from true and false QBFs~$\QBF$ is
just an additive~$1$, we need
exact control over the pebbling price of all components used in the
reduction. For \pstandard pebbling there is no  problem here---exact
bounds on pebbling price are known for quite a wide selection of
graphs---but in the 
reversible setting this becomes an issue already for almost the simplest
possible graph: the complete binary tree of height~$h$. An easy
inductive argument shows that the \pstandard pebbling price of such a 
tree is exactly $h+2$.
Since reversible pebblings find paths more challenging than do
\pstandard pebblings, one could perhaps expect an extra additive 
$\log h$ or so in the reversible pebbling bound. However, the
asymptotically correct bound turns out to be
$h + \bigtheta{\logstar h}$
as shown in~\cite{Kralovic04TimeSpaceReversible},
and the upper and lower bounds on the multiplicative constant 
obtained in that paper
are far from tight.

The story is even worse for the workhorse of the construction in
\cite{GLT80PebblingProblemComplete} (and many other pebbling results),
namely 
\introduceterm{pyramids}
of height~$h$,
which have $i$~vertices at level~$i$ for $i=1, \ldots, h+1$,
and where
the $j$th~vertex at level~$i$ has incoming edges 
from the $j$th~and $(j+1)$st~vertices at level~$i+1$.
There is a very neat proof
in~\cite{Cook74ObservationTimeStorageTradeOff} that the \pstandard
pebbling price is again exactly~$h+2$,
but for reversible pebbling price nothing has been known except that it
has to be somewhere between $h+2$ and 
$h + \bigoh{\logstar h}$
(where the latter bound follows since any strategy for a complete
binary tree of height~$h$ works for any DAG of height~$h$).
As a crucial first step towards establishing
\refth{th:pspace-construction}, 
we exactly determine the reversible pebbling price of pyramids (and
also binary trees).

\begin{theorem}
  \label{th:exact-trees-pyramids-overview}
  For $\Delta$ denoting a positive integer,
  let $\pyrg$ be the function defined recursively as
  \begin{equation*}  
    \pyrg(\pyrx) = \begin{cases}
      0 
      & \text{if $\Delta=1$;}
      \\
      2^{\pyrg(\pyrx-1) + \pyrx - 2} + \pyrg(\pyrx-1) 
      & \text{otherwise;}
    \end{cases}
  \end{equation*}
  and let the inverse $\pyrg^{-1}$ of this function be defined as
  \begin{equation*}
    \pyrg^{-1}(h)
    = \minofset[\mid]{\pyrx}{\pyrg(\pyrx)\geq h}
    \eqperiod
  \end{equation*}
  Then the  persistent pebbling price of a pyramid of height~$h$, as
  well as of a complete binary tree of height~$h$, 
  is $h + \pyrg^{-1}(h)$,
  where
  $\pyrg^{-1}$ is efficiently computable.
\end{theorem}

Even though \refth{th:exact-trees-pyramids-overview} is 
an important step, we immediately run into new problems when trying to
use it as a building block in our reduction for reversible pebbling. 
In the \pstandard pebble
game a \pebcomplete pebbling is any pebbling that reaches the
sink. For the reversible game there is a subtle distinction in
that we can ask whether it is sufficient to just reach the sink or
whether the rest of the graph must also be cleared of pebbles. 
As discussed in \refsec{sec:preliminaries}, this leads to
two different flavours of reversible pebblings, namely
\introduceterm{persistent pebblings},
which leave a pebble on the sink with the rest of the graph being empty,
and \introduceterm{visiting pebblings}, which just reach the sink (and
can then be thought to run in reverse after having visited the sink to
clear the whole graph including the sink from pebbles).  The pebblings
we actually care about are the persistent ones, but we cannot rule out
the possibility that subpebblings of gadgets are visiting
pebblings. Clearly, the difference in pebbling space is at most~$1$,
but this is exactly the additive~$1$ 
of which we cannot afford to lose control!
To make things worse, for pyramids it turns out that
persistent and visiting pebbling prices actually do differ except in
very rare cases.

Because of this, we have to build more involved graph gadgets for which we can
guarantee that visiting and persistent prices coincide.  These gadgets
are constructed in two steps.
First, we take a pyramid and append a path of suitable length,
depending on the height of the pyramid, to the pyramid sink, resulting
in a graph that we call a \introduceterm{teabag}. Second, we take such
teabags of smaller and smaller size and stack them on top of one
another, which yields a graph that looks a bit like a 
\introduceterm{Christmas tree}. These Christmas tree graphs are
guaranteed to have the same pebbling price regardless of whether a
reversible pebbling is visiting or persistent.

With this in hand
we are almost ready to follow the approach
in the \PSPACE-completeness reduction for \pstandard pebbling
in~\cite{GLT80PebblingProblemComplete}.
The idea is that we want to
build gadgets for the quantifiers in 
a formula
\mbox{$\QBF = \forall x \exists y  \cdots Q z \, F$}
of specified pebbling price 
so that the only way to pebble the graph~$\gadgetConst{\QBF}$ without
using too much space is to first pebble the gadget for 
$\forall x$, then~$\exists y$, et cetera, in the correct order until
all quantifier gadgets have been pebbled. Once we get to the clause
gadgets, we would like that the pebbles in the quantifier gadgets are
locked in place encoding a truth value assignment to the variables,
and that the only way to pebble through the clause gadgets without
exceeding the space budget is if every clause contains at least one
literal satisified by this truth value assignment.

In order to realize this plan, 
there remains one more significant technical obstacle to overcome,
however. To try to explain what the issue is, we need to discuss the
\PSPACE-completeness reduction in~\cite{GLT80PebblingProblemComplete}
in slightly more detail.  
The way this reduction imposes an order in which the quantifier
gadgets have to be pebbled
is that pyramid graphs are included ``at the bottom'' of the
gadgets (\ie topologically first
in order).  The source vertices of the quantifier gadgets
all appear in
such pyramids, and one has to pebble through these pyramids
to reach the rest of a gadget (where pebble placements encode variable
assignments as mentioned above).

In the first, outermost quantifier gadget the pyramids have large
height. In the second gadget the pyramid heights are slightly smaller, et
cetera, down to the last, innermost quantifier gadget where the
pyramids have smallest height. In this way, the pyramids are used to
``lock up'' pebbles
and force a strict order of pebbling of the gadgets.
It can be shown that in order not to exceed the pebbling space budget,
any pebbling strategy has to start by pebbling the highest pyramids in
the first gadget. If the pebbling starts anywhere else in the graph,
this will mean that there are already pebbles elsewhere in the graph
when the pebbling strategy reaches the first, highest pyramids in the
outermost quantifier, but if so the overall
pebbling has to use up too much space to pebble through this pyramid.  
One can also show that once the pyramids in the
outermost quantifier gadget have been pebbled, the pebbling cannot
proceed until the next quantifier gadget is pebbled. The pyramids in
this gadget
have smaller height, but there are also pebbles stuck in
place in the outermost gadget, meaning that pyramids must again be
pebbled in exactly the right order to stay within the space budget.

These properties can be used to \introduceterm{normalize}
pebbling strategies in the \pstandard pebble game.
Without loss of generality, one can assume that any 
strategy that starts pebbling a pyramid in a gadget will
complete this local pebbling in one go, leaving a pebble at the sink
of the pyramid, and will not place pebbles anywhere else until the
pebbling of the pyramid has been completed. Also, once a pyramid in a
quantifier gadget has been pebbled in this way, one can prove that it
will never be pebbled again since there is now at least one additional
pebble at some vertex later in the topological order in the graph, and
a repeated pebbling of the pyramid in question would therefore exceed
the space budget.  Thus, not only do the pyramids enforce that the
gadgets are pebbled in the right order---they also serve as
single-entry access points to the gadgets, making sure that each
gadget is pebbled exactly once.

There is no hope of building gadgets with such properties in a
reversible pebbling setting. It is 
simply
not true that a reversible pebbling
will pebble through a subgraph and then never return. Instead, as
already discussed
reversible pebblings will proceed in alternating
phases of interleaved ``forward sweeps'' and ``reverse sweeps,'' and
subgraphs will be entered also in reverse topological order.
Therefore, it is not sufficient to add ``space-locking'' subgraphs at
the source vertices of the gadgets.  Rather, we have to insert
``single-passage points'' inside and in between the gadgets for
quantifiers and clauses.  We obtain such subgadgets by further
tweaking our Christmas tree construction so that it can also connect
two vertices in such a way that any pebbling has to ``pay a
toll'' to go through this subgraph.  We cannot describe these gadgets,
which we call \introduceterm{\highway{}s}, in detail here, but mention
that the ``space-locking'' property that they have is that when the
entrance vertex is eliminated by having a pebble placed on that
vertex, then the cost of pebbling through the rest of the \highway
drops by~$1$. This is critically used in the subgraph compositions
described next.

Assuming the existence of the necessary technical subgraph
constructions sketched above, 
we can now describe the overall structure of our reduction from
quantified Boolean formulas to reversible pebbling
(where all parameters shown in the figures are fixed
appropriately in the formal proofs).
In the following 
figures
we denote a Christmas tree of (visiting and
persistent) pebbling price~$r$ by the symbol in
\reffig{fig:christmas-overview}, where we only display the sink
vertex. We denote the \highway gadget just discussed by the symbol in
\reffig{fig:highway-overview}. 
We write $r$ to denote the \introduceterm{toll} parameter of the
\highway, where a \highway with toll~$r$ has persistent price $r+2$,
but only $r+1$ if we do not count the source~$a$ as part of the
\highway.
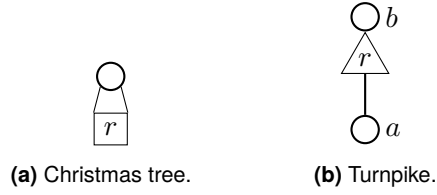
\begin{figure}[t]
  \begin{center}
  \subcaptionbox{Christmas tree.\label{fig:christmas-overview}}{%
    \trimbox{-1.5cm 0cm -1.25cm 0cm}{%
    \begin{tikzpicture}
      \node[Node] (s') at (0,0) {};
      \node[Bank] (s) at (0,-\GadgetBankShift) {$r$};
      \draw (s.north west) to (s'.south west);
      \draw (s.north east) to (s'.south east);
    \end{tikzpicture}}
  }
  \quad
  \subcaptionbox{\Highway.\label{fig:highway-overview}}{%
    \trimbox{-1cm 0cm -1cm 0cm}{%
    \begin{tikzpicture}
      \GadgetAttachmentUp ab\bank
    \end{tikzpicture}}
  }
  \end{center}
  \caption{Legend for technical gadget building blocks.}
\end{figure}

For every variable $x_i$ we have a \introduceterm{variable gadget} as
shown in \reffig{fig:variable-overview},
where we think of a
truth value
assignment~$\partassign$
as represented by pebbles on
vertices $\{\bar\variable_{i},\variable'_{i}\}$ when
$\partassign(\variable_{i})=\false$
and on
$\{\variable_{i},\bar\variable'_{i}\}$ when
$\partassign(\variable_{i})=\true$,
as shown in 
\reftwofigs{fig:variable-false}{fig:variable-true},
respectively.
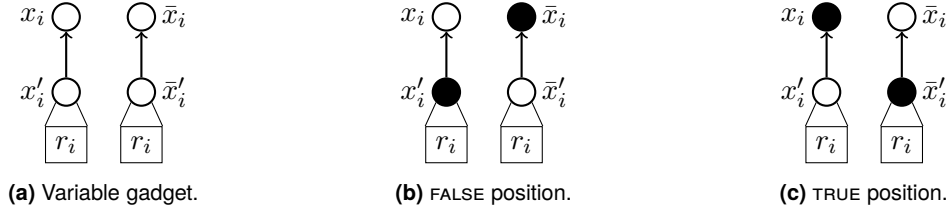
\begin{figure}[t]
  \begin{center}
    \subcaptionbox{Variable gadget.\label{fig:variable-overview}}{%
    \trimbox{-1cm 0cm -1cm 0cm}{%
    \begin{tikzpicture}
      \GadgetVariableWithSources i{\bank_i}
    \end{tikzpicture}}
  }
  \quad
    \subcaptionbox{\false position.\label{fig:variable-false}}{%
    \trimbox{-1cm 0cm -1cm 0cm}{%
    \begin{tikzpicture}
      \GadgetVariableWithSourcesFalse i{\bank_i}
    \end{tikzpicture}}
  }
  \quad
    \subcaptionbox{\true position.\label{fig:variable-true}}{%
    \trimbox{-1cm 0cm -1cm 0cm}{%
    \begin{tikzpicture}
      \GadgetVariableWithSourcesTrue i{\bank_i}
    \end{tikzpicture}}
  }
  \caption{Gadget for variable $x_i$ and pebble positions 
    corresponding to truth value assignments.}
  \end{center}
\end{figure}
 
For every clause $C_j$ we have a
\introduceterm{clause gadget} %
as depicted in \reffig{fig:clause-overview}.
The vertices labelled~$\ell'_{j,k}$ and~$\ell_{j,k}$ 
in \reffig{fig:clause-overview} are identified with the corresponding
vertices for the positive or negative literal~$\ell_{j,k}$ 
in the variable gadget
in \reffig{fig:variable-overview}.
If $\partassign$ satisfies a literal, then there is a pebble on the
entrance vertex of the corresponding \highway, meaning that we can
pebble through the gadget for a clause containing that literal with
one less pebble than if $\partassign$ does not satisfy the clause.

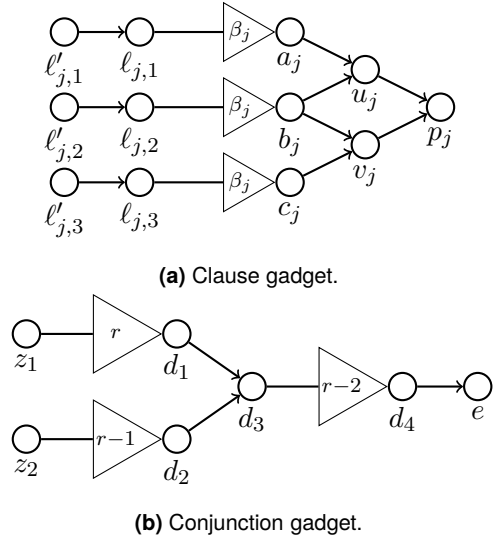
\begin{figure}[t]
  \begin{center}
  \subcaptionbox{Clause gadget.\label{fig:clause-overview}}{%
    \trimbox{-1cm 0cm -1cm 0cm}{%
    \begin{tikzpicture}[rotate=-90]
      \GadgetClause j{\scriptstyle\beta_j}
    \end{tikzpicture}}
  }
  \quad
  \subcaptionbox{Conjunction gadget.\label{fig:conjunction-overview}}{%
    \trimbox{-1cm 0cm -1cm 0cm}{%
    \begin{tikzpicture}[rotate=-90]
      \GadgetConjunction{\scriptstyle\bank}{center}{minimum size=1.2cm}
    \end{tikzpicture}}
  }
  \end{center}
  \caption{Gadgets for clauses and CNF formulas.}
\end{figure}
\begin{figure}[t]
  \begin{center}
  \subcaptionbox{Existential quantifier gadget.\label{fig:exists-overview}}{%
    \trimbox{-1cm 0cm -1cm 0cm}{%
    \begin{tikzpicture}
      \GadgetExistentialQuantifierSingle{i}{i-1}{\bank_i}{%
        \scriptstyle\gamma_i - 5}{center}{minimum size=1.4cm}
    \end{tikzpicture}}
  }
  \quad
  \subcaptionbox{Universal quantifier gadget.\label{fig:forall-overview}}{%
    \trimbox{-1cm 0cm -1cm 0cm}{%
    \begin{tikzpicture}
      \GadgetUniversalQuantifierSingle{i}{i-1}{\bank_i}{%
        \scriptscriptstyle\gamma_i - 6}{%
        \scriptscriptstyle\gamma_i - 7}{center}{minimum size=1.1cm}
    \end{tikzpicture}}
  }
  \end{center}
  \caption{Quantifier gadgets for variable $x_i$.}
\end{figure}
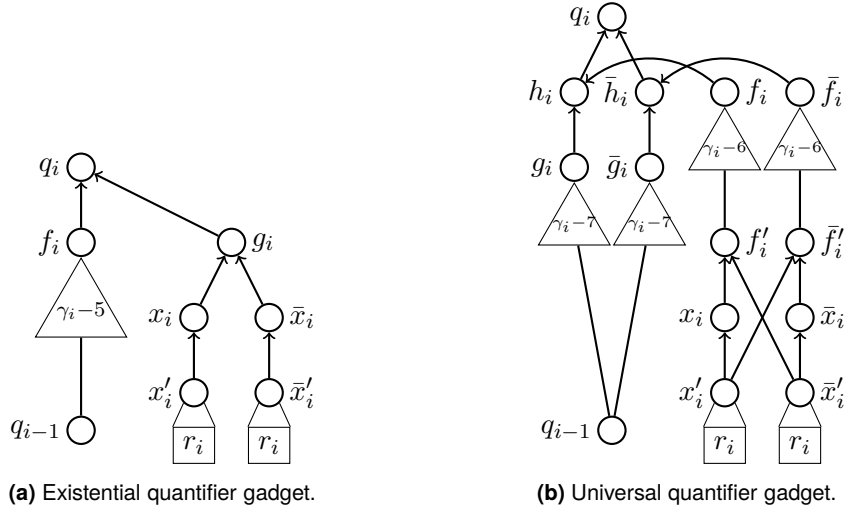

To build the subgraph corresponding to a \xcnf{$3$} formula
$\fstd = \Land^{m}_{j=1} C_{j}$ we join
clause gadgets sequentially using the 
\introduceterm{conjunction gadget} in \reffig{fig:conjunction-overview}.
For technical reasons we start by joining a dummy graph with the
first clause gadget, then we join the result to the second clause
gadget, and so on
up to the $m$th clause of~$\fstd$.
The resulting graph has the property that if pebbles are placed on the
variable gadgets according to an assignment~$\partassign$ that
satisfies~$\fstd$, then the number of additional pebbles needed to
pebble the graph is one less than if the assignment is falsifying.

Finally we have one \introduceterm{quantifier gadget} for each
variable. To describe this part of the construction, 
we
sort the variables indices
in reverse order
from the innermost to the outermost quantifier and
denote by $\QBF_{i}$ the 
subformula
with just the $i$~innermost
quantifiers, 
so that
$\QBF_{0} = 
\fstd = 
\Land^{m}_{j=1} C_{j}$,
$\QBF_{i} = Q_i \variable_{i} \, \QBF_{i-1}$ 
for $Q_i \in \set{\forall, \exists}$,
and $\QBF=\QBF_{n}$. 
We construct graphs $\gadget^{(i)}\defeq \gadgetConst{\QBF_{i}}$, 
starting with $\gadget^{(0)}$ which is just the 
subgraph corresponding to the CNF formula~$\fstd$.
To construct $\gadget^{(i+1)}$ from $\gadget^{(i)}$ we
add an existential gadget as in
\reffig{fig:exists-overview} 
if $\variable_{i}$ is existentially quantified and a universal gadget as in
\reffig{fig:forall-overview} 
if $\variable_{i}$ is universally quantified.
An example of the full construction can be found in
\reffig{fig:example-qbf}.

Given this construction we argue along the same lines as in
in~\cite{GLT80PebblingProblemComplete}, although as mentioned above
there are numerous additional technical complications that we cannot
elaborate on in this brief overview of the proof. 
We show that 
given an assignment $\partassign_{i}$ to
$\set{\variable_{n},\ldots, \variable_{i+1}}$,
the number of additional pebbles needed to pebble~$\gadget^{(i)}$
differs by~$1$ depending on whether 
$\QBF_{i}$ is true under the assignment~$\partassign_i$ or not.
An existential gadget can be optimally pebbled by setting $x_i$ to any
value that satisfies~$\QBF_{i-1}$. To pebble a universal gadget one
needs to assign~$x_i$ to some value, pebble through the gadget, unset $x_i$ and
assign it to the opposite value, and finally pebble through the gadget
again, and both assignments to~$x_i$ must yield satisfying assignments
to~$\QBF_{i-1}$ in order for the pebbling not to go over budget.
Proceeding by induction, 
we establish that the 
complete graph
$\gadget^{(n)}$
can be pebbled within the specified
space budget only if
\mbox{$\QBF=\QBF_{n}$}
is true, which yields
\refth{th:pspace-construction}. 

\begin{figure}[tp]
  \centering
  \begin{tikzpicture}
    \ExampleQBF
  \end{tikzpicture}
  \caption{Example of
    QBF-to-DAG reduction for 
    $\forall x_3 \exists x_2  \forall x_1 
    (x_1 \lor x_2 \lor x_3) \land 
    (\olnot{x_1} \lor x_2 \lor x_3) \land 
    (\olnot{x_1} \lor \olnot{x_2} \lor \olnot{x_3})
    $.}
  \label{fig:example-qbf}
\end{figure}
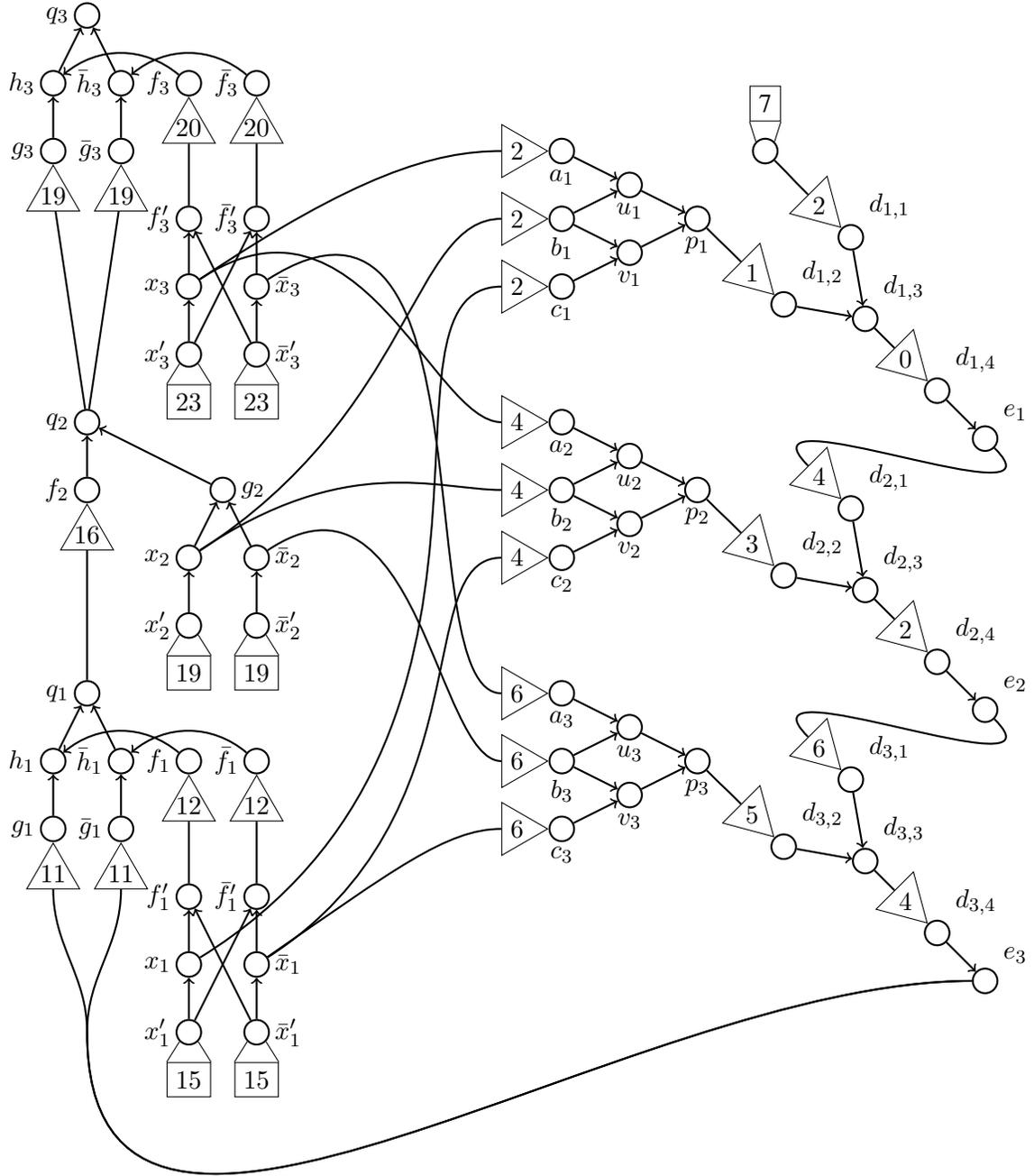
\subsection{PSPACE-Inapproximability up to Additive Constants}
\label{sec:pspace-inapproximability-overview}

Let us conclude the detailed overview of our contributions by
describing what is arguably the strongest result in this paper, namely
a strengthening of the
\PSPACE-completeness of \pstandard pebbling
in~\cite{GLT80PebblingProblemComplete} and of reversible pebbling in
\refth{th:pspace-construction} to \PSPACE-hardness results for
approximating \pstandard and reversible pebbling price to within any
additive constant~$\pebsepconst$.

\begin{theorem}
  \label{th:approx-both}
  For any fixed positive integer $\pebsepconst$ it is \PSPACE-complete
  to decide whether a \singlesinkdagtext~$G$ with fan-in~$2$ has
  (\pstandard or reversible) pebbling price at most~$\pebsp$ or at 
  least~$\pebsp + \pebsepconst$.
\end{theorem}

We remark that it would of course have been even nicer to prove
multiplicative hardness results. We want to stress again, though, that to the
best of our knowledge these are the first results ever for hardness of
approximation of pebble games in a general setting. The fact that
these results hold even for \PSPACE could perhaps be taken both as an indication
that it should be possible to prove much stronger hardness results for
algorithms limited to polynomial time, and as a challenge to do so.

We obtain \refth{th:approx-both} by defining and analyzing two graph
product constructions, one for \pstandard and one for reversible
pebbling, which take two graphs and output product graphs with
pebbling price equal to the sum of the pebbling prices of the two
input graphs (except for an additive adjustment).
These graph products can then be applied
iteratively 
$\pebsepconst-1$ times to 
the graphs obtained by the reductions from QBFs.
In the next theorem we state the formal properties of these graph
products.  

\begin{theorem}
  \label{th:blowup-both}
  Given \singlesinkdagtext{}s 
  $\graphstd_i$ of fan-in 2 and size~$n_i$ for $i=1,2$,
  there are \efficiently constructible \singlesinkdagtext{}s 
  $\blowupBlackConst{\graphstd_1}{\graphstd_2}$
  and
  $\blowupConst{\graphstd_1}{\graphstd_2}$ 
  of fan-in 2 and size~$\Bigoh{(n_1 + n_2)^2}$ such that
  \begin{itemize}
    \item   
      For \pstandard pebbling price it holds that
    $\blackprice{\blowupBlackConst{\graphstd_1}{\graphstd_2}} =
    \blackprice{\graphstd_1} + \blackprice{\graphstd_2} - 1$.
    \item
      For reversible pebbling price it holds that
      $\Persistentprice{\blowupConst{\graphstd_1}{\graphstd_2}} 
    = \persistentprice{\graphstd_1} + \persistentprice{\graphstd_2} + 1$.
  \end{itemize}
\end{theorem}

In the remainder of this section
we try to convey some of the flavour of
the arguments used to prove \refth{th:blowup-both} and to give a sense
of some of  the technical obstacles that have to be overcome during
the analysis. 
In what follows, we will mostly focus on the reversible pebble game,
since it is the technically more challenging and therefore also the
more interesting case. We will briefly discuss the product
construction for \pstandard pebbling at the very end of the
section. 
We will refer to 
$\graphstd_1$
as the \introduceterm{outer graph}
and $\graphstd_2$ 
as the \introduceterm{inner graph}
in the graph products
$\blowupConst{\graphstd_1}{\graphstd_2}$ 
and
$\blowupBlackConst{\graphstd_1}{\graphstd_2}$. 

Intuitively, when taking the graph product of $\graphstd_1$
and~$\graphstd_2$ the idea is to replace every vertex $v$ of the outer
graph~$\graphstd_1$ with a (possibly slightly modified) copy of the inner
graph~$\graphstd_{2}$. We will refer to this copy as the \mbox{$v$-block} in
the product graph.
The edges inside blocks are specified by the inner graph. 
For edges
$(u,v)\in \edges{\graphstd_{1}}$
in the outer graph, we will need to connect the sink of the $u$-block
to vertices in the \mbox{$v$-block} in some way, and this is the crux of the
construction.

A first naive approach would be to add an edge from the sink of the
$u$-block to every source vertex of the \mbox{$v$-block} (as shown in the graph
product~$\blowupNaiveConst{G_1}{G_2}$
in \reffig{fig:overview-blowup}).
Sadly, this simple idea fails for both \pstandard and reversible
pebbling. It is not hard to find examples showing that the pebbling
price of
$\blowupNaiveConst{G_1}{G_2}$
is not a function of the pebbling prices of
$G_1$ and~$G_2$.

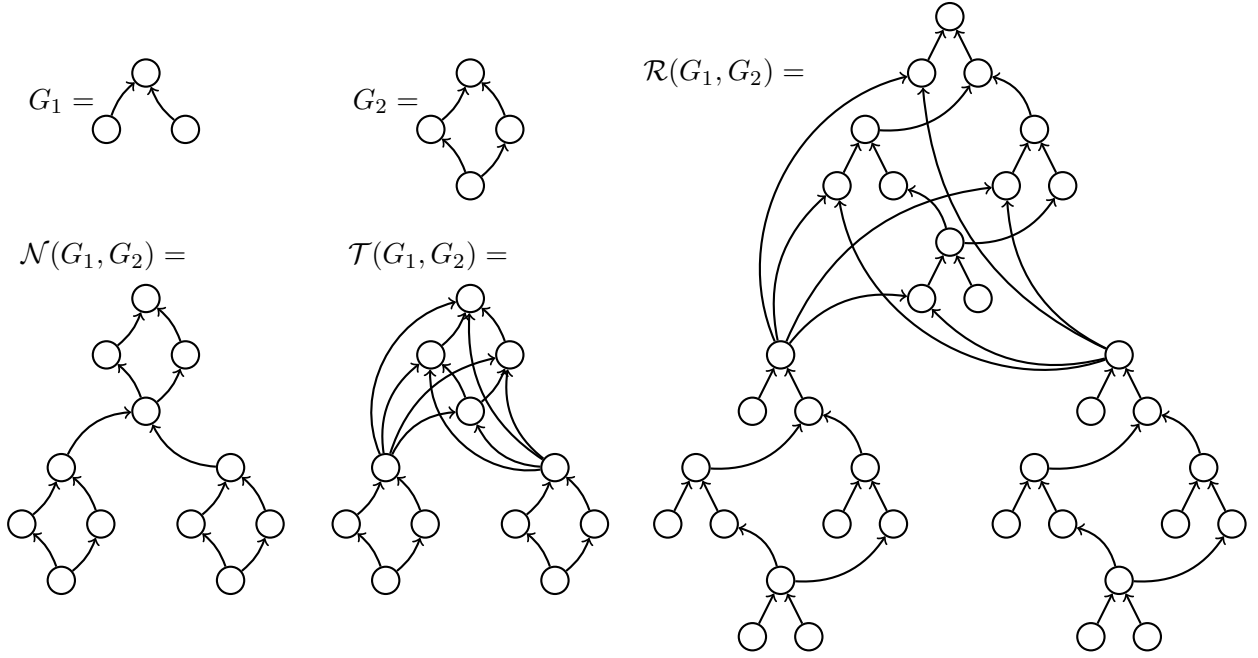
\begin{figure}[tp]
  \begin{center}
    \begin{tikzpicture}[scale=.75]
      \ExampleBlowUp
    \end{tikzpicture}
    \caption{Examples of  graph products 
      as applied to a
      pyramid of height~$1$
      (denoted $G_1$)
      and a rhombus (denoted $G_2$).}
    \label{fig:overview-blowup}
  \end{center}
\end{figure}
 
A slightly more refined idea
is to add edges from the sink of the $u$-block to
all vertices
in the \mbox{$v$-block} (as in the 
graph~$\blowupExpensiveConst{G_1}{G_2}$ 
in \reffig{fig:overview-blowup}).
While we can observe right away that this idea is a non-starter, since
it will blow up the fan-in of the product DAG (and with no bounds on fan-in
the gap amplification would be trivial), it turns out that the
analysis yields interesting insights for the graph product that we
will actually use.  We will therefore employ this toy construction to
showcase some of the ideas and technical challenges that arise in the
actual proof of \refth{th:blowup-both}.

Recall that we want to prove that
$
\Persistentprice{\blowupExpensiveConst{G_1}{G_2}} = 
\pebsp_{1}+\pebsp_{2}-1$, where
$\pebsp_{i} =
\Persistentprice{\graphstd_{i}}$ for $i=1,2$.
To reversibly pebble 
the product graph~$\blowupExpensiveConst{G_1}{G_2}$ 
in at most this amount of space
we simulate a
minimal space pebbling of~$\graphstd_1$,
where pebble placement or removal involving a vertex~$v$
of~$\graphstd_{1}$   
invokes a complete pebbling 
(or unpebbling)
of the copy of~$\graphstd_2$ corresponding
to the \mbox{$v$-block}.
This simulation uses at most $\pebsp_{2}$ pebbles 
in the relevant \mbox{$v$-block}
and at most \mbox{$\pebsp_{1}-1$ pebbles} on sinks of other blocks, \ie 
no more than $\pebsp_{1}+\pebsp_{2}-1$ pebbles in total.

Proving the lower bound
$
\Persistentprice{\blowupExpensiveConst{G_1}{G_2}} \geq
\pebsp_{1}+\pebsp_{2}-1$
is the difficult part. 
Here the approach is to assume that we are given a complete
pebbling~$\pebblingBlowUpExpensive$ of
$\blowupExpensiveConst{G_1}{G_2}$ 
and extract from it a  pebbling strategy~$\pebbling$
for~$\graphstd_{1}$ with the 
hope
that an expensive configuration in~$\pebbling$ will also
help us to pinpoint an expensive configuration
in~$\pebblingBlowUpExpensive$.

The most straightforward way to obtain a pebbling strategy~$\pebbling$
for~$\graphstd_{1}$ from~$\pebblingBlowUpExpensive$ would be to make a
vertex~$v$ in~$\graphstd_{1}$ contain a pebble or not depending
only on the local pebble configuration of the \mbox{$v$-block} 
in~$\blowupExpensiveConst{G_1}{G_2}$.
A natural idea is that $v$  should get a pebble if the  
\mbox{$v$-block} has a pebble on its sink and that this pebble should
be removed from~$v$ when the corresponding block 
has been emptied of pebbles.
If we apply this reduction to a pebbling~$\pebblingBlowUpExpensive$
of~$\blowupExpensiveConst{G_1}{G_2}$ we obtain a valid pebbling
of~$\graphstd_{1}$. 
The problem, however, is that $\pebblingBlowUpExpensive$ might locally
be doing a visiting pebbling (as defined in
\refsec{sec:prelim-pebbling})
of the copy of~$\graphstd_{2}$  corresponding to the \mbox{$v$-block}
as a way of moving pebbles on or off other blocks. The consequence of
this would be that a configuration of 
maximal space~$\pebsp_{1}$ 
in~$\pebbling$ may result from  a configuration in
$\pebblingBlowUpExpensive$ that 
uses
space only
$\pebsp_{1}+\pebsp_{2} - 2$, 
which is off by one compared to what we need and hence 
destroys
the gap in pebbling price that we are trying to create.

If the visiting price of~$G_2$ is the same as its persistent price,
then this problem does not arise, but since this does not hold for
graphs in general we need to 
argue more carefully.
It is true that a visiting pebbling of a 
copy of~$G_2$
might save one pebble
as compared to a persistent pebbling, but whenever the sink contains a
pebble in a visiting but not persistent pebbling we know that there
must also be some other vertex 
in~$G_2$
that has a pebble (or else
the pebbling would be persistent by definition). We need to count such
pebbles also in our analysis.

To this end, we make a distinction between blocks that have paid the
persistent price 
and the blocks that have paid the visiting price but not the
persistent price.
We say that 
the copy of~$G_2$ corresponding to some \mbox{$v$-block} is 
\introduceterm{\vislockedfull{}}, or just
\introduceterm{\vislocked{}} for brevity,
at some point in time if the current pebble configuration 
on its vertices requires reversible pebbling space~$\pebsp_{2}-1$ to
be reached,
and that the \mbox{$v$-block} is 
\introduceterm{\perlockedfull{}} 
(or \introduceterm{\perlocked{}} 
for short)
if the configuration 
has reversible pebbling price~$\pebsp_{2}$.

We can now define a more refined 
way of projecting $\pebblingBlowUpExpensive$-configurations to
$\pebbling$-configurations
as follows.
If a \mbox{$v$-block} has paid the persistent
price, we put a pebble on the corresponding vertex~$v$ in~$G_1$. If a
block has paid just the visiting price but not the persistent price,
then we might still put a pebble on~$v$ in~$G_1$, but we only do so if
an additional (and slightly delicate) technical 
condition\footnote{We do not want to get into too detailed a technical
  argument here, but just for the record pebble configurations on
  $\blowupExpensiveConst{G_1}{G_2}$
  can be
  projected to configurations on~$G_1$ in two stages as follows:
  \begin{enumerate}
  \item 
    Let $\pconf \subseteq \vertices{G_1}$
    consist of all vertices~$u$ such that the
    configuration on the $u$-block in
    $\blowupExpensiveConst{G_1}{G_2}$
    is \perlockedfull.
  \item
    Let $\pconf' \subseteq \vertices{G_1} \setminus \pconf$ consist of
    all vertices~$v$ such that 
    (a)~$v$ is not already surrounded by~$\pconf$, and
    (b)~the configuration on the $v$-block in 
    $\blowupExpensiveConst{G_1}{G_2}$
    is \vislockedfull.
  \end{enumerate}
  With this notation, the projected pebble configuration on~$G_1$ is
  defined to be 
  $\pconf \cup \pconf'$.}
holds for the pebbling configurations in the blocks corresponding to
predecessors of~$v$. This technical condition is designed so that with
some additional work\footnote{One added technical complication that we
  have to take care of here is that when we apply our 
  projection to a pebbling~$\pebblingBlowUpExpensive$ of
  $\blowupExpensiveConst{G_1}{G_2}$ 
  to obtain a sequence of pebble configurations on~$G_1$, this
  sequence need not be a valid pebbling of~$G_1$.
  However, when the projected pebble configuration on~$G_1$ changes
  after a pebbling move we can insert a legal pebbling sequence
  between the two projected configurations that passes through all
  vertices of~$G_1$ corresponding to \vislocked blocks, where pebbles
  are added in topological order and removed in inverse topological
  order, and this local pebbling does not affect the overall
  argument.}
we are still able to extract a legal pebbling strategy for~$G_1$ by
applying this  projection.
Furthermore, 
it will be the case 
that every pebbling
move on a vertex in the outer graph~$G_1$ is the result of the
copy of~$G_2$ corresponding to some \mbox{$v$-block} paying the
persistent or visiting price.

The reversible pebbling~$\pebbling$ thus extracted 
will be  a persistent pebbling 
of~$G_1$ by construction, so it must contain a
configuration with~$\pebsp_1$ pebbles. If this 
configuration was reached because a block paid the persistent price,
then that block contains~$\pebsp_2$ pebbles at a time when at least
$\pebsp_1-1$~other blocks have at least $1$~pebble each, which is the
lower bound that we are after. If the pebble configuration on~$G_1$
in~$\pebbling$ was reached because a block paid the visiting price,
however, then we are potentially still one pebble short. This is where
the additional technical condition mentioned above comes into play.
This condition on the predecessor blocks implies that we can find
at least one other block that also paid just the visiting price and
therefore must contain two pebbles. Summing up, we obtain one block
that  has at least $\pebsp_2-1$~pebbles, another block that has at
least $2$~pebbles, and at least $\pebsp_1-2$~additional blocks that
contain at least $1$~pebble each, and so the lower bound holds in this
case as well.
(Incidentally, this second case is the one where our first, naive,
graph product
$\blowupNaiveConst{G_1}{G_2}$
fails.)

We already observed, however, that the construction
$\blowupExpensiveConst{G_1}{G_2}$
does not get us very far
because it blows up the indegree of the resulting
product graph. 
Therefore, in the actual proof of \refth{th:blowup-both} 
we have to consider a different construction.
Briefly, the idea is to start with the 
graph~$\blowupExpensiveConst{\graphstd_{1}}{\graphstd_{2}}$
but to bring the indegree down by splitting each vertex~$w$
in every block into three vertices
$
w\extblowup,
w\intblowup,
w\outblowup
$.
All edges to~$w$ from other blocks are routed to~$w\extblowup$, all
edges from within the block are routed to~$w\intblowup$, 
and finally we add edges from~$w\extblowup$ and~$w\intblowup$
to~$w\outblowup$. 
This is the graph product
$\blowupConst{\graphstd_{1}}{\graphstd_{2}}$ 
that we use to amplify differences in reversible pebbling price, and 
that is also illustrated
in \reffig{fig:overview-blowup}.
Now we have to prove that the ideas just outlined work for this new
construction where each vertex has been replaced by a small ``cloud'' of
three vertices. The proof of this is much more technically challenging
than for the toy case discussed above, and there is no room to go
into details here. 

At this point we want to switch gears a bit and briefly discuss an
application in proof complexity of the \PSPACE-hardness result for
reversible pebbling. 
Perhaps the most well-studied proof system for proving the
unsatisfiability of, or refuting, CNF formulas is \introduceterm{resolution}
(we do not give any formal definition here, referring instead to, \eg, 
\cite{Segerlind07Complexity} for the necessary details).
Every resolution proof can be represented as a DAG, and the
\introduceterm{depth} of this proof is the length of a longest path in
this DAG. The \introduceterm{resolution depth} of refuting an
unsatisfiable CNF formula 
is the smallest depth of any resolution proof for the formula.
It was shown in~\cite{Chan13JustAPebble} that computing the reversible
pebbling price of a graph of fan-in~$\ell$ 
reduces to computing the resolution depth of a $(\ell+1)$-CNF
formula, and from this we can obtain the following corollary.

\begin{corollary}
  \label{th:resolution-pspace}
  For any fixed positive integer~$\pebsepconst$, 
  it is \PSPACE-complete to compute the resolution depth of refuting 
  $3$-CNF formulas up
  to an additive error~$\pebsepconst$.
\end{corollary}

\ifthenelse{\boolean{conferenceversion}}
{\begin{IEEEproof}}
{\begin{proof}}
  Assuming that we can efficiently compute the resolution depth within
  an additive error at most~$\pebsepconst$, we show how to efficiently
  compute the reversible pebbling price of any graph~$\graphstd$
  within an additive error  $\pebsepconst+1$, contradicting
  \refth{th:approx-both}. 
  
  Letting~$\sink$ denote the unique sink of~$\graphstd$, we consider a
  new graph~$\graphstd'$ 
  which is $\graphstd$ augmented with a new successor~$\sink'$
  of~$\sink$ (\ie
  $\graphstd'=(V\cup\{\sink'\},E\cup\{(\sink,\sink')\})$ 
  in formal notation). 
  The reversible
  pebbling prices of~$\graphstd$ and~$\graphstd'$ differ by at most one.
  For any graph~$\graphstd$, \cite{Chan13JustAPebble}~exhibits an
  efficiently constructible unsatisfiable CNF
  formula~$\fstd_{\graphstd}$ that requires resolution depth equal to
  the  reversible pebbling price of~$\graphstd'$.  The width of the formula
  is equal to the fan-in of~$\graphstd$ plus one, so the result holds
  \mbox{for $3$-CNFs}.

  Hence, if we could estimate the resolution depth of refuting
  $\fstd_{\graphstd}$, \ie the reversible pebbling price of
  $\graphstd'$, within error~$\pebsepconst$, this would yield an
  estimate of the reversible pebbling price of~$\graphstd$ to within error
  $\pebsepconst+1$.
\ifthenelse{\boolean{conferenceversion}}
{\end{IEEEproof}}
{\end{proof}}

We wrap up this section by switching back to pebbling and describing
the product construction
$\blowupBlackConst{\graphstd_{1}}{\graphstd_{2}}$ 
used to amplify \pstandard pebbling price. 
In this construction we also replace every vertex of~$G_1$ with a copy
of~$G_2$, but this time we append what we refer to as a 
\introduceterm{centipede} graph to 
the sink of every copy. A centipede is a path where each vertex but
the source has an extra, unique predecessor. 
To connect the blocks, for every edge
$(u,v)\in E(G_1)$ we add edges from the sink of the~$u$-centipede
to every source of the~$v$-centipede. See 
\reffig{fig:overview-black-blowup}
for an illustration.

\begin{figure}[tp]
  \begin{center}
    \begin{tikzpicture}[scale=.75]
      \BlowUpBlackExample
    \end{tikzpicture}
    \caption{Illustration of  \pstandard pebbling graph product 
      $\blowupBlackConst{\graphstd_{1}}{\graphstd_{2}}$. 
    }
    \label{fig:overview-black-blowup}
  \end{center}
\end{figure}
 
Setting
$\pebsp_{i} = \Pebblingprice{\graphstd_{i}}$ for $i=1,2$,
we can pebble the graph product
$\blowupBlackConst{\graphstd_{1}}{\graphstd_{2}}$ 
in space~$\pebsp_{1}+\pebsp_{2}-1$
by simulating an optimal pebbling of~$G_1$: placing a pebble
on a vertex~$v$ of~$G_1$ is simulated by optimally pebbling the sink
of the corresponding \mbox{$v$-block}, and removing a pebble is
simulated by removing the pebble on the sink.

This pebbling strategy is in fact optimal, and we can show this by
projecting any \pstandard pebbling~$\pebblingBlowUpBlack$ of
$\blowupBlackConst{\graphstd_{1}}{\graphstd_{2}}$ 
to a strategy~$\pebbling$ for~$\graphstd_1$. Each time any block
in~$\blowupBlackConst{\graphstd_{1}}{\graphstd_{2}}$ 
contains $\pebsp_2$~pebbles, we pebble all vertices in~$G_1$ whose
predecessors 
have pebbles and whose corresponding block in~$\blowupBlack$ has a
pebble. When a block 
in~$\blowupBlackConst{\graphstd_{1}}{\graphstd_{2}}$ 
becomes empty, we remove the pebble from the
corresponding vertex in~$G_1$. This projection has the property that when the
sink of a block is pebbled, the corresponding vertex in $\graphstd_1$
is also pebbled. Arguing similarly to in the reversible case, we show
that a strategy~$\pebblingBlowUpBlack$ for $\blowupBlack$
using~$\pebsp$ pebbles yields a strategy for $\graphstd_1$
using~$\pebsp - \pebsp_2 +1$ pebbles. Therefore,
$\pebblingBlowUpBlack$ must use space at 
\mbox{least~$\pebsp_1 + \pebsp_2 -1$, and}
hence the graph product
$\blowupBlackConst{\graphstd_{1}}{\graphstd_{2}}$ has the property
claimed in  \refth{th:blowup-both}.

\section{Separation between \Pstandard and Reversible Pebbling}
\label{sec:separation-standard-reversible}

In this section we discuss how the reversible pebbling price compares
with the \pstandard one.
A reversible pebbling is also a legal \pstandard pebbling, but the
opposite is not always true. However it is possible to construct a
reversible pebbling from a \pstandard pebbling of time $\stoptime$ that
costs at most $ \log \stoptime $ times the price of the \pstandard pebbling.

\begin{theorem}[\cite{Kralovic04TimeSpaceReversible}]
  \label{th:reversible-from-standard}
  If graph $G$ has a \pstandard pebbling of time $\stoptime$ and space
  $p$, then $G$ has reversible pebbling price at most $p \ceiling{\log
    \stoptime}$.
\end{theorem}

\begin{proof}[Proof sketch] 
  Let $\pebbling=(\pconf_{0},\ldots,\pconf_{\stoptime})$ be a \pstandard
  pebbling of $G$ in space $p$. We show a \dtpebbler strategy
  for the Dymond--Tompa game on $G$ that allows \dtpebbler to win in at
  most $p\ceiling{\log \stoptime}$ rounds. Since
  $\dtprice{G} = \persistentprice{G}$ this is sufficient.
  \dtpebbler keeps as an invariant an interval $[a,b]$ such that the
  challenged pebble is in $\pconf_b$ and all vertices in $\pconf_a$
  are pebbled but not challenged. Initially the interval is
  $[0,\stoptime]$, and the strategy proceeds by bisection. At each
  bisection step \dtpebbler starts pebbling the vertices in the
  configuration $\pconf_m$, with $m=(a+b)/2$, in any order. If
  \dtchallenger \dtjumps to a vertex $v$, then let $t$ be the smallest
  number such that $a\leq t$ and $v\in \pconf_t$. \dtpebbler now plays
  in $[a,t]$. The interval halves because $t\leq m$. If \dtchallenger
  \dtstays in all moves, then \dtpebbler plays in $[m,b]$ and the
  interval also halves. When the interval becomes unit, the \dtpebbler
  invariant implies that the move from $\pconf_a$ to $\pconf_b$ is
  precisely a placement on the challenged vertex. Therefore, the
  predecessors of the challenged vertex are pebbled and the game
  ends. The game considers at most $\ceiling{\log \stoptime}$
  configurations of $\pebbling$ and spends at most $p$ rounds on each.
\end{proof}

We already know that the difference between the \pstandard and
reversible pebbling price is unbounded. For example the \pstandard
pebbling price for a path of length $n$ is 2, while its reversible
pebbling price is $\bigtheta{\log n}$.
It follows that if a DAG $ G $ has depth $d$ and a \pstandard pebbling
of time $ \stoptime $ and space $ p $, then 
\begin{equation*}
  \maxofexpr{p,\log d} \leq \persistentprice{G} \leq p\log \stoptime
\end{equation*}

We rule out the possibility of a simulation with only an additive
loss. Indeed we show a separation which is multiplicative in terms of
the logarithm of the size of the graph.

\begin{theorem}
  \label{th:separation-standard-reversible}
  For any function $\pebsp(n) = \Bigoh{n^{1/2 - \epsilon}}$
  where $\epsilon > 0$ is constant there are DAGs
  $\set{G_n}_{n=1}^{\infty}$ of size~$\bigtheta{n}$
  with a single sink and fan-in~$2$
  such that
  $\blackprice{G} = \bigoh{\pebsp(n)}$
  and 
  $\persistentprice{G} = \bigomega{\pebsp(n) \log n}$
  (where the hidden constant depends linearly on~$\epsilon$). 
\end{theorem}

The graphs that we use to witness the separation are the chains or
``wide paths''. A \introduceterm{chain} of width $w$ and length $\ell$
is a graph with $\ell+1$ layers, each having $ w $ vertices, where the
$i$-th vertex of a layer has two incoming edges from the $i$ and
$i+1$-th vertices of the previous layer (modulo $w$). The layers are
indexed from $0$ (the layer of the sources) to $\ell$ (the layer of
the sinks).

Since we want single sink graphs, we define a \introduceterm{\pike} of
width $w$ and length $\ell$ to be a chain of width $w$ and length
$\ell-w+1$ plus a pyramid of height $w-1$, where we identify the sinks of
the chain with the sources of the pyramid. The layers are indexed in
the same way as in the chain.%
\footnote{Equivalently, a \pike of length $ \ell $ and width $ w \leq
  \ell$ is an induced subgraph of a chain of length $ \ell $ and width
  $ w $. Fix one arbitrarily sink $ s $ in the chain: the subgraph
  induced by vertices in $\ancnode{s}$ is indeed a \pike
  of length $ \ell $ and width $ w $.}

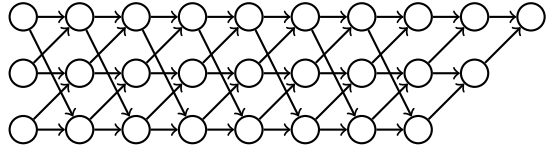
\begin{figure}[tp]
  \begin{center}
    \begin{tikzpicture}[scale=.75,rotate=-90]
      \ExamplePike
    \end{tikzpicture}
    \caption{\Pike graph of length 9 and width 3.}
    \label{fig:pike}
  \end{center}
\end{figure}

By pebbling each layer in order, we get a \pstandard pebbling of a \pike
of width $w$ which uses $w+2$ pebbles.
The reversible pebbling of a \pike depends on its length: a \pike of
width $w$ and length $\ell$ has a reversible pebbling price
$\bigoh{w\log \ell}$. The idea is to simulate in parallel $ w $
copies of the reversible pebbling of the path of length $ \ell $,
which has price $ \bigoh{\log \ell} $. 
We prove \refth{th:separation-standard-reversible} by choosing for each $n$ a \pike of width $w=s(n)$ and length $\ell=n/w$, and showing that this pebbling is essentially optimal.

\medskip 

\newcommand{\stdbset}{B}

A \introduceterm{blocking set} for a vertex set $ T \subseteq V(G)$ is
a subset of vertices $\stdbset\subseteq V(G)$ such that every path
from any source to any vertex in $T$ must contain a vertex in
$\stdbset$. We also say that $\stdbset$ \introduceterm{blocks} $T$.
A blocking set for all sinks of a directed acyclic graph $G$ is also
called a blocking set of $ G$.
We say that $\stdbset$ is a minimal blocking set if no subset of
$\stdbset$ is a blocking set.

\newcommand{\firstlayernot}{a}
\newcommand{\lastlayernot}{b}

A chain of width $ w $ has blocking sets with $ w $ vertices, all in
the same layer. It turns out that if a minimal blocking set has
vertices in multiple layers then it must be larger than that.
We say that a blocking set \introduceterm{spreads} over $d$ layers
when $a$ and $b$ are the lowest and the highest layers that the
blocking set intersects, respectively, and
$d=\lastlayernot-\firstlayernot+1$.

\begin{lemma}
  \label{lem:separator-is-large}
  A minimal blocking set of a chain that spreads over $d$ layers has
  size at least $d+w-1$.
\end{lemma}

\begin{proof}
  Consider such a minimal blocking set $\stdbset$.
  Sort the layers of the chain from $ 0 $ (sources) to $ \ell $
  (sinks) and
  let $\firstlayernot$ and $\lastlayernot$ be the first and last layers with vertices
  in $ \stdbset $, so that $ d=\lastlayernot -\firstlayernot +1$.
  If $ \firstlayernot = \lastlayernot $ then $\stdbset$ must contain $ w $ vertices
  and the Lemma holds. For the rest of the proof we assume $ \firstlayernot
  < \lastlayernot$.

  Define $f(i)$ to be the number of vertices at layer
  $i$ that can reach a sink without passing through a vertex in $\stdbset$
  (and are not in $\stdbset$ themselves).
  By definition and minimality we get that $ f(i)=0 $ for all layers
  $i\leq\firstlayernot$; $0 < f(i) < w$ for all layers $\firstlayernot < i \leq \lastlayernot$ and
  $f(i)=w$ for $ i > \lastlayernot$. 

  Now we compute the intersection between $\stdbset$ and each
  layer. 
  All vertices at layer $ \lastlayernot $ can reach the sink unless
  they are in $ \stdbset$, therefore $ \stdbset$ must have $w-f(\lastlayernot)$
  elements at the last layer. We claim that that $
  \stdbset$ contains at least $f(i+1)-f(i)+1$ vertices from layer $i$,
  for $\firstlayernot \leq i < \lastlayernot$.

  Indeed, consider the set of the $f(i+1)$ vertices at layer $i+1$
  that can reach a sink, and define $ N_{i} $ to be the set of their
  predecessors in layer $ i $. Since $0<f(i+1)<w$, there are at least
  $f(i+1)+1$ vertices in $N_{i}$.
  
  A vertex in the $i$-th layer can reach a sink if and only if it is
  in $N_i$ and not blocked by $B$.  Since we assumed that exactly $
  f(i) $ of them can reach a sink, it must be that $ |N_{i}| - f(i)$
  vertices are blocked by $ \stdbset $ right on layer $i$, i.e., they
  are contained in $\stdbset$. Thus the intersection between $
  \stdbset $ and layer $ i $ is $\setsize{B_i} \geq
  \setsize{N_{i}}-f(i) = f(i+1)-f(i)+1$.
  
  Using these facts we get that 
  \begin{equation*}
    \setsize{\stdbset} = 
    \sum^{\lastlayernot}_{i=\firstlayernot} \setsize{B_i} \geq 
    w - f(\lastlayernot) +
    \sum^{\lastlayernot-1}_{i=\firstlayernot} f(i+1)-f(i) + 1
    =w +
    f(\firstlayernot) + (d-1) = d+w-1 \eqperiod
  \end{equation*}
\end{proof}

We need to generalize \reflem{lem:separator-is-large} to a \pike in
order to handle blocking sets within the pyramid part.

\begin{lemma}
  \label{lem:pike-separator-is-large}
  A minimal blocking set of a \pike that spreads over $d$ layers has
  size at least $d+q-1$, where $q$ is the width of the topmost layer. 
\end{lemma}
\begin{proof}[Proof]
  If the blocking set is located on a single layer the lemma follows
  immediately. 
  Otherwise the proof is very similar to the one of
  \reflem{lem:separator-is-large}, except that the intersection
  between $\stdbset$ and its last layer $\lastlayernot$ has size at
  least $q-f(\lastlayernot)$.
\end{proof}

Now we prove the lower bound in \refth{th:separation-standard-reversible}

\begin{lemma}
  \label{lem:chain-upper-bound}
  The reversible pebbling price of a \pike of width $w$ and length $\ell$ is at least $w\log (\ell/w)/2$.
\end{lemma}

\begin{proof}
We give a \dtchallenger strategy by induction over $\ell$ that lasts
for $w\log (\ell/w)/2$ moves. Furthermore, the strategy \dtstays as
long as the sink is connected to the sources. The base case is a
\pike of length $\ell \leq w-1$, \ie a pyramid of height $\ell$,
in which case the lemma holds vacuously.

We say that a directed path in the graph is \introduceterm{\semiopen} when
there are no pebbles on it except for its last vertex. A \semiopen path from a
vertex to itself is a single vertex with a pebble on it.

During the game \dtchallenger{} \introduceterm{focuses} on a subgraph of
the \pike, and keeps the following invariant at every round: there is a
\semiopen path from the sink of this subgraph to the currently challenged
pebble.
This concretely means that if \dtpebbler places a pebble inside the
subgraph then \dtchallenger plays according to its strategy for
that subgraph. Instead if the new pebble blocks the \semiopen path
between the currently challenged pebble and the sink of the subgraph,
\dtchallenger \dtjumps to the new pebble---essentially making the
path shorter.

If at some round \dtchallenger focuses on a subgraph, in later rounds
\dtchallenger will never challenge a vertex which is neither in the
subgraph nor in the \semiopen path between its sink and the current
challenge.

Let us now give the strategy for playing inside the subgraph.
As long as the sink of the subgraph is connected to the sources,
\dtchallenger \dtstays.
If the sink is disconnected from the sources by a blocking set,
\dtchallenger decides to \dtjump or to \dtstay depending on the
position of the blocking set. Before describing how \dtchallenger
decides, we describe how the strategy continues in both cases.

We consider a minimal blocking set $\stdbset$ and note that the last pebbled vertex $u$ is in any blocking set. Indeed, there is a path from the sources to the sink that only has pebbles at $u$ and at the sink, otherwise the sink would have already been blocked.

We first consider the strategy after \dtchallenger decides to \dtjump.
Let $v$ be the vertex in the \semiopen path from the sources to $u$ at
the layer immediately before all the vertices in $\stdbset$. From now on
\dtchallenger focuses on the subgraph induced by the ancestors of $v$,
which is a \pike. 
This \pike is not blocked, has no vertex in the blocking set
$\stdbset$, and there is a \semiopen path from $v$ to the
challenged pebble $u$. 

If \dtchallenger decides to \dtstay, it focuses on the \pike with
sources at the layer immediately after all vertices in $ \stdbset $.
Again, this \pike is not blocked, it is disjoint from the blocking set
$\stdbset$, and there is a \semiopen path from its sink to the
currently challenged pebble.

It remains to describe how \dtchallenger decides to \dtjump or
to \dtstay.
If the last layer of the blocking set has width $q < w$, then by \reflem{lem:pike-separator-is-large} it has size at least $q + d -1$.
In this case \dtchallenger \dtjumps and focuses on a \pike of length
$\ell-(q + d - 2)$.
Let $\dtprice{w,\ell}$ be the \dtpricetext of a \pike of width $w$ and length $\ell$.
Overall the \dtchallenger strategy lasts for
\begin{equation}
\setsize{\stdbset} + \dtprice{w,\ell-(q+d-2)}
\geq q + d - 1 + w\log((\ell-(q+d-2))/w)/2 \geq w\log(\ell/w)/2
\end{equation}
steps.
Otherwise the blocking set $\stdbset$ has size at least $w+d-1 \geq w$
for some $ d\geq 1 $ and by \reflem{lem:separator-is-large} it spreads over at most $d$
layers, where $\firstlayernot$ is the lowest and $\lastlayernot$ is the highest, with $ d=\lastlayernot-\firstlayernot+1 $ and $m=(\firstlayernot+\lastlayernot)/2$.  
If $m \leq \ell/2$, \dtchallenger \dtstays and focuses on the \pike of
length $\ell - \lastlayernot - 1$ obtained considering only the vertices at
the layers from $ \lastlayernot + 1 $ to $ \ell $.
Overall the \dtchallenger strategy lasts for
\begin{multline}
\setsize{\stdbset} + \dtprice{w,\lastlayernot-1}
\geq w + d-1 + w\log((\lastlayernot-1)/w)/2
\geq w + d-1 + w\log((\ell-d-1)/2w)/2 \\
= w/2\bigl(\log((\ell-d-1)/2w) + 2 + 2(d-1)/w\bigr)
\geq w/2\bigl(\log(\ell/2w)+1\bigr)
= w\log(\ell/w)/2
\end{multline}
steps.
If $m>\ell/2$, \dtchallenger \dtjumps and focuses on a \pike of length
$\firstlayernot-1$ obtained considering one vertex at layer $ \firstlayernot-1 $ which is
connected to the sources, and taking all vertices which have a path
toward such vertex. 
Overall the \dtchallenger strategy lasts for
\begin{multline}
\setsize{\stdbset} + \dtprice{w,\firstlayernot-1}
\geq w + d-1 + w\log((\firstlayernot-1)/w)/2 \\
\geq w + d-1 + w\log((\ell-d-1)/2w)/2
\geq w\log(\ell/w)/2
\end{multline}
steps.
\end{proof}

Note that \refth{th:separation-standard-reversible} follows if the \pike width is $w=s(n)$ and the length is $\ell=n/w$ because $w\log(\ell/w)=w\log(n/w^2)=\bigtheta{w\log n}$ if $w=\Bigoh{n^{1/2-\epsilon}}$.

\section{Tight Bounds for Trees and Pyramids}
\label{sec:tight-bounds}

In this section we show matching upper and lower bounds for the
persistent pebbling price of complete binary trees and pyramids.
Asymptotically tight results for trees were given in
\cite{Kralovic04TimeSpaceReversible}.
The pyramid graph of height $h$ has a vertex for every pair $(i,j)$
with $0 \leq i \leq j \leq h$. The sources are the vertices
$(0,j)$ for $j\geq 0$, the sink is the vertex $(h,h)$, and every
vertex $(i,j)$ for $i<h$ has one outgoing edge going left to vertex
$(i+1,j)$ if $j>i$ and one outgoing edge going right to vertex
$(i+1,j+1)$ if $j<h$.
A pyramid can be obtained from a complete binary tree of height $h$ by
identifying together some of its vertices, in such a way that the left
and right predecessors of two vertices that get identified, get
pairwise identified as well.
For this reason an upper bound for binary trees also holds for the
pyramids, and a lower bound for pyramids also holds for binary trees.

In the following, let $p = h+\pyrx$ be the persistent price of a
pyramid. A pyramid of height $h$ has \pstandard pebbling price $h+2$
\cite{Cook74ObservationTimeStorageTradeOff}, which means that in the
\pstandard pebbling only two extra pebbles are needed compared to the
height of the pyramid. In a similar fashion $\pyrx$ can be interpreted
as the extra space needed by persistent pebbling, with respect to
pyramind height.
We want to estimate the height of the pyramid that has persistent
pebbling with at most $\pyrx$ extra pebbles.
\begin{definition}\label{def:pyramid-extra-pebbles-pyrg}
  Consider $\Delta \in \Nplus$ and let $\pyrg(\pyrx)$ be the function
  defined by the following recursion,
  \begin{equation*}  
    \pyrg(\pyrx) = \begin{cases}
      0 & \text{if $\Delta=1$}\\
      2^{\pyrg(\pyrx-1) + \pyrx - 2} + \pyrg(\pyrx-1) & \text{otherwise.}
    \end{cases}
  \end{equation*}
We define its inverse as 
\begin{equation*}
  \pyrg^{-1}(h) = \min\{\pyrx \mid \pyrg(\pyrx)\geq h\}.
\end{equation*}
\end{definition}
\noindent We show that $\pyrg(\pyrx)$ is the maximum height of a pyramid that can be persistently pebbled using $\pyrx$ pebbles on top of $h$. Observe that $\pyrg(\pyrx)=\bigomega{\underbrace{2^{2^{\cdots^2}}}_\pyrx}$, therefore $\pyrg^{-1}(h) = \bigoh{\log^* h}$.

\begin{proposition}
  We can compute $\pyrg^{-1}(h)$ in time ${(\log h)}^{\bigoh{1}}$ and
  space $\bigoh{\log h}$.
\end{proposition}
\begin{proof}
  If $h=0$ then $\pyrg^{-1}(0)$ is $1$ by definition. 
  If $h>0$ then 
  we need to find the smallest $\pyrx>1$ such that
  $h \ge \pyrg(\pyrx)$. 
  We start from $\pyrx\defeq{2}$ and go upward. At each step we keep in
  memory the value $\pyrg(\pyrx-1)$, which is smaller than $h$,
  and we test the condition 
  \begin{equation*}
    h < \pyrg(\pyrx) 
    \quad
    \text{equivalent to}
    \quad 
    \lfloor\log(h - \pyrg(\pyrx-1))\rfloor 
    < 
    \pyrg(\pyrx-1) + \pyrx - 2 \eqperiod
  \end{equation*}
  The latter test can be achieved in time in $\log h$ by checking the
  length of the bit representation of $h - \pyrg(\pyrx-1)$.
  If the test fails we output $\pyrx$ otherwise we store the value
  $\pyrg(\pyrx)$, we fix $\pyrx\defeq \pyrx+1$ and we continue.
  We can do the whole computation by storing at most $4$ numbers less
  than $h$, each step is polynomial in the length of the binary
  representation of the numbers involved, and we need to do at most
  $\bigoh{\log^* h}$ steps.
\end{proof}

\begin{theorem}
  \label{th:exact-trees-pyramids}
  The persistent pebbling price of a binary tree of height $h$ and a
  pyramid of height $h$ is $h + \pyrg^{-1}(h)$.
\end{theorem}

To prove the theorem we need the exact value of the persistent
pebbling price of paths. 

\begin{lemma}[Path graphs~\cite{LV96Reversibility}]
  \label{lem:PathGraphs}
  The persistent pebbling price of a path of length $h$ (\ie with
  $h+1$ vertices) is $\floor{\log(h)} + 2$.
\end{lemma}

\begin{lemma}[Upper bound for binary trees]
  \label{lem:UpperBound}
  The persistent pebbling price of a complete binary tree of height
  $h \le \pyrg(\pyrx)$ is at most $h+\pyrx$.
\end{lemma}
\begin{proof}
  We are going to prove the lemma by induction over $ h $.
  For the base case we observe that a binary tree of height $0$ can be
  pebbled with $1$ pebbles.
  For the general case we assume the statement of the lemma for height
  $i<h$, and we show that the surrounding pebbling price of the
  complete binary tree of height $h$ is at most $ h+ \pyrx -1$.
  Proposition~\ref{lem:surrounding-eq-persistent} immediately implies
  the lemma for height $h$.

  Let us denote the root by $v_h$ and the right child of $v_i$ by
  $v_{i-1}$. The strategy is as follows. First we persistently pebble
  the left child of $v_i$ for $i$ from $h$ down to $k \defeq \pyrg(\pyrx -1) +
  1$, in this order. By the induction hypothesis $(i-1) + \pyrx$ pebbles
  are enough to persistently pebble the left child of $v_i$, and there
  are $h - i$ pebbles left on the rest of the graph from previous steps. So we are
  within the bound $h-1+\pyrx$, and after the last step we have $h -
  (k-1)$ pebbles on the tree.

  Then we persistently pebble $v_{k-1}$, the right child of
  $v_k$. Since $k-1 = \pyrg(\pyrx-1)$, by induction hypothesis $(k-1)+(\pyrx-1)$
  pebbles are enough and we are within the bound. 
  Let $j \defeq h-k+1$. So far we used $j+1 = h - k + 2$ pebbles.

  Finally we surround the sink of path $(v_k, v_{k+1}, \ldots,
  v_h)$, which has $j$ vertices, using $\floor{\log(j-1)} + 1$
  pebbles.  
  Observe that by
  construction $j - 1 = h - k \le \pyrg(\pyrx) - \pyrg(\pyrx-1) - 1 <
  2^{\pyrg(\pyrx-1) + \pyrx - 2}$, hence we have the bound
  $\floor{\log(j-1)} < \pyrg(\pyrx-1) + \pyrx - 2$. Counting the total
  number of pebbles in the graph gives $\bigl( h - k + 2 \bigr) +
  \floor{\log(j-1)} + 1 \le \bigl( h - \pyrg(\pyrx-1) + 1 \bigr) +
  \bigl( \pyrg(\pyrx-1)+ \pyrx - 3 \bigr) + 1 = h + \pyrx - 1$
  pebbles.
\end{proof}

We prove the lower bound for a slight generalization of pyramids in order to
obtain a lower bound on the visiting price in addition to the persistent price.

\begin{definition}
An $(h,\ell)$-teabag is the union of a pyramid of height $h$ and a path of length $\ell$, where we identify the sink of the pyramid and the source of the path.
\end{definition}
Observe that an $(h,0)$-teabag is a pyramid.

For the lower bound we will also need the following basic fact about
pyramids. Recall that a blocking set is a subset of vertices
$\stdbset\subseteq V(G)$ such that every path from any source to the
sink must contain a vertex in $\stdbset$. Also recall that a directed
path in the graph is \semiopen when there are no
pebbles on it except for its last vertex.

\begin{proposition}[\cite{Cook74ObservationTimeStorageTradeOff}]
  \label{stm:minimalcuts}
  Consider a blocking set $\stdbset$ on a pyramid of height $ h$;
  consider a vertex $ v $ at level $ k $ such that there is a path
  between $v$ and the sink whose intersection with $\stdbset$ is at most
  $\{v\}$. Let $ U $ be the set of vertices in the sub-pyramid rooted at
  $ v $. Then $\setsize{\stdbset\setminus{U}} \geq h-k$.
\end{proposition}
\begin{proof}
  Pick an arbitrary path which starts at vertex $v$, reaches the
  pyramid sink and does not intersect $\stdbset$ anywhere other than in
  $v$.
  Denote such path as $(v_{k}, v_{k+1}, \ldots, v_{h})$ where $v_{k}$
  is another name for $v$ and $v_{h}$ is the sink. Each $v_{i}$ is at
  height $i$ in the pyramid.
  On pyramids there is a natural notion for edges to go either left or
  right.
  For each $i>k$ we define the path $\DAGNamePath_{i}$ as follows: if
  edge $(v_{i-1},v_{i})$ goes right then $\DAGNamePath_{i}$ is the
  unique path that starts at a source vertex, always goes left, and  ends at
  $v_{i}$;
  if edge $(v_{i},v_{i-1})$ goes left then $\DAGNamePath_{i}$ is the unique path
  that starts at a source vertex, always goes right, and ends at $v_{i}$.
  It is easy to verify that none of $\DAGNamePath_{h}, \ldots, \DAGNamePath_{k+1} $
  intersects any of the vertices in $U$, and that these paths are all
  pairwise vertex disjoint. 
  Since $\stdbset$ is a blocking set it must contain one vertex for
  each $\DAGNamePath_{i}$ and the proposition follows.
\end{proof}

\begin{lemma}
\label{lem:teabag-persistent}
  The persistent pebbling price of the $(h,\ell)$-teabag is at least $h+\pyrx+1$ if either of the following holds:
  \begin{itemize}
    \item $h>\pyrg(\pyrx)$,
    \item $h > \pyrg(\pyrx-1)$ and $\ell>\pyrg(\pyrx)-h$. 
  \end{itemize}
\end{lemma}

\begin{proof}
  We define a \dtchallenger strategy for the Dymond--Tompa game by
  induction over $h$ and $\ell$ in this order. Furthermore this
  strategy \dtstays on the sink until \dtpebbler blocks the graph. For
  the base case $h=0$, the statement is trivial.

  Assume that the last \dtpebbler move blocks the graph, meaning that
  the currently pebbled vertices form a blocking set, and fix
  $\stdbset$ to be a minimal one. The vertex $v$ pebbled at that round
  must be in $\stdbset$. We have two cases depending on $k$ the layer
  of vertex $v$.
\begin{itemize}

\item Case $ k > \pyrg(\pyrx-1) $: \dtchallenger \dtjumps to $v$. The pebble
  on $ v $ blocks the sources from the sink, so there must be a \semiopen
  path between $ v $ and a source.
  Let $ U $ be the set of pebbles contained in the vertices of $P_v$,
  the subgraph of predecessors of $ v $ (notice that $ v \in U$).
  Consider a new game on $P_{v}$, in which the first actions of
  \dtpebbler are to pebble $ U\setminus\{v\}$ in any order, while
  \dtchallenger \dtstays on $ v $. The set $ U\setminus\{v\} $ does
  not block the subgraph.
  If $k\geq h$ then the new sub-game ends in at least $h + \pyrx$ steps, and
  the total number of rounds is at least $1+ h + \pyrx$. 
  Otherwise $v$ is inside the pyramid, and the new sub-game ends in $k
  + \pyrx$ steps. We use Proposition~\ref{stm:minimalcuts} to claim
  that $\setsize{\stdbset\setminus{U}} \geq h-k$. So in
  total the rounds in the game are at least
  \begin{equation*}
    \underbrace{1}_{\text{challenge to sink}} + 
    \underbrace{\setsize{\stdbset\setminus{U}}}_{\text{outside
        $P_{v}$}} + 
    \underbrace{k+ \pyrx}_{\text{subgame on $P_{v}$}} \geq h+\pyrx+1 \eqperiod
  \end{equation*}

\jncomment{I really don't like underbraces. Can't we just remove them
  and explain equations properly in running text?}

\item Case $ k \leq \pyrg(\pyrx-1) $: there is a path of length $h-k+\ell$
  from $v$ to the sink having only a pebble on each end. So any optimal
  \dtpebbler strategy must contain a strategy for playing on the
  semiopen path of length $ h-k+\ell-1 $ from one unpebbled successor
  of $v$ to the sink.
  Fix $q=\floor{\log(h-k+\ell-1)}$, the sub-game on the path of length
  $ h-k+\ell-1 $ lasts at least $q+2$ rounds (see
  \reflem{lem:PathGraphs}). The initial challenge on the sink of the
  graph is part of this sub-game, but all moves on $\stdbset$ are not,
  so the total number of rounds is
  \begin{equation*}
    \underbrace{\setsize{\stdbset}}_{\text{blocking set}} + 
    \underbrace{q+2}_{\text{subgame on path}} \geq 
    \underbrace{h - k + 1}_{\text{blocking set}} + 
    \underbrace{\pyrg(\pyrx-1)+\pyrx}_{\text{subgame on path}} \geq 
    h +\pyrx+1 \eqperiod
  \end{equation*}
  The bound on $\setsize{\stdbset}$ holds because
  Proposition~\ref{stm:minimalcuts} on vertex $ v $ implies
  $\setsize{\stdbset\setminus\set{v}} \geq h-k$. The bound on $q+2$
  holds because by hypothesis 
  $ h-k+\ell -1 \geq \pyrg(\pyrx) - \pyrg(\pyrx-1)$,
  which implies that 
  $q \geq \floor{\log(\pyrg(\pyrx)-\pyrg(\pyrx-1))} \geq
  \pyrg(\pyrx-1)+\pyrx-2$. 
  \qedhere
\end{itemize}
\end{proof}

\begin{corollary}[Lower bound for pyramids]
  The persistent pebbling price of a pyramid of height $h>\pyrg(\pyrx-1)$ is
  at least $h+\pyrx$.
  The visiting price of a pyramid of height $h=\pyrg(\pyrx)$ is at least
  $h+\pyrx$.
\end{corollary}
\begin{proof}
  Lemma~\ref{lem:teabag-persistent} claims the first statement.  The
  second one holds because if the pyramid of height $\pyrg(\pyrx)$ had
  visiting price $h+\pyrx-1$, then the $(\pyrg(\pyrx),1)$-teabag would
  have persistent pebbling price $h+\pyrx$, which contradicts
  Lemma~\ref{lem:teabag-persistent}.
\end{proof}

\section{Technical Constructions}
\label{sec:technical-constructions}

In order to discuss lower bounds on pebbling price we need to identify
expensive pebbling configurations, namely the configurations that are
expensive to reach from the empty configuration.
We will often use the reverse direction, \ie that the empty
configuration cannot be reached without passing through an expensive
configuration.

\begin{definition}\label{def:VPlocked}
  A configuration $\pconf$ is \introduceterm{\vislocked} if
  $\persistentprice[G]{\pconf}=\visitprice{G}$. 
  A configuration $\pconf$ is \introduceterm{\perlocked} if
  $\persistentprice[G]{\pconf}=\persistentprice{G}$.
\end{definition}

\subsection{Christmas Tree Construction}
  \label{st:ChristmasTreeConstruction}

  This section builds on the pyramid graphs to provide a graph $T_{r}$
  with equal visiting and persistent prices $r$ for every
  $r\in \N^{+}$. As a preliminary step we show a graph $G_{p}$ with
  persistent price $p$ for every $p\in \N^{+}$.

  \begin{lemma}[Modified Pyramids]
    \label{lem:ModifiedPyramids}
    There is a family of graphs $\{ G_p \}_{p \in \N^+}$ such that
    \begin{enumerate}
      \item $\persistentprice{G_p} = p$;
      \item $G_p$ has in-degree at most two and a unique sink; and
      \item $G_p$ is \efficiently computable given $p$, and $G_p$ has at most
        $p^2$ nodes.
    \end{enumerate}
  \end{lemma}
  \begin{proof}
    The value of $\pyrg^{-1}(h)$, which is the extra pebbling price of pyramids with respect to the
    height, increases only when $h=\pyrg(\pyrx)+1$. 
    Therefore the persistent pebbling price of a pyramid increases by 1 unless
    $h=\pyrg(\pyrx)+1$, in which case it increases by 2.
    If $p = h + \pyrg^{-1}(h)$ for some $h \in \N$ we let $G_p$
    be the pyramid graph of height $h$.
    In this way $G_{p}$ is defined for every $p>0$, unless
    $p = h + \pyrg^{-1}(h) + 1$ for some $h =\pyrg(\pyrx)$. 
    In this case we let $G_p$ be the $(h,1)$-teabag which, by
    Lemmas~\ref{lem:UpperBound} and~\ref{lem:teabag-persistent}, has
    persistent pebbling price $p = h + \pyrg^{-1}(h) + 1$.
  \end{proof}

  We want a \efficiently computable family of graphs $\{ T_\bank \}_{\bank \in
    \N^+}$ with matching visiting price and persistent price, \ie
  $\visitprice{T_\bank} = \persistentprice{T_\bank} = \bank$.
  The idea is to stack up $\bank$ appropriately chosen graphs, so that any
  visiting or persistent pebbling strategy has to spend one pebble per graph.

  We will use $r$ graphs from the family $\{ G_p \}_{p \in \N^+}$
  where each $G_p$ has persistent pebbling price $p$, as constructed
  in \reflem{lem:ModifiedPyramids}. 
  The resulting graph is a stack of modified pyramids of increasing
  sizes. If there is justice in this world, the resulting graph should
  be called a Christmas Tree;
  though a graph theorist may have a hard time calling
  this a ``tree''.

  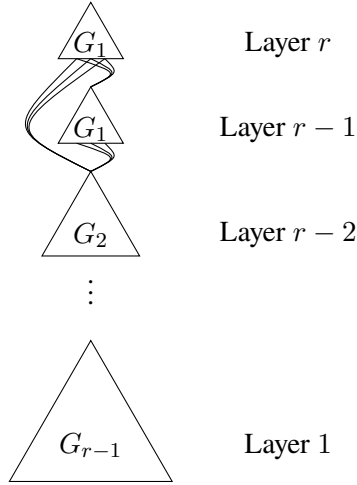
\begin{figure}[ht]
    \begin{center}
      \begin{tikzpicture}[scale=.75]
        \ExampleChristmasTree
      \end{tikzpicture}
      \caption{Illustration of a Christmas Tree in Construction~\ref{con:ChristmasTree}.}
      \label{fig:ChristmasTree}

    \end{center}
  \end{figure}

  \begin{construction}[Christmas Tree]
    \label{con:ChristmasTree}
    Let $\{ G_p \}_{p \in \N^+}$ be given as in \reflem{lem:ModifiedPyramids}.
    Given $\bank \in \N^+$, construct a graph $T_\bank \defeq (V, E)$ as follows.
    Its vertex set $V \defeq V^1 \disjunion V^2 \disjunion
    \dotsb \disjunion V^\bank$ is a disjoint union of $\bank$ layers, where
    for $1 \le \timet < \bank$ the $\timet\th$ layer is a copy of
    $G_{\bank - \timet}$ with vertices $V^\timet \defeq V(G_{\bank - \timet})$,
    and the top-most layer is another copy of
    $G_1$ with vertices $V^\bank \defeq V(G_1)$.
    Its edge set $E \defeq \intraedges \disjunion
    \interedges$ consists of intra- and inter-layer edges.
    The intra-layer edges $\intraedges \defeq E^1 \disjunion
    E^2 \disjunion \dotsb \disjunion E^\bank$ come from the
    corresponding copies of $G_p$, \ie
    for $1 \le \timet < \bank$ the edges on the $\timet\th$ layer
    $E^\timet$ are copies of $ E(G_{\bank - \timet})$ and
    and the edges on the top-most layer $E^\bank$ are copies of $
    E(G_1)$.
    The inter-layer edges 
    $\interedges$ connect, for each
    $1 \leq \timet \leq \bank -2 $, the sink of the subgraph at
    layer $\timet$ with all sources of the subgraphs at layer
    $\timet+1$ and $\timet+2$, and also connect the sink of the copy
    of $G_{1}$ at layer $\bank-1$ with the sources of the copy of
    $G_{1}$ at layer $\bank$.
  \end{construction}

  \begin{lemma}[Christmas Tree]
    \label{lem:ChristmasTree}
    The family of graphs $\{ T_\bank \}_{\bank \in \N^+}$ satisfies
    \begin{enumerate}
      \item $\visitprice{T_\bank} = \persistentprice{T_\bank} = \bank$;
      \item $T_\bank$ has in-degree at most two and a unique sink; and
      \item $T_\bank$ is \efficiently computable given $\bank$, and $T_\bank$
        has at most $\bank^3$.
    \end{enumerate}
  \end{lemma}
  \begin{proof}
    For Item~3, note that each of the $\bank$ layers has at most $\bank^2$
    nodes.

    For Item~2, if $v$ is a node in $T_\bank$ we have two cases
    depending on whether $v$ is some layer's source node or not.
    If it is, then at most two inter-layer edges from lower layers
    point to $v$, and no intra-layer edge does.
    If $v$ is not a source on any layer then only two intra-layer
    edges point to it, since all $G_{p}$ have fan-in at most 2.
    The only sink of $T_\bank$ is the sink of layer $\bank$.

    To see that $\surroundprice{T_\bank} \le \bank - 1$, and thus that
    $\persistentprice{T_{\bank}} \leq \bank$, persistently pebble the
    sink node of $G_{\bank-\timet}$ on layer $\timet$, for $\timet$
    from 1 to $\bank - 1$, keeping only the pebbles on the sinks of
    the previous layers. 
    To persistently pebble layer $\timet$ takes $\bank - \timet$
    pebbles, assuming there is a pebble on each of the sinks of the
    $\timet-1$ lower layers, so in total $\bank - 1$ pebbles suffice.
    Note that
    $G_1$ is a single node, hence when the sinks of the lower layers
    are all pebbled the sink of $T_\bank$ is surrounded.
    The bound $\persistentprice{T_\bank} \le \bank$ follows by
    Proposition~\ref{lem:surrounding-eq-persistent}.

    To see that $\visitprice{T_\bank} \ge \bank$, we argue that when
    visiting layer $\bank$ there is a \vislocked pebble configuration
    on each of the previous layers (see
    Definition~\ref{def:VPlocked}).
    In particular, any layer with a pebble on the sink has
    a \vislocked configuration and if a configuration is \vislocked,
    then it contains a pebble.
    Given a pebbling configuration on $T_\bank$, for the rest of this
    proof we say that layer $\timet$ is \introduceterm{\vislocked} if
    the configuration, restricted to the corresponding subgraph, is
    \vislocked for the subgraph.

    \begin{claim}[Christmas Tree Locker]
      \label{cla:TreeLocker}
      Consider any pebbling that uses less than $\bank$ pebbles.
      In such a pebbling, whenever some layer $(\timet-1)$ and layer
      $\timet$ contain some pebbles, for $2 \le \timet\le \bank$, then
      layers $1, \ldots, \timet - 2$ are all \vislocked.
    \end{claim}
    \begin{proof}
      The claim is true for $\timet = 2$ vacuously, establishing the base case.
      When $\timet > 2$, consider a time that layer $\timet$ starts to have a
      pebble:
      a source node on layer $\timet$ is pebbled, hence there are pebbles on
      the sink nodes of layers $\timet-1$ and $\timet-2$.
      Thus layers $\timet-1$ and $\timet-2$ are \vislocked and each has a pebble.
      Induction hypothesis (on $\timet-1$) further says that layer $\eta$ is
      \vislocked for any $1 \le \eta < \timet - 2$.
      As long as there are pebbles on layers $\timet$ and $\timet-1$, all lower
      layers remain \vislocked:
      for $1 \le \eta \le \timet - 2$, to unlock layer $\eta$ takes
      $\visitprice{G_{\bank-\eta}} \ge \surroundprice{G_{\bank-\eta}} = \bank - \eta - 1
      \ge \bank - \timet + 1$ pebbles (recall Eq.~\eqref{eq:surrounding-leq-visiting-leq-persistent}), but
      there are $\timet-1$ pebbles on layers other than $\eta$, which cannot be
      done with less than $\bank$ pebbles.
    \end{proof}
    Assume for some $\bank \ge 2$, the sink node of layer $\bank$ is pebbled
    using less than $\bank$ pebbles.
    When a source node of layer $\bank$ is pebbled, there is a pebble on the
    sink node of layer $\bank - 1$.
    Claim~\ref{cla:TreeLocker} shows that there is a pebble on layer $\eta$ for $1
    \le \eta \le \bank - 2$, for a total of $\bank$ pebbles, contradicting
    that less than $\bank$ pebbles are used.
    This shows $\visitprice{T_\bank} \ge \bank$ for $\bank \ge 2$, and the case for
    $\bank = 1$ is obvious.
    In the end we get that 
    \begin{equation}
      \bank \le \visitprice{T_\bank} \le \persistentprice{T_\bank} = \surroundprice{T_\bank} + 1 \le \bank
    \end{equation}
    by equation~\ref{eq:surrounding-leq-visiting-leq-persistent} and
    Proposition~\ref{lem:surrounding-eq-persistent}, which gives Item~1.
  \end{proof}

\subsection{Molding}
\label{sst:Molding}

Given a graph $G$ we want to construct a graph $\mold(G)$ with
a special source $s$ and a single sink, such that any pebbling that
visits the sink must go through a configuration with at least
$\visitprice{\mold(G)}$ pebbles, one of which is on vertex $s$.

  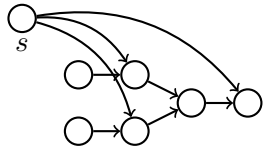
\begin{figure}[ht]
    \begin{center}
      \begin{tikzpicture}[scale=.75]
        \ExampleMolding
      \end{tikzpicture}
      \caption{Example of \refcon{con:Molding}: molding of a pyramid of height 1.}
      \label{fig:ExampleMolding}
    \end{center}
  \end{figure}

\begin{construction}[Molding]
  \label{con:Molding}
  Given a graph $G$, we construct a graph $\mold(G)$ as follows.
  For every vertex $v \in V(G)$, we add to $\mold(G)$ two vertices
  $v\inmold$ and $v\outmold$, and a directed edge
  $(v\inmold, v\outmold)$.
  Also, for every edge $(u,v) \in E(G)$, we add to $\mold(G)$
  a corresponding edge $(u\outmold,v\inmold)$.
  Finally we add to $\mold(G)$ a special new vertex $s$ that we
  connect to all vertices $v\outmold$, \ie for every
  $v \in V(G)$ we add edge $(s, v\outmold)$ to $\mold(G)$. 
  Formally,
  $V(\mold(G)) 
  \defeq 
  \{ s \} \disjunion 
  \bigl\{ v\inmold, v\outmold \where v \in V(G) \bigr\}$
  and 
  $E(\mold(G)) \defeq E_1 \disjunion E_2 \disjunion E_3$, 
  where
  $E_1 \defeq \bigl\{ (v\inmold, v\outmold) \where v \in V(G)
  \bigr\}$,
  $E_2 \defeq \bigl\{ (u\outmold, v\inmold) \where (u,v) \in E(G)
  \bigr\}$
  and 
  $E_3 \defeq \bigl\{ (s, v\outmold) \where v \in V(G) \bigr\}$.
\end{construction}

By construction, if $G$ has in-degree at most two and a unique sink then so does
$\mold(G)$.

\begin{lemma}[Molding]
  \label{lem:Molding}
  Given a graph $G$, the graph $\mold(G)$ has the following properties.
  \begin{enumerate}
    \item $\Persistentprice{ \mold(G) } \le \persistentprice{G} + 2$; and
    \item For any visiting pebbling $\pebbling' = \langle \pconf'_0,
      \pconf'_1, \dotsc, \pconf'_{\stoptime} \rangle$ of $\mold(G)$, there
      is a configuration $P'_b$ (for some $0 \le b \le \stoptime$) using at
      least $\visitprice{G} + 2$ pebbles and containing $s$.
  \end{enumerate}
  \begin{proof}
    For Item~(1), fix any persistent pebbling $\pebbling$ of $G$.
    Simulate the pebbling $\pebbling$ as a persistent pebbling $\pebbling'$ of
    $\mold(G)$ as follows.
    First, pebble the special source $s$ of $\mold(G)$ in $\pebbling'$.
    Afterwards, whenever there is a move in $\pebbling$ to pebble a node $v
    \in V(G)$, make a phase of three moves in $\pebbling'$:
    pebble $v\inmold$, pebble $v\outmold$, unpebble $v\inmold$.
    Similarly, whenever there is a move in $\pebbling$ to unpebble a node $v
    \in V(G)$, make a phase of three moves in $\pebbling'$:
    pebble $v\inmold$, unpebble $v\outmold$, unpebble $v\inmold$.
    If the current configuration in $\pebbling$ is $\pconf$, and
    the configuration in $\pebbling'$ at the end of a phase is $\pconf'$, then
    $\pebbling'$ maintains
    the invariant that $\pconf' = \{ s \} \union \bigl\{ v\outmold
      \where v \in \pconf \bigr\}$.
    As a result, $\pebbling'$ is a legal pebbling on $\mold(G)$:
    whenever $v\inmold$ is pebbled or unpebbled, all its predecessors
    $\prednode{v\inmold} = \bigl\{ u\outmold \where u \in \prednode{v}
      \bigr\}$ have pebbles in $\pconf'$,
    since $\prednode{v}$ have pebbles in $\pconf$;
    whenever $v\outmold$ is pebbled or unpebbled, its predecessors $s$ and
    $v\inmold$ have pebbles.
    If we add a final move in $\pebbling'$ to unpebble $s$, then $\pebbling'$
    persistently pebbles $\mold(G)$.
    Note that whenever $\pebbling$ pebbles or unpebbles $v \in V(G)$ to get to
    configuration $\pconf$, the simulating pebbling $\pebbling'$ uses at most
    two more pebbles to get to $\pconf'$, namely $s$ and $v\inmold$.
    Hence $\Persistentprice{ \mold(G) } \le \persistentprice{G} + 2$.

    For Item~(2), we start with a visiting pebbling $\pebbling' = \langle \pconf_0',
    \pconf_1', \dotsc, \pconf_{\stoptime}'\rangle$ of $\mold(G)$ and
    we construct a visiting pebbling $\pebbling$ of $G$.
    We now define two projection operators that turn pebble
    configurations for $\mold(G)$ into configurations for $G$.
    Let $\outmoldprojTime{\timet} \defeq
    \bigl\{ v \in V(G) \where v\outmold \in \pconf_\timet' \bigr\}$
    be the projection of $\pconf_\timet'$ to $V(G)$ via $v\outmold$, and
    $\anymoldprojTime{\timet} \defeq
    \bigl\{ v \in V(G) \where v\inmold \in \pconf_\timet' \text{ or
      } v\outmold \in \pconf_\timet' \}$ be the projection of $\pconf_\timet'$
    to $V(G)$ via $v\inmold$ or $v\outmold$.
    By definition $\outmoldprojTime{\timet} \subseteq
    \anymoldprojTime{\timet}$.
    Whenever a vertex of $\mold(G)$ is pebbled or unpebbled during
    a pebbling step from  $\pconf_\timet'$ to $\pconf_{\timet+1}'$,
    \begin{enumerate}[(i)]
        \item if the vertex is $s$, then both
        $\outmoldprojTime{\timet}=\outmoldprojTime{\timet+1}$ and
        $\anymoldprojTime{\timet}=\anymoldprojTime{\timet+1}$;
      \item if the vertex is $v\inmold$ for some $v \in V(G)$, then
        $\outmoldprojTime{\timet} = \outmoldprojTime{\timet+1}$
        but $\anymoldprojTime{\timet}$ may differ from
        $\anymoldprojTime{\timet+1}$; 
        and
      \item if the vertex is $v\outmold$ for some $v \in V(G)$, then
        $\anymoldprojTime{\timet} = \anymoldprojTime{\timet+1}$ 
        but $\outmoldprojTime{\timet} \not=
        \outmoldprojTime{\timet+1}$.
    \end{enumerate}
    
    To construct $\pebbling$, we analyze in order each configuration
    $\pconf'_{t}$ in $\pebbling'$. Depending on how the sequences of
    $\outmoldprojTime{\timet}$ and $\anymoldprojTime{\timet}$ evolve,
    we may append new configurations to $\pebbling$.
    In the following $\eta$ is the index of the last configuration
    added to $\pebbling$ and $\timet$ is the configuration of
    $\pebbling'$ under analysis. We maintain the
    following invariants:
    \begin{enumerate}[(a)]
      \item $\outmoldprojTime{\timet} \subseteq \pconf_\eta$; and
      \item for any $v$ in $\anymoldprojTime{\timet} \symmdiff \pconf_\eta$
        it holds that $\prednode{v} \subseteq \outmoldprojTime{\timet}$.
    \end{enumerate}
    Initially at $\timet=0$ and $\eta=0$,
    $\pconf_\timet' = \outmoldprojTime{\timet} =
    \anymoldprojTime{\timet} = \pconf_\eta = \emptyset$,
    so the invariant holds for $\timet=0$.
    Consider a pebbling move in $\pebbling'$ from $\pconf_\timet'$ to
    $\pconf_{\timet+1}'$.
    \begin{enumerate}[(I)]
      \item If $\anymoldprojTime{\timet} = \anymoldprojTime{\timet+1}$
        and $\outmoldprojTime{\timet} = \outmoldprojTime{\timet+1}$ the
        construction does not append any new $\pconf_\eta$ and the
        invariant is preserved.
      \item If $\anymoldprojTime{\timet} \not=
        \anymoldprojTime{\timet+1}$ then we are in case (ii) above, so
        $\outmoldprojTime{\timet}$ does not
        change and some node $v\inmold$ is pebbled or unpebbled. Hence the
        current configuration $\pconf_\timet' \supseteq \prednode{v\inmold} =
        \bigl\{ u\outmold \where u \in \prednode{v} \bigr\}$, thus
        $\outmoldprojTime{\timet} \supseteq \prednode{v}$. 
        The construction does not append a new $\pconf_{\eta}$ and
        the invariant is preserved.
      \item If $\outmoldprojTime{\timet} \not=
        \outmoldprojTime{\timet+1}$ 
        then we are in case (iii) above, so
        $\anymoldprojTime{\timet}$ does not change and some node $v\outmold$
        is pebbled or unpebbled.
        The construction appends to $\pebbling$ the two sequences of
        moves (``Eras'') described below. After each move
        $\anymoldprojTime{\timet} \symmdiff \pconf_\eta$ gets smaller.
        Note that for any $u \in \anymoldprojTime{\timet} \symmdiff
        \pconf_\eta$, the invariant gives $\prednode{u} \subseteq
        \outmoldprojTime{\timet} \subseteq \pconf_\eta$, so
        they can be pebbled or unpebbled in $\pconf_\eta$.
        \begin{description}
          \item[Unpebble Era] while $\pconf_\eta \setminus
          \anymoldprojTime{\timet} \ne \emptyset$,
          pick any $u \in \pconf_\eta \setminus \anymoldprojTime{\timet}$, and
          unpebble $u$ in $\pebbling$ (increment $\eta \defeq \eta + 1$ and
          then set $\pconf_\eta \defeq \pconf_{\eta - 1} \setminus \{ u \}$).
          Since $\outmoldprojTime{\timet} \subseteq \anymoldprojTime{\timet}$,
          $u \notin \outmoldprojTime{\timet}$ and the invariant is
          preserved at time $\timet$.
          \item[Pebble Era] while $\anymoldprojTime{\timet} \setminus
          \pconf_\eta \ne \emptyset$, pick
          any $u \in \anymoldprojTime{\timet} \setminus \pconf_\eta$, and
          pebble $u$ in $\pebbling$
          (increment $\eta \defeq \eta + 1$ and then set $\pconf_\eta \defeq
          \pconf_{\eta - 1} \union \{ u \}$).
          The invariant is preserved at time $\timet$.
      \end{description}
      At the end of the two sequences, $\outmoldprojTime{\timet+1} \subseteq
      \anymoldprojTime{\timet+1} = \anymoldprojTime{\timet} = \pconf_\eta$, so
      the invariant now holds also at time $\timet+1$.
    \end{enumerate}
    We complete the proof of Item~(2). For any visiting pebbling
    $\pebbling'$ of $\mold(G)$, the corresponding pebbling $\pebbling$
    is a visiting pebbling of $G$ by invariant~(a), thus some
    constructed configuration $\pconf_\eta$ has at least
    $\visitprice{G}$ pebbles. 
    The configuration has been appended to
    $\pebbling$ in case (III) above, and without loss of generality we
    can assume it is either at the beginning of an ``unpebble era'' or
    at the end of a ``pebble era'', since the number of pebbles in
    $\pconf_\eta$ decreases in the former and increases in the latter.
    Since the beginning of an ``unpebble era'' other than the first is
    also the end of a ``pebble era'', we can furthermore assume the
    latter.
    This means that for some $\timet$ we have
    $\anymoldprojTime{\timet} = \anymoldprojTime{\timet+1} =
    \pconf_\eta$,
    so either the corresponding configuration $\pconf_\timet'$
    (when $v\outmold$ is unpebbled) or $\pconf_{\timet+1}'$ (when
    $v\outmold$ is pebbled) has at least $\visitprice{G} + 2$ pebbles,
    including $s$, $v\inmold$ and $v\outmold$. This completes
    Item~(2).
  \end{proof}
\end{lemma}

\subsection{\Highway{}s}

As an application of molding (Construction~\ref{con:Molding}) we show
a construction that controls the visiting price of a node, and that
allows us to construct gadgets for the components of a quantified boolean
formula.

  \begin{construction}[\Highway]
    \label{con:Highway}
    For any non-negative integer $\bank$ we define the \emph{\highway}
    of toll $\bank$ from $a$ to $b$ (represented by
    Fig.~\ref{fig:Highway}) as follows.
    If $r=0$ then the \highway just joins the vertices
    $a$ and $b$ with an edge $(a,b)$.
    If $r>0$ let $T_\bank$ be the graph having $\visitprice{T_\bank} =
    \persistentprice{T_\bank} = \bank$ given by
    \reflem{lem:ChristmasTree}.
    The \highway of toll $\bank$ from $a$ to $b$ is the graph
    $\mold(T_\bank)$, identifying $a$ with the special
    source $s$ in $\mold(T_\bank)$, and identifying $b$ with
    the unique sink in $\mold(T_\bank)$.
  \end{construction}

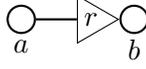
\begin{figure}[ht]
  \begin{center}
    \begin{tikzpicture}[rotate=-90]
      \GadgetAttachment ab\bank
    \end{tikzpicture}
  \end{center}
  \caption{\Highway of toll $\bank$ from $a$ to $b$.}
  \label{fig:Highway}
\end{figure}

  Let $G$ be any graph that contains a \highway of toll $\bank$ from $a$ to $b$.
  Call the nodes $\tilde R \defeq \ancnode[G]{b} \setminus
  \ancnode[G]{a}$ to be \emph{properly} in the \highway, and call the
  nodes $R \defeq \tilde R \union \{ a \}$ to be in the \highway.
  The \highway construction is sensitive to whether the pebble on node $a$ is
  counted, \ie whether the pebbling prices are restricted to $R$ or to
  $\tilde R$.

  \begin{lemma}[\Highway]
    \label{lem:Highway}
    We have $\persistentprice[R]{b} = \visitprice[R]{b} = \bank + 2$ and
    $\persistentprice[\tilde R]{b} = \visitprice[\tilde R]{b} = \bank + 1$.
  \end{lemma}
  \begin{proof}
    $\persistentprice[R]{b} \le \bank + 2$ and
    $\persistentprice[\tilde R]{b} \le \bank + 1$
    by (the proof of) Item~(1) of \reflem{lem:Molding} (since pebbles
    outside of $R$ or of $\tilde R$ are not counted), and $\visitprice[R]{b}
    \ge \bank + 2$ and $\visitprice[\tilde R]{b} \ge \bank + 1$ by Item~(2) of
    \reflem{lem:Molding}.
  \end{proof}

  The fact that $\persistentprice[R]{b}$ and $\persistentprice[\tilde R]{b}$ do not
  differ by more than one holds not only for \highway{}s but for any graph.

  \begin{lemma}[Source Difference]
    \label{lem:SourceDifference}
    Consider regions $R_1$ and $R_2$ such that $R_1 = R_2 \setminus \{ s_2 \}$
    for some source $s_2$ of $R_2$ (\ie $\prednode{s_2} \intersect R_2 =
    \emptyset$).
    We have
    \[
      \begin{matrix}
        \persistentprice[R_1]{v}
        &\le& \persistentprice[R_2]{v}
        &\le& \persistentprice[R_1]{v} + 1 \\
        \visitprice[R_1]{v}
        &\le& \visitprice[R_2]{v}
        &\le& \visitprice[R_1]{v} + 1 \\
        \surroundprice[R_1]{v}
        &\le& \surroundprice[R_2]{v}
        &\le& \surroundprice[R_1]{v} + 1 \enspace .
      \end{matrix}
    \]
  \end{lemma}

\section{PSPACE-Completeness}
\label{sec:pspace-completeness-reversible}

In this section we give all details of the construction in
\refth{th:pspace-construction}, restated in the following theorem, and
its full proof.

\begin{theorem}
  \label{th:pspace-construction-only}
  Given a quantified \xcnf{$3$}~$\QBF$,
  there is a \efficiently
  constructible graph $\gadgetConst{\QBF}$ and a \efficiently
  computable number $\gamma(\QBF)$ such that
  $\Persistentprice{\gadgetConst{\QBF}}=\gamma(\QBF)+\ib{\QBF \text{\
      is \false}}$.
\end{theorem}

Let $ \variable_{1} \ldots \variable_{n}$ be the variables of $\QBF$.
We sort the variables from the innermost to the outermost quantifier
and denote by $\QBF_{i}$ the quantified \xcnf{$3$} with just the $i$
innermost quantifiers, namely $\QBF_{0}=\Land^{m}_{j=1} C_{j}$,
\mbox{$\QBF_{i} = Q_i \variable_{i}\QBF_{i-1}$} for
$Q \in \set{\forall, \exists}$, and $\QBF=\QBF_{n}$. 
For each $\QBF_{i}$ we consider the gadget $\gadgetConst{\QBF_{i}}$,
where $\gadgetConst{\QBF_{0}}$ is built as defined in
Construction~\ref{con:CNFGadget} and each $\gadgetConst{\QBF_{i}}$ for
$i\in[n]$ is built according to either
Construction~\ref{con:ExistentialQuantifierGadget} or
Construction~\ref{con:UniversalQuantifierGadget}, depending on the
$i\th$ innermost quantifier.
Furthermore we fix the sequence of integers $\{\gamma_{i}\}^{n}_{i=0}$
where $\gamma_{0} \defeq 2m+7$, $\gamma_{i} \defeq 3+\gamma_{(i-1)}$
if the $i\th$ innermost quantifier is existential, and
$\gamma_{i}=5+\gamma_{(i-1)}$ if the $i\th$ innermost quantifier
is universal.
To analyze the gadgets for the subformulas we need the next definition.

\begin{definition}
  A \introduceterm{region} is an induced subgraph of the final gadget
  $\gadgetConst{\QBF}$ (or of component gadgets as we build up the
  final gadget). 
  We slightly abuse notation and refer to a region by
  a subset of vertices; for example, given a subset of vertices $\Region$ of
  a gadget $\gadget$, we use $\persistentprice[\Region]{\gadget}$,
  $\visitprice[\Region]{\gadget}$ and $\surroundprice[\Region]{\gadget}$
  to denote the different pebbling prices over the subgraph of $\gadget$ induced on
  $\Region$.
\end{definition}

With these notations and definitions in place,
Theorem~\ref{th:pspace-construction-only} follows immediately from the
next lemma when $i$ is equal to $n$.
In this case $\partassign$ and $\Region$, as stated in the
lemma, are respectively the empty assignment and the full graph
$\gadgetConst{\QBF}$.
In the proof of the lemma we refer to definitions, lemmas and
constructions that we will present in full details in the
coming subsections.
 
\begin{lemma}[Main Lemma]
  \label{lem:QuantifierGadget}
Fix an arbitrary $0\leq i\leq n$ and let
\begin{itemize}
  \item $\partassign$ be an assignment to all but the first $i$
  variables;
  \item $\canonical$ be the canonical set of $\partassign$ according
  to Definition~\ref{def:CanonicalNodes};
  \item $\Region$ be the subset of vertices of
  $\gadgetConst{\QBF_{i}}$ defined as
  $V(\gadgetConst{\QBF_{i}}) \setminus \ancnode{\canonical}$.
\end{itemize}
It holds that
\begin{equation}\label{eq:quantifierRecursion}
  \Persistentprice[\Region]{\gadgetConst{\QBF_{i}}} =  \gamma_{i}
  +\ib{\restrict{\QBF_{i}}{\partassign} \text{\ is \false}}\eqperiod
\end{equation}
\end{lemma}
\begin{proof}
  When $i = 0$, the gadget $\gadgetConst{\QBF_0}$ is the CNF gadget 
  from Construction~\ref{con:CNFGadget} and the base case follows
  immediately by \reflem{lem:CNFGadget}, where $\beta_{m}=2m$.
  When $i>0$, consider the two possible extensions of $\partassign$ that
  assign $\variable_{i}$, namely 
  $\partassign_{0} \defeq \partassign \cup \{\variable_{i}=0\}$ 
  and
  $\partassign_{1} \defeq \partassign \cup \{\variable_{i}=1\}$.
  Consider also the corresponding canonical sets
  $\canonical_{0}, \canonical_{1}$ and the regions
  $\Region_{0} \defeq V(\gadgetConst{\QBF_{(i-1)}}) \setminus
  \ancnode{\canonical_{0}}$
  and
  $\Region_{1} \defeq V(\gadgetConst{\QBF_{(i-1)}}) \setminus
  \ancnode{\canonical_{1}}$. By induction it holds that
  \begin{itemize}
    \item $\Persistentprice[\Region_{0}]{\gadgetConst{\QBF_{(i-1)}}}=
    \gamma_{(i-i)}+\ib{\partassign_{0} \text{\ falsifies\ } \QBF_{(i-1)}}$
    \item   $\Persistentprice[\Region_{1}]{\gadgetConst{\QBF_{(i-1)}}}=
    \gamma_{(i-1)}+\ib{\partassign_{1} \text{\ falsifies\ } \QBF_{(i-1)}}$
  \end{itemize}
  so the hypothesis of Lemmas~\ref{lem:UniversalQuantifierGadget}
  and~\ref{lem:ExistentialQuantifierGadget} holds.
  When $\QBF_{i}=\forall\variable_{i}\QBF_{i-1}$ \reflem{lem:UniversalQuantifierGadget} gives that  
  \begin{equation}\label{eq:universalRecursion}
    \Persistentprice{\gadgetConst{\QBF_{i}}}= 5+\gamma_{(i-i)}+
    \ib{
      \partassign_{0} \text{\ falsifies\ } \QBF_{(i-1)} \text{\ or\ } 
      \partassign_{1} \text{\ falsifies\ } \QBF_{(i-1)}},
  \end{equation}
  and $\QBF_{i}=\exists\variable_{i}\QBF_{(i-1)}$
  \reflem{lem:ExistentialQuantifierGadget} gives that
  \begin{equation}\label{eq:existentialRecursion}
    \Persistentprice{\gadgetConst{\QBF_{i}}}= 3+\gamma_{(i-i)}+
    \ib{
      \partassign_{0} \text{\ falsifies\ } \QBF_{(i-1)} \text{\ and\ } 
      \partassign_{1} \text{\ falsifies\ } \QBF_{(i-1)}}\eqperiod
  \end{equation}
  Equations~\eqref{eq:universalRecursion}
  and~\eqref{eq:existentialRecursion} are equivalent, for the
  respective quantifier type, to
  Equation~\eqref{eq:quantifierRecursion}.
  \end{proof}

\subsection{Literal Gadget}
\label{ssst:LiteralGadget}

  \reflem{lem:ChristmasTree} allows us to create an edge $(u,v)$ with
  $\persistentprice{v} = \persistentprice{u} + 1 = \bank + 1$ given any $\bank \in
  \N^+$. 
  This is used for constructing gadgets for literals (represented as
  \reffig{fig:Literal}).

  \begin{construction}[Literal Gadget]
    \label{con:LiteralGadget}
    Fix a literal $\literal$ and an integer $\bank \in \N^+$.
    Let $T_\bank$ be the graph having $\visitprice{T_\bank} = \persistentprice{T_\bank}
    = \bank$ given by \reflem{lem:ChristmasTree}.
    The literal gadget of price $\bank$ for $\literal$, is denoted
    as $\literalConst{\literal}{\bank}$.
    To construct it, take a copy of $T_\bank$ and call its sink $\literal'$.
    Then add a new node, named $\literal$ as the corresponding literal, and add the edge $(\literal', \literal)$.
  \end{construction}

\begin{figure}[ht]
  \begin{center}
    \begin{tikzpicture}
      \GadgetLiteralWithSources{\literal}{l}\bank{left}
    \end{tikzpicture}
  \end{center}
  \caption{Literal gadget of price $\bank$ for $\literal$.}
  \label{fig:Literal}
\end{figure}
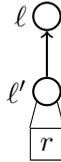

  \begin{lemma}[Literal Gadget]
    \label{lem:LiteralGadget}
    In $\literalConst{\literal}{\bank}$, $\persistentprice{\literal} - 1 =
    \persistentprice{\literal'} = \visitprice{\literal'} = \bank$.
  \end{lemma}
  \begin{proof}
    \mvcomment{Should we have a lemma ``visit($G$)=surr($G^+$)''?}
    Note that visiting node $\literal'$ is equivalent to surrounding node
    $\literal$.
    Hence $\persistentprice{\literal} - 1 = \surroundprice{\literal} = \visitprice{\literal'}
    = \persistentprice{\literal'} = \bank$ by \refeq{eq:surrounding-leq-visiting-leq-persistent}, \refpr{lem:surrounding-eq-persistent},
    and \reflem{lem:ChristmasTree}.
  \end{proof}
  \mvcomment{It would be nice to have T/F/* or True/False/Undef.}

  We associate pebbling configurations on $\Literal \defeq
  \literalConst{\bank}{\literal}$ with truth value \true, \false or \tvundef
  (undefined) as follows.

  \begin{definition}[Literal Position]
    \label{def:LiteralPosition}
    Fix a literal gadget $\Literal \defeq \literalConst{\bank}{\literal}$.
    Given a pebbling configuration $\pconf$, say
    node $\literal'$ is \introduceterm{\vislocked} on
    $\ancnode[\Literal]{\literal'}$
    if
    $\persistentprice[L]{\pconf}=\visitprice{\literal'} = \bank$ (when restricted to
    $\ancnode[\Literal]{\literal'}$);
    that is, if the empty configuration cannot be reached without entering a
    configuration which has $\bank$ pebbles.
    Given a pebbling configuration on $\Literal$, say literal $\literal$ is in
    \begin{itemize}
      \item \true\ position, if node $\literal$ has a pebble and node $\literal'$ is not
        \vislocked;
      \item \false position, if node $\literal'$ is \vislocked; and
      \item \tvundef position, if node $\literal$ has no pebble and node $\literal'$ is not
        \vislocked.
    \end{itemize}
    A \emph{transition} of literal $\literal$ is a change of position for
    $\literal$ (\eg from \true position to \false position, or from \false
    position to \tvundef position) due to a pebble move.
    Finally, we identify certain \emph{canonical} positions with
    configurations on $\Literal$ as follows:
    \begin{itemize}
      \item the canonical \true position is the configuration where only node
        $\literal$ has a pebble;
      \item the canonical \false position is the configuration where only node
        $\literal'$ has a pebble; and
      \item the canonical \tvundef position is the empty configuration.
    \end{itemize}
  \end{definition}

  Clearly, the canonical \true position (resp.~canonical \false position,
  canonical \tvundef position) is indeed a \true position (resp.~\false
  position, \tvundef position).

  \begin{lemma}[Literal Transition]
    \label{lem:LiteralTransition}
    Fix a literal gadget $\Literal \defeq \literalConst{\bank}{\literal}$.
    \begin{enumerate}
      \item it is impossible to transition directly from \tvundef position to
        \true position, and vice versa; and
      \item at a transition, there are $\bank$ pebbles on
        $\ancnode[\Literal]{\literal'}$.
    \end{enumerate}
  \end{lemma}
  \begin{proof}
    Note that node $\literal'$ is \vislocked on any configuration with a pebble on
    node $\literal'$.
    To change from \tvundef position to \true position, node $\literal$ needs to
    be pebbled, and hence node $\literal'$ must have a pebble at some time.
    At that time, node $\literal'$ is \vislocked, so the configuration is in
    \false position.
    This gives Item~(1).

    Hence the only valid transitions are from \tvundef position to \false
    position (or its reverse), and from \false position to \true position (or
    its reverse).
    In any of these, node $\literal'$ needs to switch between locked and
    unlocked status, which requires a configuration with $\bank$ pebbles on
    $\ancnode[\Literal]{\literal'}$ by definition of \vislocked.
    This gives Item~(2).
  \end{proof}

\subsection{Variable Gadget}
\label{ssst:VariableGadget}

  Let $\bank_i \in \N^+$ be an integer to be specified later, which is
  associated with the $i\th$ variable $\variable_i$.

  Inspired by previous works~\cite{GLT80PebblingProblemComplete,%
    HP10PspaceCompleteness,%
    Chan13Thesis}
  truth values are represented using the gadget in \reffig{fig:Variable}.
  \begin{construction}[Variable Gadget]
    \label{con:VariableGadget}
    For the variable $\variable_i$, its variable gadget
    $\gadgetConst{\variable_i}$ is constructed as the disjoint union two
    literal gadgets of price $\bank_i$, one for literal $\variable_i$ and one
    for literal $\bar\variable_i$.
  \end{construction}

\begin{figure}[ht]
  \begin{center}
    \begin{tikzpicture}
      \GadgetVariableWithSources i{\bank_i}
    \end{tikzpicture}
  \end{center}
  \caption{Variable $\variable_i$.}
  \label{fig:Variable}
\end{figure}

  For the gadget $\gadget \defeq \gadgetConst{\variable_i}$, its nodes are
  $V(\gadget) \defeq \Ancnode[\gadget]{ \{ \variable_i, \bar\variable_i \} }$.

  \begin{definition}[Variable Position]
    \label{def:VariablePosition}
    Fix a variable gadget $\gadget \defeq \gadgetConst{\variable_i}$
    consisting of literal gadgets $\Literal_1 \defeq
    \literalConst{\bank_i}{\variable_i}$ and
    $\Literal_0 \defeq \literalConst{\bank_i}{\bar\variable_i}$.
    We identify certain \emph{canonical} positions with configurations on
    $\gadget$ as follows:
    \begin{itemize}
      \item the canonical \true position is the configuration where only nodes
        $\variable_i$ and $\bar\variable_{i}'$ have pebbles;
      \item the canonical \false position is the configuration where only
        nodes $\variable_i'$ and $\bar\variable_{i}$ have pebbles; and
      \item the canonical \tvundef position is the empty configuration.
    \end{itemize}
  \end{definition}

  \begin{lemma}[Variable Assignment]
    \label{lem:VariableAssignment}
    Variable $\variable_i$ can be put into one among canonical \true
    and canonical \false positions using at most $ \bank_{i}+1 $
    pebbles.
  \end{lemma}
  \begin{proof}
    To put variable $\variable_i$ in canonical \true position, persistently
    pebble node $\variable_i$, then persistently pebble $\bar\variable_i'$.
    It can be done with $\bank_i + 1$ pebbles by \reflem{lem:LiteralGadget}.
    A symmetric argument shows that $\variable_i$ can be put in canonical
    \false position with $\bank_i + 1$ pebbles.
  \end{proof}

  As we will see in later sections, the design of the quantifier gadgets would
  ensure that any pebbling strategy would effectively associate truth value
  via the canonical positions.
  This motivates the following definition.

  \begin{definition}[Canonical Nodes]
    \label{def:CanonicalNodes}
    Given a partial assignment $ \funcdescr{\partassign}{[n]}{\{ \true,
      \false, \tvundef \}} $,
    the canonical nodes of variable $\variable_i$ under $\partassign$ are
    \begin{itemize}
      \item $\{ \variable_i, \bar\variable_i' \}$ if $\partassign(i) = \true$;
      \item $\{ \variable_i', \bar\variable_i \}$ if $\partassign(i) =
        \false$; and
      \item $\{\}$ if $\partassign(i) = \tvundef$.
    \end{itemize}
  \end{definition}

  Note that if $\partassign(i) \ne \tvundef$, then there are two pebbles on the
  ancestors of the canonical nodes of $\variable_i$.
  For example, if $\partassign(i) = \true$, then there is a pebble on
  $\ancnode{\variable_i}$ and a pebble on $\ancnode{\bar\variable_i'}$.

  In general, we consider a partial assignment on variables
  $ \funcdescr{\partassign}{[n]}{\{ \true, \false, \tvundef \}} $
  as a partial assignment on literals:
  \begin{itemize}
    \item if variable $\variable_i$ is assigned \true under $\partassign$
      (\ie $\partassign(i) = \true$), then literal $\variable_i$ is assigned
      \true and literal $\bar\variable_i$ is assigned \false;
    \item if variable $\variable_i$ is assigned \false under $\partassign$
      (\ie $\partassign(i) = \false$), then literal $\variable_i$ is assigned
      \false and literal $\bar\variable_i$ is assigned \true;
    \item if variable $\variable_i$ undefined under $\partassign$
      (\ie $\partassign(i) = \tvundef$), then literal $\variable_i$ is \tvundef
      and literal $\bar\variable_i$ is \tvundef.
  \end{itemize}

  As we will argue later, the design of the quantifier gadgets would ensure
  that ``invalid variable assignments'' would not be a problem:
  for example,
  the two literals of the same variable cannot be put into the \true position
  at the same time (\eg\refclaim{cla:Clearance}); also,
  it does not help to put the two literals of the same variable into the
  \false position at the same time, and
  each variable will be assigned eventually (\reftwolems{lem:ExistentialLowerBound}{%
    lem:UniversalLowerBound}).
  
  The canonical nodes (of a partial assignment) are useful for defining
  certain \emph{regions} over different component gadgets in the overall
  construction.
\subsection{Clause Gadget}
\label{ssst:ClauseGadget}

  Let $\beta_j \ge 2$ be an integer to be specified later, which is
  associated with the $j\th$ clause $\clause_j$.
  The gadget for the $j\th$ clause, $\clause_j = \literal_{j,1} \lor
  \literal_{j,2} \lor \literal_{j,3}$, uses as a component the \highway
  gadget which is described in Construction~\ref{con:Highway}.
  Its skeleton is shown in
  \reffig{fig:Clause} (the literal gadgets are simplified in \reffig{fig:Clause}
  for a cleaner diagram).
  Assume that the literals $\literal_{j,1}, \literal_{j,2}, \literal_{j,3}$
  are over distinct variables.

  \begin{construction}[Clause Gadget]
    \label{con:ClauseGadget}
    Assume that for each variable $\variable_i$, $1 \le i \le n$ we have
    the corresponding variable gadget $\gadgetConst{\variable_i}$, \ie two
    literal gadgets for $ \variable_{i} $ and $ \bar\variable_{i} $.
    For the $j\th$ clause $\clause_j$, its clause gadget
    $\gadgetConst{\clause_j}$ is constructed as follows.
    Create nodes $a_j, b_j, c_j, u_j, v_j, p_j$, and edges 
    \begin{equation}
      (a_j, u_j), (b_j,u_j), (b_j, v_j), (c_j, v_j), (u_j, p_j), (v_j,
      p_j).
    \end{equation}
    Finally, add three \highway{}s of toll $\beta_j$, from $\literal_{j,1}$ to
    $a_j$, from $\literal_{j,2}$ to $b_j$, and from $\literal_{j,3}$ to $c_j$
    (where the nodes $\literal_{j,1}, \literal_{j,2}, \literal_{j,3}$
    are the ones from the corresponding literal gadgets).
  \end{construction}

  Note that in \reffig{fig:Clause} the six nodes $\literal_{j,1}$,
  $\literal_{j,1}'$, $\literal_{j,2}$, $\literal_{j,2}'$, $\literal_{j,3}$,
  and $\literal_{j,3}'$ come from the variable gadgets corresponding to the
  variables in literals $\literal_{j,1}$, $\literal_{j,2}$ and
  $\literal_{j,3}$.
  Recall the definitions of canonical nodes in
  \refdef{def:CanonicalNodes}.
  For example, if literal $\literal_{j,1}$ is in \true position, literals
  $\literal_{j,2}, \literal_{j,3}$ \false position, then their canonical nodes
  are $\literal_{j,1}, \literal_{j,2}', \literal_{j,3}'$.

\begin{figure}[ht]
  \begin{center}
    \begin{tikzpicture}[rotate=-90]
      \GadgetClause j{\beta_j}
    \end{tikzpicture}
  \end{center}
  \caption{Clause $j$.}
  \label{fig:Clause}
\end{figure}
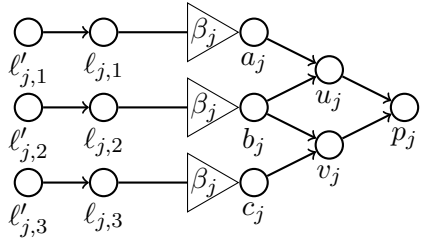

  In this subsection, focus on the gadget $\gadget \defeq
  \gadgetConst{\clause_j}$ constructed for the $j\th$ clause $\clause_j$.
  The gadget $\gadget$ behaves like a disjunction in the sense that at least
  one literal is assigned \true if, and only if, $\beta_j + 3$ additional
  pebbles are needed to surround $p_j$.

  \begin{lemma}[True Clause, upper bound]
    \label{lem:TrueClause}
    Fix a partial assignment $\partassign$.
    Assume none of the literals $\literal_{j,1}, \literal_{j,2},
    \literal_{j,3}$ is assigned \tvundef, and at least one of them is assigned
    \true.
    Let $\canonical_j$ be their canonical nodes.
    Consider the region $R_j \defeq \ancnode[\gadget]{p_j} \setminus
    \ancnode[\gadget]{\canonical_j}$ beyond the canonical nodes.
    Then $\surroundprice[R_j]{p_j} \le \beta_j + 3$.
  \end{lemma}
  \begin{proof}
    Note that the $j\th$ clause gadget is symmetric in the three literals
    $\literal_{j,1}, \literal_{j,2}, \literal_{j,3}$ when restricting
    attention to 
    $\Ancnode[\gadget]{\{ a_j, b_j, c_j \}}
    \setminus
    \Pancnode[\gadget]{\{\literal_{j,1}', \literal_{j,2}',
        \literal_{j,3}'\}}$.
    If at least one of the literals $\literal_{j,1}, \literal_{j,2},
    \literal_{j,3}$ is assigned \true, we claim that it takes at most
    $\beta_j + 3$ pebbles over $R_j$ to leave pebbles only on $\{ a_j, b_j, c_j
    \}$;
    afterwards $p_j$ can be surrounded by pebbling $u_j$ and $v_j$ (note that
    $\beta_j + 3 \ge \underbrace2_{u_j, v_j} + \underbrace3_{a_j, b_j, c_j} =
    5$).

    To prove the claim, note that if $\literal_{j,1}$ is assigned \true, then
    the intersection of $R_j$ with the \highway from $\literal_{j,1}$ to $a_j$
    is precisely the nodes \emph{properly} in the \highway;
    if $\literal_{j,1}$ is assigned \false, then the intersection is precisely
    the nodes in the \highway. This holds similarly for literals $
    \literal_{j,2} $ and $ \literal_{j,3} $.
    By symmetry over
    $\Ancnode[\gadget]{\{ a_j, b_j, c_j \}}
    \setminus
    \Pancnode[\gadget]{\{\literal_{j,1}', \literal_{j,2}',
        \literal_{j,3}'\}}$,
    assume
    $\literal_{j,1}$ is assigned \true, and each of $\literal_{j,2}$,
    $\literal_{j,3}$ is assigned \true or \false.
    Consider the following strategy to place three pebbles on $\{ a_j, b_j, c_j \}$,
    where pebbles outside of $R_j$ are not counted:
    \begin{enumerate}
      \item Persistently pebble the \highway from $\literal_{j,2}$ to $b_j$.
        It takes at most $\beta_j + 2$ pebbles over $R_j$ by
        \reflem{lem:Highway}.
      \item Persistently pebble the \highway from $\literal_{j,3}$ to $c_j$.
        Now only $b_j$ and $c_j$ have pebbles over $R_j$.
        It takes at most $\underbrace{\beta_j + 2}_{\ancnode{c_j}} +
        \underbrace{\vphantom{\beta_j}1}_{b_j} = \beta_j + 3$ pebbles over $R_j$ by
        \reflem{lem:Highway}.
      \item Persistently pebble the \highway from $\literal_{j,1}$ to $a_j$.
        Since $\literal_{j,1}$ is outside of $R_j$, it takes at most $\beta_j
        + 1$ pebbles over the intersection of the \highway and $R_j$ by
        \reflem{lem:Highway}, for a total of $\underbrace{\beta_j +
          1}_{\ancnode{a_j}} + \underbrace{\vphantom{\beta_j}2}_{b_j,c_j} = \beta_j + 3$
        pebbles.
        Now $a_j$, $b_j$ and $c_j$ have pebbles.  \qedhere
    \end{enumerate}
  \end{proof}

  \begin{lemma}[False Clause, lower bound]
    \label{lem:FalseClause}
    Fix a partial assignment $\partassign$.
    Assume all of the literals $\literal_{j,1}$, $\literal_{j,2}$ and
    $\literal_{j,3}$ are assigned \false.
    Let 
    $\canonical_j \defeq \{ \literal_{j,1}', \literal_{j,2}',
    \literal_{j,3}' \}$  be their canonical nodes.
    Consider the region $R_j \defeq \ancnode[\gadget]{p_j} \setminus
    \ancnode[\gadget]{\canonical_j}$ beyond the canonical nodes.
    Then $\surroundprice[R_j]{p_j} \ge \beta_j + 4$.
  \end{lemma}
  \begin{proof}
    Note that $\gadget$ consists of a pyramid whose sources are attached to three
    \highway{}s.
    Since all literals are assigned \false, the intersection of $R_j$ with
    each of the \highway is precisely the nodes in the \highway.

    Fix an induced subgraph $F \subseteq \gadget$ (\eg $F =$ the \highway from
    $\literal_{j,1}$ to $a_j$) having a unique sink.
    Say a pebbling configuration on $F$ is \introduceterm{\vislocked} if 
    $\persistentprice[F]{\pconf}=\visitprice{F}$, that is if
    in order to reach the empty configuration it is necessary to pass
    through a configuration with $\visitprice{F}$ pebbles.
    In particular, any configuration with a pebble on the sink of $F$ is
    \vislocked on $F$.
    Also, if a configuration is \vislocked on $F$, then there is a pebble on $F$.
    Given a pebbling configuration on $\gadget$, say $a_j$ (resp.~$b_j$,
    $c_j$) is \vislocked if the configuration is \vislocked on the \highway from
    $\literal_{j,1}$ to $a_j$ (resp.~from $\literal_{j,2}$ to $b_j$, from
    $\literal_{j,3}$ to $ c_{j} $).

    With locking in mind, consider a ``projected'' configuration on the
    pyramid defined as follows.
    Given a pebbling configuration $\pconf$ on $\gadget$, its projection to the
    pyramid is $\projection{\pconf}
    \defeq \bigl( \{ u_j, v_j, p_j \} \intersect \pconf \bigr)
    \union \bigl\{ t \in \{ a_j, b_j, c_j \} \where t\text{ is \vislocked under }\pconf
      \bigr\}$.
    Note that given a strategy on $\gadget$, its (configuration-wise)
    projection to the pyramid is a legal strategy on the pyramid.

    We are interested in the truncated paths $\breve\pi$ on the pyramid (\ie
    source to sink paths excluding the sink, which are in bijection to the
    edges $(a_j, u_j)$, $(b_j, u_j)$, $(b_j, v_j)$, $(c_j, v_j)$), and in
    particular whether they are blocked under $\projection{\pconf}$.
    A truncated path $\breve\pi$ is \emph{blocked} under $\projection{\pconf}$ if
    $\breve\pi \intersect \projection{\pconf} \ne \emptyset$.

    Consider a strategy on $\gadget$ to surround $p_j$.
    Its projection to the pyramid is a strategy on the pyramid to surround
    $p_j$.
    At the beginning, all truncated paths on the pyramid are not blocked;
    at the end, all truncated paths are blocked.
    Consider the first time that all truncated paths on the pyramid are
    blocked:
    in the strategy projected on the pyramid, this must be the result of
    pebbling a source node $a_j$, $b_j$, or $c_j$.
    By symmetry (in the rest of this argument), assume $a_j$ is pebbled, then
    there are at least two more pebbles on the pyramid in the projected
    configuration.
    Since $a_j$ is being pebbled in the projected strategy, $a_j$ is getting
    \vislocked on $\gadget$ (\ie restricting attention to the \highway from
    $\literal_{j,1}$ to $a_j$, there are as many pebbles as the visiting
    price of the \highway), accounting for $\beta_j + 2$ pebbles in the
    intersection of $R_j$ and the \highway to $a_j$ by \reflem{lem:Highway}.
    The two other pebbles in the projected strategy each account for one more
    pebble over $R_j$, for a total of $\beta_j + 4$ pebbles.
  \end{proof}

  The lower bound shown in \reflem{lem:FalseClause} holds for
  clauses which are falsified. We now prove a weaker lower bound on
  the surrounding price for the satisfied clauses, which matches the
  upper bound.

  \begin{lemma}[Any Clause, lower bound]
    \label{lem:AnyClause}
    Fix a partial assignment $\partassign$.
    Assume none of literals $\literal_{j,1}$, $\literal_{j,2}$, or
    $\literal_{j,3}$ is assigned \tvundef.
    Let $\canonical_j$ be their canonical nodes.
    Consider the region $R_j \defeq \ancnode[\gadget]{p_j} \setminus
    \ancnode[\gadget]{\canonical_j}$ beyond the canonical nodes.
    Then $\surroundprice[R_j]{p_j} \ge \beta_j + 3$.
  \end{lemma}
  \begin{proof}
    Follow the proof of \reflem{lem:FalseClause} to define \vislocked, projection
    to the pyramid, truncated paths on the pyramid, and blocking on the
    pyramid.
    The only difference is that, since some literal can be assigned \true,
    the intersection of $R_j$ with some of the \highway can be the nodes
    \emph{properly} in the \highway.

    Consider a strategy on $\gadget$ to surround $p_j$.
    Its projection to the pyramid is a strategy on the pyramid to surround
    $p_j$.
    As in the proof of \reflem{lem:FalseClause}, consider the first time that
    all truncated paths on the pyramid are blocked, which in the projected
    strategy must be the result of pebbling a source node, say, $a_j$.
    And there are at least two more pebbles on the pyramid in the projected
    configuration.
    Since $a_j$ is being pebbled in the projected strategy, $a_j$ is getting
    \vislocked on $\gadget$ (\ie restricting attention to the \highway from
    $\literal_{j,1}$ to $a_j$, there are as many pebbles as the visiting
    price of the \highway), accounting for $\beta_j + 1$ pebbles in the
    intersection of $R_j$ and the \highway to $a_j$ by \reflem{lem:Highway} (note
    that node $\literal_{j,1}$ may be outside of $R_j$ if literal
    $\literal_{j,1}$ is in \true position).
    The two other pebbles in the projected strategy each account for one more
    pebble over $R_j$, for a total of $\beta_j + 3$ pebbles.
  \end{proof}

  Assuming that the literal gadgets are in the position corresponding
  to $ \partassign $, we represent ``whether $\partassign$ falsifies
  $\clause_j$'' through \emph{an increase} in the persistent price of the sink
  of the corresponding gadget,
  \ie
  if $\partassign$ \emph{satisfies} $\clause_j$, then the persistent
  price would be a certain number ($\beta_j + 4$);
  but
  if $\partassign$ \emph{falsifies} $\clause_j$, then the persistent price
  would be one plus that number ($\beta_j + 5$).
  The difference in pebbling prices can be succintly expressed using
  the Iverson bracket notation $\ibleft \partassign \text{ falsifies
  }\clause_j \ibright$.

  \begin{corollary}[Clause Gadget]
    \label{cor:ClauseGadget}
    Fix a partial assignment $\partassign$.
    Assume none of literals $\literal_{j,1}$, $\literal_{j,2}$, or
    $\literal_{j,3}$ is assigned \tvundef.
    Let $\canonical_j$ be their canonical nodes.
    Consider the region $R_j \defeq \ancnode[\gadget]{p_j} \setminus
    \ancnode[\gadget]{\canonical_j}$ beyond the canonical nodes.
    Then 
    $\persistentprice[R_j]{p_j} = 
    \beta_j + 4 + 
    \ibleft 
    \partassign \text{\ falsifies\ }\clause_j \ibright$.
  \end{corollary}
  \begin{proof}
    If $\partassign$ satisfies $\clause_j$, \ie at least one literal is
    assigned \true, then $\persistentprice[R_j]{p_j} = \beta_j + 4$ because
    persistent price is one plus surrounding price
    (\refpr{lem:surrounding-eq-persistent}), and the upper bound
    (\reflem{lem:TrueClause}) matches the lower bound (\reflem{lem:AnyClause}).
    If $\partassign$ falsifies $\clause_j$, \ie all literals are assigned
    \false, then $\persistentprice[R_j]{p_j} = \beta_j + 5$ because it has to
    increase (\reflem{lem:FalseClause}) but not by more than one
    (\reflem{lem:SourceDifference}).
  \end{proof}

\subsection{Conjunction Gadget}
\label{ssst:ConjunctionGadget}

  To construct a gadget for the conjunction of $m$ clauses, it suffices to
  repeatedly compose a gadget for the conjunction two smaller gadgets, using
  the \emph{conjunction gadget} represented in \reffig{fig:GadgetConjunction}.

  \begin{construction}[Conjunction Gadget]
    \label{con:ConjunctionGadget}
    Assume two gadgets $G_1$ and $G_2$ with unique sinks are constructed.
    Construct the \emph{conjunction gadget} of weight $\bank$ of $G_1$ and
    $G_2$, denoted $\conjConst{\bank}{G_1}{G_2}$, as follows.
    Call $\sink_1$ the sink of $G_1$, $\sink_2$ the sink of $G_2$.
    Construct nodes $d_1$, $d_2$, $d_3$, $d_4$, $e$, and edges $(d_1, d_3)$,
    $(d_2, d_3)$, $(d_4, e)$.
    Add a \highway of toll $\bank$ from $\sink_1$ to $d_1$, a \highway of toll
    $\bank-1$ from $\sink_2$ to $d_2$, and a \highway of toll $\bank-2$ from
    $d_3$ to $d_4$.
  \end{construction}

\begin{figure}[ht]
  \begin{center}
    \begin{tikzpicture}[rotate=-90]
      \GadgetConjunction{\scriptstyle\bank}{center}{minimum size=1.2cm}
    \end{tikzpicture}
  \end{center}
  \caption{Conjunction gadget of weight $\bank$ of $G_1$ and $G_2$.}
  \label{fig:GadgetConjunction}
\end{figure}
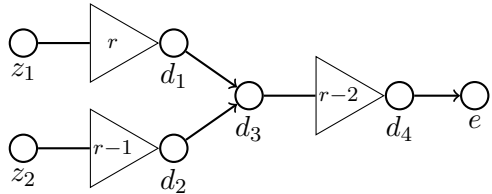

  For the gadget $G \defeq \conjConst{\bank}{G_1}{G_2}$, the nodes in the
  gadget are $V(G) = \ancnode{e} \setminus \bigl( \propancnode{\sink_1} \union
  \propancnode{\sink_2} \bigr)$.
  We want to analyze pebbling prices restricted to a certain region $\Region$
  (in the final gadget containing the conjunction gadgets) which will be a
  superset of the nodes of $G$.

  \begin{lemma}[True Conjunction]
    \label{lem:TrueConjunction}
    Fix a region $\Region$ where $\Region \supseteq V(G)$.
    If $\persistentprice[\Region]{\sink_1} \le \bank + 1$ and
    $\persistentprice[\Region]{\sink_2} \le \bank$, then
    $\surroundprice[\Region]{e} \le \bank + 2$.
  \end{lemma}
  \begin{proof}
    Consider the following strategy to visit $d_4$ (equivalently, surround
    $e$) using at most $\bank+2$ pebbles:
    \begin{itemize}
      \item Persistently pebble $\sink_1$, persistently pebble the \highway
        from $\sink_1$ to $d_1$, persistently unpebble $\sink_1$.
        Now $d_1$ has a pebble.
        Over $\Region$, the first sub-step takes at most $\bank + 1$ pebbles,
        the second sub-step takes at most $\bank + 2$ pebbles by
        \reflem{lem:Highway}, and the third sub-step takes at most
        $\underbrace{(\bank + 1)}_{\ancnode{\sink_1}} + \underbrace{\vphantom{()}1}_{d_1}
        = \bank + 2$ pebbles.
      \item Persistently pebble $\sink_2$, persistently pebble the \highway
        from $\sink_2$ to $d_2$, persistently unpebble $\sink_2$.
        Now $d_1$ and $d_2$ have pebbles.
        Over $\ancnode{d_2} \intersect \Region$, the first sub-step takes
        at most $\bank$ pebbles, the second sub-step most $(\bank - 1) + 2 =
        \bank + 1$ pebbles by \reflem{lem:Highway}, and the third sub-step takes
        at most $\underbrace{\bank}_{\ancnode{\sink_2}} +
        \underbrace1_{d_2}$ pebbles.
        Over $\Region$, this step takes at most
        $\underbrace{\bank+1}_{\ancnode{d_2}} + \underbrace1_{d_1} = \bank
        + 2$ pebbles.
      \item Pebble $d_3$, persistently pebble the \highway from $d_3$ to $d_4$.
        Now $d_1$, $d_2$, $d_3$, $d_4$ have pebbles.
        Over $\Region' \defeq \bigl( \ancnode{e} \intersect \Region \bigr)
        \setminus \bigl( \ancnode{d_1} \union \ancnode{d_2} \bigr)$,
        the first sub-step takes 1 pebble, the second sub-step takes at most
        $(\bank - 2) + 2 = \bank$ pebbles by \reflem{lem:Highway}.
        Over $\Region$, this step takes at most $\underbrace\bank_{\Region'} +
        \underbrace2_{d_1, d_2}$ pebbles.  \qedhere
    \end{itemize}
  \end{proof}

  \begin{lemma}[Any Conjunction]
    \label{lem:AnyConjunction}
    Fix a region $\Region$ where $\Region \supseteq V(G)$.
    We have the following:
    \begin{enumerate}
      \item $\surroundprice[\Region]{e} \ge \bank + 2$; and
      \item if $\persistentprice[\Region]{\sink_1} \le \bank + 2$ and
        $\persistentprice[\Region]{\sink_2} \le \bank + 1$, then
        $\surroundprice[\Region]{e} \le \bank + 3$.
    \end{enumerate}
  \end{lemma}
  \begin{proof}
    For Item~(1), note that any strategy to surround $e$ must visit $d_1$, and
    $\visitprice[\Region]{d_1} \ge \bank + 2$ by \reflem{lem:Highway}.

    Item~(2) follows from the proof of \reflem{lem:TrueConjunction} to visit
    $d_4$ (equivalently, surround $e$):
    \begin{enumerate}
      \item persistently pebble $\sink_1$, persistently pebble the \highway
        from $\sink_1$ to $d_1$, persistently unpebble $\sink_1$.
        Now $d_1$ has a pebble.
        This step takes at most $\bank+3$ pebbles over $\Region$.
      \item persistently pebble $\sink_2$, persistently pebble the \highway
        from $\sink_2$ to $d_2$, persistently unpebble $\sink_2$.
        Now $d_1$ and $d_2$ have pebbles.
        This step takes at most $\bank+3$ pebbles over $\Region$.
      \item pebble $d_3$, persistently pebble the \highway from $d_3$ to $d_4$.
        Now $d_1$, $d_2$, $d_3$, $d_4$ have pebbles.
        This step takes at most $\bank+2$ pebbles over $\Region$. \qedhere
    \end{enumerate}
  \end{proof}

  \begin{lemma}[False Conjunction]
    \label{lem:FalseConjunction}
    Fix a region $\Region$ where $\Region \supseteq V(G)$.
    We have the following:
    \begin{enumerate}
      \item if $\persistentprice[\Region]{\sink_1} \ge \bank + 2$, then
        $\surroundprice[\Region]{e} \ge \bank + 3$; and
      \item if $\persistentprice[\Region]{\sink_2} \ge \bank + 1$, then
        $\surroundprice[\Region]{e} \ge \bank + 3$.
    \end{enumerate}
  \end{lemma}
  \begin{proof}
    Fix a \highway $\highwaynot \subseteq G$, say from $s$ to $\sink$ (\eg if $\highwaynot$ is the
    \highway from $\sink_1$ to $d_1$, then $s = \sink_1$ and $\sink = d_1$).
    Say a pebbling configuration on $\highwaynot$ is \introduceterm{\srclocked} if the pebbles cannot
    be cleared without using $\visitprice[\Region]{\highwaynot}$ pebbles including one on
    $s$ (when restricted to $\highwaynot$);
    that is, if the empty configuration cannot be reached without entering a
    configuration which has $\visitprice[\Region]{\highwaynot}$ pebbles and contains $s$.
    In particular, any configuration with a pebble on the sink of $\highwaynot$ is
    \srclocked on $\highwaynot$ by \reftwolems{lem:Molding}{lem:Highway}.
    Also, if a configuration is \srclocked on $\highwaynot$, then there is a pebble on some
    node \emph{properly} in the \highway (there is a pebble on $V(\highwaynot) \setminus
    \{ s \}$).
    Given a pebbling configuration on $G$, say $d_1$ (resp.~$d_2$, $d_4$) is
    \srclocked if the configuration is \srclocked on on the \highway from
    $\sink_1$ to $d_1$ (resp.~from $\sink_2$ to $d_2$, from $d_3$ to $d_4$).

    Fix an induced subgraph $F \subseteq G$ having a unique sink $v$ (\eg $F =
    \ancnode{\sink_1}$ and $v = \sink_1$).
    Say a pebbling configuration on $F$ is \introduceterm{\perlocked} if 
$\persistentprice[F]{\pconf}=\persistentprice[\Region]{F}$; this is if the pebbles cannot
    be cleared without using $\persistentprice[\Region]{F}$ pebbles (when
    restricted to $F$).
    In particular, the configuration with just a single pebble on $v$ over
    $\Region$ is \perlocked on $F$.
    Also, if a configuration is \perlocked on $F$, then there is a pebble on $F
    \intersect \Region$.
    Given a pebbling configuration on $G$, say $\sink_1$ (resp.~$\sink_2$) is
    \introduceterm{\perlocked} if the configuration is \perlocked on $\ancnode{\sink_1}$
    (resp.~$\ancnode{\sink_2}$).

    \begin{claim}[\srclocked implies \perlocked]
      \label{cla:JammingAndLocking}
      Fix any \highway $\highwaynot$ on $G$, say of toll $r$ from $s$ to $\sink$.
      Assume $\persistentprice[\Region]{s} \ge r + 2$.
      In any pebbling that uses at most $r+2$ pebbles over $\ancnode{\sink} \intersect
      \Region$, if $\sink$ is \srclocked then $s$ is \perlocked.
    \end{claim}
    \begin{proof}
      Assume $\sink$ starts to get \srclocked.
      By definition of \srclocked, there are $r + 2$ pebbles on the \highway $\highwaynot$,
      and $s$ has one of the pebbles.
      Since at most $r + 2$ pebbles are used over $\Region$, only $s$ has
      pebble over $\ancnode{s} \intersect \Region $, so $s$ is \perlocked.
      Until $\sink$ is not \srclocked, there is one pebble \emph{properly} in the
      \highway (\ie over $V(\highwaynot) \setminus \{ s \}$).
      Since unlocking $s$ requires $\persistentprice[\Region]{s} \ge r + 2$
      pebbles over $\ancnode{s} \intersect \Region$, until $\sink$ stops being
      \srclocked, $s$ remains \perlocked.
    \end{proof}

    Fix a strategy to surround node $e$, which at some time $t_3$ must pebble
    or unpebble $d_3$.
    At time $t_3$, both node $d_1$ and node $d_2$ have pebbles, hence both
    node $d_1$ and node $d_2$ are \srclocked.
    At the beginning, both node $d_1$ and node $d_2$ are not \srclocked.
    Let $t_1$ (resp.~$t_2$) be the earliest time before time $t_3$ such that
    node $d_1$ (resp.~$d_2$) remains \srclocked between time $t_1$ and time $t_3$.
    Thus node $d_1$ (resp.~$d_2$) is \srclocked at time $t_1$ (resp.~$t_2$).

    \begin{claim}[\srclocked]
      \label{cla:Jamming}
      If the pebbling uses at most $\bank+2$ pebbles over $\Region$ to surround
      $e$, then
      \begin{enumerate}[(i)]
        \item $t_1 < t_2$; and
        \item at time $t_2$, there is exactly one pebble over
          $\ancnode{d_1} \intersect \Region$.
      \end{enumerate}
    \end{claim}
    \begin{proof}
      For Item~(i), if $t_2 < t_1$, then at time $t_1$ there is a pebble on
      the \highway from $\sink_2$ to $d_2$ (as node $d_2$ is already \srclocked),
      and there are $\bank + 2$ pebbles on the \highway from $\sink_1$ to $d_1$
      by definition of \srclocked, for a total of $\bank + 3$ pebbles over
      $\Region$.

      For Item~(ii), there is at least one pebble on the \highway from
      $\sink_1$ to $d_1$ as node $d_1$ is already \srclocked, and there is at most
      one pebble over $\ancnode{d_1} \intersect \Region$ since there are
      at least $(\bank-1) + 2 = \bank + 1$ pebbles on the \highway from
      $\sink_2$ to $d_2$ by definition of \srclocked, and we assumed that at most
      $\bank+2$ pebbles are used over $\Region$.
    \end{proof}

    For Item~(1), assume $\persistentprice[\Region]{\sink_1} \ge \bank + 2$.
    If at most $\bank + 2$ pebbles are used over $\Region$ to surround node $e$,
    then \refclaim{cla:Jamming} shows that at time $t_2$, node $d_1$ is \srclocked (as
    $t_1 < t_2$), and there is exactly one pebble on $\ancnode{d_1}
    \intersect \Region$.
    Since $\persistentprice[\Region]{\sink_1} \ge \bank + 2$ and the \highway
    from $\sink_1$ to $d_1$ has toll $\bank$, \refclaim{cla:JammingAndLocking}
    says that when node $d_1$ is \srclocked (which is the case at time $t_2$),
    node $\sink_1$ is \perlocked.
    Therefore at time $t_2$, there are two pebbles over $\ancnode{d_1}
    \intersect \Region$: one \emph{properly} on the \highway from $\sink_1$ to
    $d_1$ (since $d_1$ is \srclocked); and 
    one on $\ancnode{\sink_1} \intersect \Region$ (since $\sink_1$ is
    \perlocked).
    This contradiction shows that $\surroundprice[\Region]{e} \ge \bank + 3$.

    For Item~(2), assume $\persistentprice[\Region]{\sink_2} \ge \bank + 1$.
    Fix a strategy using at most $\bank+2$ pebbles over $\Region$ to surround
    node $e$, \ie to visit node $d_4$.
    At the end, node $d_4$ is \srclocked;
    at the beginning, node $d_4$ is not \srclocked.
    Let $t_4$ be the earliest time such that node $d_4$ remains \srclocked since
    $t_4$ until the end.
    Thus node $d_4$ is \srclocked at time $t_4$.

    Redefine $t_3$ if necessary, assume it is the last time before $t_4$ such
    that node $d_3$ is pebbled or unpebbled.
    Then time $t_1$ and time $t_2$ are defined (as above) relative to this
    $t_3$, giving $t_1 < t_2 < t_3 < t_4$ (the first inequality is by
    \refclaim{cla:Jamming}).

    Note that at time $t_3$, node $d_3$ is being pebbled:
    to see this, we know that at time $t_4$, the \highway from $d_3$ to $d_4$
    is being \srclocked, so there is a pebble on node $d_3$ by definition of
    \srclocked.
    Since there is no pebble move on node $d_3$ after time $t_3$ and before
    time $t_4$, it follows that $d_3$ is being pebbled at time $t_3$, and
    there is a pebble on node $d_3$ between time $t_3$ and time $t_4$.

    We know that both node $d_1$ and node $d_2$ are \srclocked at time $t_3$.
    In fact, they remain \srclocked between time $t_3$ and time $t_4$:
    to see this, note that to make node $d_1$ not \srclocked takes $\bank + 2$ pebbles over
    $\ancnode{d_1} \intersect \Region$, but there are two pebbles outside
    this region (one on the \highway from $\sink_2$ to $d_2$, and one on $d_3$),
    which cannot be done with at most $\bank + 2$ pebbles over $\Region$.
    Likewise, note that to make node $d_2$ not \srclocked takes $(\bank-1) + 2 = \bank + 1$
    pebbles over $\ancnode{d_2} \intersect \Region$, but there are two
    pebbles outside this region (one on the \highway from $\sink_1$ to $d_1$,
    and one on $d_3$), which cannot be done with at most $\bank + 2$ pebbles
    over $\Region$.
    Therefore, node $d_1$ is \srclocked from time $t_1$ to time $t_4$, and node
    $d_2$ is \srclocked from time $t_2$ to time $t_4$.

    At time $t_1$, the \highway from $\sink_1$ to $d_1$ is getting \srclocked, so
    by definition of \srclocked there are $\bank + 2$ pebbles on the \highway.
    By assumption, at most $\bank + 2$ pebbles are used over $\Region$, so
    there is no pebble over $\ancnode{d_2} \intersect \Region$ at time
    $t_1$.
    Restrict attention to the sub-strategy $\sP'$ between time $t_1$ and time
    $t_4$.
    In the sub-strategy $\sP'$, node $d_1$ remains \srclocked, so at most
    $(\bank+2) - 1 = \bank + 1$ pebbles can be used over $\ancnode{d_2}
    \intersect \Region$.
    Since $\persistentprice[\Region]{\sink_2} \ge (\bank-1) + 2$ and the \highway
    from $\sink_1$ to $d_1$ has toll $\bank-1$, \refclaim{cla:JammingAndLocking}
    says that when node $d_2$ is \srclocked, node $\sink_2$ is \perlocked.
    At time $t_4$,
    \begin{itemize}
      \item node $d_1$ is \srclocked, so there is a pebble properly in the \highway
        from $\sink_1$ to $d_1$;
      \item node $d_2$ is \srclocked, hence node $\sink_2$ is \perlocked, so there is
        a pebble properly in the \highway from $\sink_2$ to $d_2$, and a pebble
        over $\ancnode{\sink_2} \intersect \Region$; and
      \item there are $(\bank-2) + 2 = \bank$ pebbles in the \highway from
        $d_3$ to $d_4$.
    \end{itemize}
    This accounts for $1 + 2 + \bank = \bank + 3$ pebbles over $\Region$.
    This contradiction shows that $\surroundprice[\Region]{e} \ge \bank + 3$.
  \end{proof}

  Recall that we are going to represent the unsatisfiability of a clause by
  increased persistent prices, \ie let $\condition_j$ be the condition that
  clause $\clause_j$ is \emph{satisfied}, and $\bar\condition_j$ be its
  negation, then $\persistentprice[\Region_j]{p_j} = \beta_j + 4 + \ib{
    \bar\condition_j }$ by \refcor{cor:ClauseGadget}.

  \begin{corollary}[Conjunction Gadget]
    \label{cor:ConjunctionGadget}
    Fix a region $\Region$ where $\Region \supseteq V(G)$.
    Assume that for some conditions $\condition_1$ and $\condition_2$
    if holds that
    $\persistentprice[\Region]{\sink_1} = \bank + 1 + \ib{\bar\condition_1}$
    and
    $\persistentprice[\Region]{\sink_2} = \bank + \ib{\bar\condition_2}$.
    Then $\persistentprice[\Region]{e} = \bank + 3 +
    \ib{\overline{\condition_1 \land \condition_2}}$.
  \end{corollary}
  \begin{proof}
    If both $\condition_1$ and $\condition_2$ are \true, then
    $\persistentprice[\Region]{e} = \bank + 3$ because persistent price is
    one plus surround price (\refpr{lem:surrounding-eq-persistent}), and the
    upper bound (\reflem{lem:TrueConjunction}) matches the lower bound (Item~(1)
    of \reflem{lem:AnyConjunction}).
    If $\condition_1$ or $\condition_2$ is \false, then
    $\persistentprice[\Region]{e} = \bank + 4$ because the lower bound
    increases (\reflem{lem:FalseConjunction}) to match the new upper bound
    (Item~(2) of \reflem{lem:AnyConjunction}).
  \end{proof}

\subsection{CNF Gadget}
\label{ssst:CNFGadget}

  Assume $\CNF = \clause_1 \land \clause_2 \land \dotsb \land \clause_m$ is a
  conjunction of $m$ clauses.
  Let $\CNF_k \defeq \Land_{1 \le j \le k} \clause_j$ be the conjunction of
  the first $k$ clauses.
  We will construct a gadget $F_k \defeq \gadgetConst{\CNF_k}$ for $\CNF_k$
  with increasing $k$ by successive conjunction of two smaller gadgets,
  then the CNF gadget for $\CNF$ is $\gadgetConst{\CNF} =
  \gadgetConst{\CNF_m}$.

  For $1 \le j \le m$, let $\beta_j \defeq 2j$, then $\beta_j \ge 2$.

  \begin{construction}[CNF Gadget]
    \label{con:CNFGadget}
    Assume for each clause $\clause_j$, $1 \le j \le m$, a clause gadget
    $\gadgetConst{\clause_j}$ is constructed.
    Let $\CNF_k \defeq \Land_{1 \le j \le k} \clause_j$ be the conjunction of
    the first $k$ clauses when $0 \le k \le m$.
    Construct a partial CNF gadget $F_k$ for increasing $k$ as follows.
    Construct $F_0 \defeq T_7$ as the graph with $\visitprice{T_7} =
    \persistentprice{T_7} = 7$ given by \reflem{lem:ChristmasTree}.
    Then for $1 \le k \le m$, construct $F_k \defeq
    \ConjConst{\beta_k}{F_{k-1}}{\gadgetConst{\clause_k}}$ be the conjunction
    gadget of weight $\beta_k$ of the previous partial CNF gadget ($F_{k-1}$)
    and the gadget of clause $k$ ($\gadgetConst{\clause_k}$).
    The CNF gadget $\gadgetConst{\CNF}$ for $\CNF$ is $F_m$.
  \end{construction}

  For the gadget $G \defeq \gadgetConst{\CNF}$, the nodes in the gadget $V(G)$
  contains the nodes of $F_0 = T_7$, the nodes of all clause gadgets
  $\gadgetConst{C_j}$, and the nodes of all intermediate conjunction gadgets.

  \begin{lemma}[CNF Gadget]
    \label{lem:CNFGadget}
    Let $\partassign$ be an assignment, and let $\canonical$ be the
    canonical nodes in all variable gadgets according to $\partassign$
    (\refdef{def:CanonicalNodes}).
    Let $\Region \defeq V(G) \setminus \ancnode{\canonical}$ be the region
    beyond the canonical nodes of $\partassign$.
    Then $\persistentprice[\Region]{F_m} = \beta_m + 7 +
    \ib{\partassign\text{\ falsifies\ }\Gamma}$.
  \end{lemma}
  \begin{proof}
    For $1 \le j \le m$, let $\condition_j$ be the condition that clause $j$
    ($\clause_j$) is satisfied by $\partassign$.
    For $0 \le k \le m$, let $\condition'_k$ be the condition that the first
    $k$ clauses ($\CNF_k$) are satisfied by $\partassign$.
    We show by induction that $\persistentprice[\Region]{F_k} = \beta_k + 7 +
    \ib{\bar\condition'_k}$.

    When $k = 0$, $\CNF_k$ is satisfied vacuously, so $\condition'_0$ is
    \true.
    The base case holds as $\persistentprice[\Region]{F_0} =
    \persistentprice{T_7} = 7$.

    For the general case $1 \le k \le m$, plug $\bank = \beta_k + 4$ into
    \refcor{cor:ConjunctionGadget}.
    Since induction hypothesis gives $\persistentprice[\Region]{ F_{k-1} } =
    \beta_{k-1} + 7 + \ib{\bar\condition'_{k-1}} = \beta_k + 5 +
    \ib{\bar\condition'_{k-1}}$, and
    \refcor{cor:ClauseGadget} gives $\Persistentprice[\Region]{
      \gadgetConst{\clause_k} } = \beta_k + 4 + \ib{\bar\condition_k}$,
    it follows that $\persistentprice[\Region]{ F_k } = \beta_k + 7 +
    \ib{\bar\condition'_k}$, because $\condition'_k = \condition'_{k-1} \land
    \condition_k$.
  \end{proof}

\subsection{Existential Quantifier Gadget}
\label{ssst:ExistentialQuantifierGadget}

  Assume that we already have the gadget $\gadgetConst{\QBF_{i-1}}$ and
  that the $i\th$ inner-most quantifier is existential, \ie
  $\quantifier_i = \exists$.
  This quantifier refers to $\variable_{i}$, we set the
  parameter $\bank_i \defeq \gamma_i - 2$ for the corresponding variable gadget.
  We construct $\gadgetConst{\QBF_i}$ as follows.

  \begin{construction}[Existential Quantifier Gadget]
    \label{con:ExistentialQuantifierGadget}
    Let $q_{i-1}$ denote the sink of $\gadgetConst{\QBF_{i-1}}$.
    Construct nodes $f_i$, $g_i$, $q_i$, and edges $(\variable_i, g_i)$,
    $(\bar\variable_i, g_i)$, $(f_i, q_i)$, $(g_i, q_i)$.
    Add a \highway of toll $\gamma_i - 5$ from $q_{i-1}$ to $f_i$.
  \end{construction}

\begin{figure}[ht]
  \begin{center}
    \begin{tikzpicture}
      \GadgetExistentialQuantifierSingle{i}{i-1}{\bank_i}{%
        \scriptstyle\gamma_i - 5}{center}{minimum size=1.4cm}
    \end{tikzpicture}
  \end{center}
  \caption{Existentially quantified variable $\exists \variable_i$.}
  \label{fig:ExistentialQuantifier}
\end{figure}
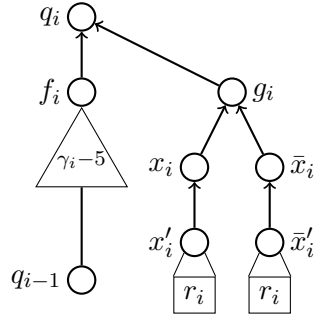

  For the gadget $\gadget \defeq \gadgetConst{\QBF_i}$, the nodes in the
  gadget $V(\gadget)$ contains the nodes in the previous gadget $V\bigl(
  \gadgetConst{\QBF_{i-1}} \bigr)$, the new nodes $f_i, g_i, q_i$ and the
  nodes in the \highway from $q_{i-1}$ to $f_i$.

  \begin{lemma}[Existential Upper Bound]
    \label{lem:ExistentialUpperBound}
    Assume \reflem{lem:QuantifierGadget} holds for $i-1$.
    \begin{enumerate}
      \item If $\QBF_i\restriction_\partassign$ is \true, then
        $\surroundprice[\Region]{q_i} \le \gamma_i - 1$.
      \item If $\QBF_i\restriction_\partassign$ is \false, then
        $\surroundprice[\Region]{q_i} \le \gamma_i$.
    \end{enumerate}
  \end{lemma}
  \begin{proof}
    For Item~(1), by assumption $\QBF_i\restriction_\partassign\ =
    \exists\variable_i \QBF_{i-1}\restriction_\partassign$ is \true, so there
    is an assignment of $\variable_i$ to \true or to \false to satisfy
    $\QBF_{i-1}$, \ie
    $\QBF_{i-1}\restriction_{\partassign_1}$ is \true or
    $\QBF_{i-1}\restriction_{\partassign_0}$ is \true.
    Assume the former by symmetry.
    Since $\QBF_{i-1}\restriction_{\partassign_1}$ is \true, by the assumption
    that \reflem{lem:QuantifierGadget} holds for $i-1$, we have
    $\persistentprice[\Region_1]{ q_{i-1} } = \gamma_{i-1} = \gamma_i - 3$.
    Consider the following strategy to surround $q_i$ with at most $\gamma_i -
    1$ pebbles over $\Region$:
    \begin{enumerate}[(i)]
      \item Put $\variable_i$ into the canonical \true position with $\bank_i
        + 1 = \gamma_i - 1$
        pebbles by \reflem{lem:VariableAssignment}.
        Now node $\variable_i$ and node $\bar\variable_i'$ have pebbles;
      \item Persistently pebble node $q_{i-1}$ using at most $\gamma_i - 3$
        pebbles over the new region $\Region_1$.
        Now $\variable_i$, $\bar\variable_i'$ and $q_{i-1}$ have pebbles.
        Over the old region $\Region$, at most $(\gamma_i - 3) + 2 = \gamma_i
        - 1$ pebbles are used (this step is legal as $\Region = \Region_1
        \union \bigl( \ancnode{ \variable_i } \union \ancnode{
        \bar\variable_i' } \bigr)$ and $\variable_i$ and $\bar\variable_i$
        have pebbles);
      \item Persistently pebble the \highway from $q_{i-1}$ to $f_i$.
        Now $\variable_i$, $\bar\variable_i'$, $q_{i-1}$ and $f_i$ have pebbles.
        At most $(\gamma_i - 5) + 2 = \gamma_i - 3$ pebbles are used over the
        \highway (including the pebble on node $q_{i-1}$) by
        \reflem{lem:Highway}, so at most $(\gamma_i - 3) + 2 = \gamma_i - 1$
        pebbles are used over $\Region$; and
      \item Pebble node $\bar\variable_i$ then node $g_i$.
        Now the six nodes $\variable_i$, $\bar\variable_i'$,
        $\bar\variable_i$, $q_{i-1}$, $f_i$, $g_i$ have pebbles.
    \end{enumerate}

    For Item~(2),
    since $\QBF_{i-1}\restriction_{\partassign_1}$ is \false, by the assumption
    that \reflem{lem:QuantifierGadget} holds for $i-1$, we have
    $\persistentprice[\Region_1]{ q_{i-1} } = \gamma_{i-1} + 1
    = (\gamma_i - 3) + 1
    = \gamma_i - 2$.
    We run the same strategy as in Item~(1) to surround $q_i$, using at most
    $\gamma_i$ pebbles over $\Region$ (only Step~(ii) uses one more pebble).
  \end{proof}

  \begin{lemma}[Existential Lower Bound]
    \label{lem:ExistentialLowerBound}
    Assume \reflem{lem:QuantifierGadget} holds for $i-1$.
    \begin{enumerate}
      \item $\surroundprice[\Region]{q_i} \ge \gamma_i - 1$.
      \item If $\surroundprice[\Region]{q_i} \le \gamma_i - 1$, then
        $\QBF_i\restriction_\partassign$ is \true.
    \end{enumerate}
  \end{lemma}
  \begin{proof}
    Fix a strategy to surround $q_i$ using at most $\gamma_i - 1$ pebbles over
    $\Region$.
    At the end, there is a pebble on the \highway from $q_{i-1}$ to $f_i$, and
    there is a pebble on $g_i$.
    Let $t_2$ be the earliest time such that since $t_2$ there is at least one
    pebble on the \highway from $q_{i-1}$ to $f_i$.
    Let $t_3$ be the earliest time such that since $t_3$ there is a pebble on
    $g_i$.

    At time $t_3$, node $g_i$ is being pebbled, so both nodes $\variable_i$
    and $\bar\variable_i$ have pebbles, and none of literal $\variable_i$ or
    $\bar\variable_i$ is in \tvundef position.
    Let $t_0$ (resp.~$t_1$) be the last time before time $t_3$ such that
    literal $\bar\variable_i$ (resp.~literal $\variable_i$) has a transition
    (\refdef{def:LiteralPosition}).
    Note that neither literal $\bar\variable_i$ nor $\variable_i$ can have a
    transition after time $t_3$:
    to make a transition for literal $\bar\variable_i$ (resp.~$\variable_i$)
    takes $r_i = \gamma_i - 2$ pebbles on the ancestors of node
    $\bar\variable_i'$ (resp.~$\variable_i'$) by \reflem{lem:LiteralTransition},
    but there is a pebble on the other literal gadget, \ie on
    $\literalConst{r_i}{\variable_i}$
    (resp.~$\literalConst{r_i}{\bar\variable_i}$),
    and there is a pebble on $g_i$, which cannot be done with at most
    $\gamma_i - 1$ pebbles over $\Region$.
    So time $t_0$ (resp.~$t_1$) is in fact the last time that literal
    $\bar\variable_i$ (resp.~$\variable_i$) has a transition, and
    literal $\bar\variable_i$ (resp~$\variable_i$) is not in \tvundef position
    since $t_0$ (resp.~$t_1$).

    Assume $t_0 < t_1$ by symmetry (in the rest of this argument).
    \begin{claim}[Clearance]
      \label{cla:Clearance}
      At time $t_1$, there is no pebble over $\Region_0 = \Region \setminus
      \bigl( \ancnode{ \variable_i' } \union \ancnode{ \bar\variable_i
      } \bigr)$.
    \end{claim}
    \begin{proof}
      At time $t_1$, there is a transition of literal $\variable_i$,
      accounting for $r_i = \gamma_i - 2$ pebbles on $\ancnode{
      \variable_i' }$ by \reflem{lem:LiteralTransition}.
      And there is a pebble on the other literal gadget
      $\literalConst{r_i}{\bar\variable_i}$, because literal
      $\bar\variable_i$ is not in \tvundef position.
      This accounts for at least $\gamma_i - 1$ pebbles over $\Region$.
    \end{proof}

    The proof of \refclaim{cla:Clearance} establishes Item~(1).

    By \refclaim{cla:Clearance}, we know that since time $t_1$ literal
    $\variable_i$ is not in \true position, hence in \false position.
    As there is no transition of literals $\bar\variable_i$ or
    $\variable_i$ after time $t_1$, there is
    a pebble on $\ancnode{ \bar\variable_i }$ and
    a pebble on $\ancnode{ \variable_i' }$
    since time $t_1$.
    Over the region $\Region_0$, there are at most $(\gamma_i - 1) - 2 =
    \gamma_i - 3$ pebbles since time $t_1$.
    Note that region $\Region_0$ is associated with the $(i-1)$-assignment
    $\partassign_0$.

    By \refclaim{cla:Clearance}, we have $t_0 < t_1 < t_2$, where $t_2$ is defined
    (in the first paragraph of this proof) as the earliest time since
    which there is at least one pebble on the \highway from $q_{i-1}$ to $f_i$.
    Note that at time $t_2 - 1$ there is no pebble on the \highway from
    $q_{i-1}$ to $f_i$, but at the end there is a pebble on $f_i$.
    The sub-strategy since time $t_2-1$ visits $f_i$ when restricted to the
    \highway from $q_{i-1}$ to $f_i$, so by \reflem{lem:Molding}, there is a time
    $t_4$ after $t_2$ such that there are $(\gamma_i - 5) + 2 = \gamma_i - 3$
    pebbles over the \highway, including one pebble on $q_{i-1}$.
    As a result, over the region $\Region_0 \intersect \ancnode{q_{i-1}}$,
    there is only one pebble at time $t_4$, which is on node $q_{i-1}$.
    The sub-strategy from time $t_1$ to $t_4$ persistently pebble node
    $q_{i-1}$ over the region $\Region_0 \intersect \ancnode{q_{i-1}}$
    using $\gamma_i - 3 = \gamma_{i-1}$ pebbles, so
    $\Persistentprice[\Region_0]{ \gadgetConst{\QBF_{i-1}} } \le
    \gamma_{i-1}$.
    By the assumption that \reflem{lem:QuantifierGadget} holds for $i-1$, we
    know $\QBF_{i-1}\restriction_{\partassign_0}$ is \true.
    As a result, $\exists \variable_i \QBF_{i-1}\restriction_\partassign\ =
    \QBF_i\restriction_\partassign$ is \true, giving Item~(2).
  \end{proof}

  \begin{lemma}[Existential Quantifier Gadget]
    \label{lem:ExistentialQuantifierGadget}
    Assume that \reflem{lem:QuantifierGadget} holds for $i-1$.
    We have
    $\Persistentprice[\Region]{ \gadgetConst{\QBF_i} } = \gamma_i + \ib{
      \QBF_i \restriction_\partassign\text{is \false}}$.
  \end{lemma} 
  \begin{proof}
    Since persistent price is one plus surrounding price
    (\refpr{lem:surrounding-eq-persistent}), it suffices to show that
    $\Surroundprice[\Region]{ \gadgetConst{\QBF_i} } = \gamma_i - 1 + \ib{
      \QBF_i \restriction_\partassign\text{is \false}}$.
    If $\QBF_i \restriction_\partassign$ is \true, then
    $\Surroundprice[\Region]{ \gadgetConst{\QBF_i} } = \gamma_i - 1$, as
    the upper bound (\reflem{lem:ExistentialUpperBound})
    matches the lower bound (\reflem{lem:ExistentialLowerBound}).
    If $\QBF_i \restriction_\partassign$ is \false, then
    $\Persistentprice[\Region]{ \gadgetConst{\QBF_i} } = \gamma_i$, as
    the upper bound (\reflem{lem:ExistentialUpperBound})
    matches the lower bound (\reflem{lem:ExistentialLowerBound}).
  \end{proof}

\subsection{Universal Quantifier Gadget}
\label{sst:UniversalQuantifierGadget}

  Assume that we already have the gadget $\gadgetConst{\QBF_{i-1}}$ and
  that the $i\th$ inner-most quantifier is existential, \ie
  $\quantifier_i = \forall$.
  This quantifier refers to $\variable_{i}$, we set the
  parameter $\bank_i \defeq \gamma_i - 3$ for the corresponding variable gadget.
  We construct $\gadgetConst{\QBF_i}$ as follows.

  \begin{construction}[Universal Quantifier Gadget]
    \label{con:UniversalQuantifierGadget}
    Let $q_{i-1}$ denote the sink of $\gadgetConst{\QBF_{i-1}}$.
    Construct nodes
    $f_i'$, $\bar f_i'$,
    $f_i$, $\bar f_i$,
    $g_i$, $\bar g_i$,
    $h_i$, $\bar h_i$,
    $q_i$,
    and edges
    $(\variable_i, f_i')$, $(\bar\variable_i', f_i')$,
      $(f_i, h_i)$, $(g_i, h_i)$,
    $(\bar\variable_i, \bar f_i')$, $(\variable_i', \bar f_i')$,
      $(\bar f_i, \bar h_i)$, $(\bar g_i, \bar h_i)$,
    $(h_i, q_i)$, $(\bar h_i, q_i)$.
    Add
    a \highway of toll $\gamma_i - 6$ from $f_i'$ to $f_i$,
    a \highway of toll $\gamma_i - 6$ from $\bar f_i'$ to $\bar f_i$,
    a \highway of toll $\gamma_i - 7$ from $q_{i-1}$ to $g_i$,
    and
    a \highway of toll $\gamma_i - 7$ from $q_{i-1}$ to $\bar g_i$.
  \end{construction}

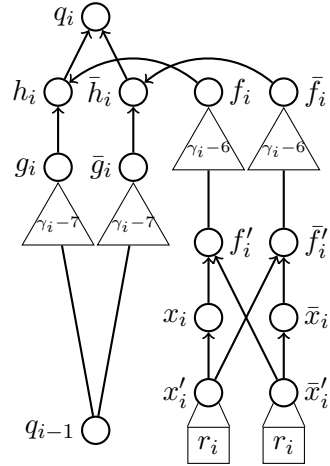
\begin{figure}[ht]
  \begin{center}
    \begin{tikzpicture}
      \GadgetUniversalQuantifierSingle{i}{i-1}{\bank_i}{%
        \scriptscriptstyle\gamma_i - 6}{%
        \scriptscriptstyle\gamma_i - 7}{center}{minimum size=1.1cm}
    \end{tikzpicture}
  \end{center}
  \caption{Universally quantified variable $\forall \variable_i$.}
  \label{fig:UniversalQuantifier}
\end{figure}

  For the gadget $\gadget \defeq \gadgetConst{\QBF_i}$, the nodes in the
  gadget $V(\gadget)$ contains the nodes in the previous gadget $V\bigl(
  \gadgetConst{\QBF_{i-1}} \bigr)$, the new nodes
  $f_i'$, $\bar f_i'$,
  $f_i$, $\bar f_i$,
  $g_i$, $\bar g_i$,
  $h_i$, $\bar h_i$,
  $q_i$,
  and
  the nodes in the four \highway{}s.

  \begin{lemma}[One-Sided Upper Bound]
    \label{lem:OneSidedUpperBound}
    Assume \reflem{lem:QuantifierGadget} holds for $i-1$.
    \begin{enumerate}
      \item If $\QBF_{i-1}\restriction_{\partassign_1}$ is \true, then using
        at most $\gamma_i - 2$ pebbles over $\Region$, we can leave pebbles
        on nodes $\variable_i$, $\bar\variable_i'$, $f_i$, $q_{i-1}$, $g_i$.
      \item If $\QBF_{i-1}\restriction_{\partassign_0}$ is \true, then using
        at most $\gamma_i - 2$ pebbles over $\Region$, we can leave pebbles on
        nodes $\bar\variable_i$, $\variable_i'$, $\bar f_i$, $q_{i-1}$, $\bar
        g_i$.
      \item If $\QBF_{i-1}\restriction_{\partassign_1}$ is \false, then using
        at most $\gamma_i - 1$ pebbles over $\Region$, we can leave pebbles on
        nodes $\variable_i$, $\bar\variable_i'$, $f_i$, $q_{i-1}$, $g_i$.
      \item If $\QBF_{i-1}\restriction_{\partassign_0}$ is \false, then using
        at most $\gamma_i - 1$ pebbles over $\Region$, we can leave pebbles on
        nodes $\bar\variable_i$, $\variable_i'$, $\bar f_i$, $q_{i-1}$, $\bar
        g_i$.
    \end{enumerate}
  \end{lemma}
  \begin{proof}
    For Item~(1), since $\QBF_{i-1}\restriction_{\partassign_1}$ is \true, by
    the assumption that \reflem{lem:QuantifierGadget} holds for $i-1$, we have
    $\persistentprice[\Region_1]{q_{i-1}} = \gamma_{i-1} = \gamma_i - 5$.
    Consider the following strategy to leave pebbles on
    $\variable_i$, $\bar\variable_i'$, $f_i$, $q_{i-1}$, $g_i$
    using at most $\gamma_i - 2$ pebbles over $\Region$:
    \begin{enumerate}[(i)]
      \item Put $\variable_i$ into the canonical \true position with $\bank_i
        + 1 = \gamma_i - 2$
        pebbles by \reflem{lem:VariableAssignment}.
        Now node $\variable_i$ and node $\bar\variable_i'$ have pebbles;
      \item Pebble $f_i'$, persistently pebble the \highway from $f_i'$ to
        $f_i$, then unpebble $f_i'$.
        Now $\variable_i$, $\bar\variable_i'$ and $f_i$ have pebbles.
        Over the \highway from $f_i'$ to $f_i$, at most $(\gamma_i - 6) + 2 =
        \gamma_i - 4$ pebbles are used by \reflem{lem:Highway}.
        Over $\Region$, at most $(\gamma_i - 4) + \underbrace2_{\variable_i,
        \bar\variable_i'} = \gamma_i - 2$ pebbles are used.
      \item Persistently pebble node $q_{i-1}$ using at most $\gamma_i - 5$
        pebbles over the new region $\Region_1$.
        Now $\variable_i$, $\bar\variable_i'$, $f_i$ and $q_{i-1}$ have
        pebbles.
        Over the old region $\Region$, at most
        $ \underbrace{\gamma_i - 5}_{\Region_1}
        + \underbrace{\vphantom{\gamma_i}3}_{\variable_i, \bar\variable_i', f_i}
        = \gamma_i - 2$
        pebbles are used (this step is legal since
        $\Region
        = \Region_1 \union \bigl( \ancnode{ \variable_i } \union \ancnode{
        \bar\variable_i' } \bigr)$ and since $\variable_i$ and $\bar\variable_i$
        have the only pebbles on $\ancnode{ \variable_i } \union \ancnode{\bar\variable_i' }$);
      \item Persistently pebble the \highway from $q_{i-1}$ to $g_i$.
        Now $\variable_i$, $\bar\variable_i'$, $f_i$, $q_{i-1}$ and $g_i$ have
        pebbles.
        Over the \highway from $q_{i-1}$ to $g_i$, at most $(\gamma_i - 7) + 2
        = \gamma_i - 5$ pebbles are used by \reflem{lem:Highway}.
        Over $\Region$, at most $(\gamma_i - 5) + \underbrace3_{\variable_i,
        \bar\variable_i', f_i} = \gamma_i - 2$ pebbles are used.
    \end{enumerate}

    Item~(2) is symmetric to Item~(1).

    For Item~(3),
    since $\QBF_{i-1}\restriction_{\partassign_1}$ is \false, by the assumption
    that \reflem{lem:QuantifierGadget} holds for $i-1$, we have
    $\persistentprice[\Region_1]{ q_{i-1} } = \gamma_{i-1} + 1
    = (\gamma_i - 5) + 1
    = \gamma_i - 4$.
    We run the same strategy as in Item~(1) to leave pebbles on
    $\variable_i$, $\bar\variable_i'$, $f_i$, $q_{i-1}$, $g_i$, using at most
    $\gamma_i - 1$ pebbles over $\Region$ (only Step~(iii) uses one more pebble).

    Item~(4) is symmetric to Item~(3).
  \end{proof}

  \begin{lemma}[Universal Upper Bound]
    \label{lem:UniversalUpperBound}
    Assume \reflem{lem:QuantifierGadget} holds for $i-1$.
    \begin{enumerate}
      \item If $\QBF_i\restriction_\partassign$ is \true, then
        $\surroundprice[\Region]{q_i} \le \gamma_i - 1$.
      \item If $\QBF_i\restriction_\partassign$ is \false, then
        $\surroundprice[\Region]{q_i} \le \gamma_i$.
    \end{enumerate}
  \end{lemma}
  \begin{proof}
    For Item~(1), by assumption $\QBF_i\restriction_\partassign\ =
    \forall\variable_i \QBF_{i-1}\restriction_\partassign$ is \true, so
    assigning $\variable_i$ to \true and to \false both satisfy
    $\QBF_{i-1}$, \ie
    $\QBF_{i-1}\restriction_{\partassign_1}$ is \true and
    $\QBF_{i-1}\restriction_{\partassign_0}$ is \true.
    Consider the following strategy to surround $q_i$ with at most $\gamma_i -
    1$ pebbles over $\Region$:
    \begin{enumerate}[(i)]
      \item Run Item~(1) of \reflem{lem:OneSidedUpperBound} to pebble nodes
        $\variable_i$, $\bar\variable_i'$, $f_i$, $q_{i-1}$, $g_i$,
        using at most $\gamma_i - 2$ pebbles over $\Region$.
      \item Pebble $h_i$.
        Now the six nodes $\variable_i$, $\bar\variable_i'$, $f_i$, $q_{i-1}$,
        $g_i$, $h_i$ have pebbles.
      \item Run the reverse of Item~(1) of \reflem{lem:OneSidedUpperBound} to
        remove pebbles from nodes $\variable_i$, $\bar\variable_i'$, $f_i$,
        $q_{i-1}$, $g_i$.
        Now node $h_i$ has a pebble.
        Over $\Region$, at most $\gamma_i - 2 + \underbrace1_{h_i}
        = \gamma_i - 1$ pebbles are used.
      \item Run Item~(2) of \reflem{lem:OneSidedUpperBound} to pebble nodes
        $\bar\variable_i$, $\variable_i'$, $\bar f_i$, $q_{i-1}$, $\bar g_i$.
        Now the six nodes $h_i$, $\bar\variable_i$, $\variable_i'$, $\bar
        f_i$, $q_{i-1}$, $\bar g_i$ have pebbles.
        Over $\Region$, at most $\gamma_i - 2 + \underbrace1_{h_i}
        = \gamma_i - 1$ pebbles are used.
      \item Pebble $\bar h_i$ to surround $q_i$.
        Seven nodes have pebbles.
    \end{enumerate}

    For Item~(2),
    we run the same strategy as in Item~(1) to surround $q_i$, using at most
    $\gamma_i$ pebbles over $\Region$ (each of Steps~(i), (iii), (iv) may use
    one more pebble by Items~(3) and (4) of \reflem{lem:OneSidedUpperBound}).
  \end{proof}
  
  \begin{lemma}[Universal Lower Bound]
    \label{lem:UniversalLowerBound}
    Assume \reflem{lem:QuantifierGadget} holds for $i-1$.
    \begin{enumerate}
      \item $\surroundprice[\Region]{q_i} \ge \gamma_i - 1$.
      \item If $\surroundprice[\Region]{q_i} \le \gamma_i - 1$, then
        $\QBF_i\restriction_\partassign$ is \true.
    \end{enumerate}
  \end{lemma}
  \begin{proof}
    Fix a strategy to surround $q_i$ using at most $\gamma_i - 1$ pebbles over
    $\Region$.
    Let
    $\Region_f \defeq \{ h_i \} \union \ancnode{ f_i } \setminus
    \propancnode{ f_i' }$
    be the region to augment $h_i$ to the \highway from $f_i'$ to $f_i$,
    and
    $\Region_{\bar f} \defeq \{ \bar h_i \} \union \ancnode{ \bar f_i }
    \setminus \propancnode{ \bar f_i' }$
    be the region to augment $\bar h_i$ to the \highway from $\bar f_i'$ to
    $\bar f_i$.
    At the end, the region $\Region_f$ has a pebble (on $h_i$) and the region
    $\Region_{\bar f}$ has a pebble (on $\bar h_i$).
    Let $t_1$ (resp.~$t_0$) be the earliest time such that since time $t_1$
    (resp.~$t_0$) the region $\Region_f$ (resp.~$\Region_{\bar f}$) has
    pebble.

    Assume $t_0 < t_1$ by symmetry (in the rest of this argument).
    At time $t_1 - 1$, there is no pebble on $\Region_f$, and there are
    pebbles on nodes $\variable_i$ and $\bar\variable_i'$ (so that node $f_i'$
    can be pebbled at time $t_1$).
    Hence literal $\bar\variable_i$ is in \false position, and literal
    $\variable_i$ is not in \tvundef position.
    Note that there is no transition of literals $\variable_i$ or
    $\bar\variable_i$ since time $t_1$:
    to make a transition for literal $\variable_i$ (resp.~$\bar\variable_i$)
    takes $r_i = \gamma_i - 3$ pebbles over $\ancnode{\variable_i'}$
    (resp.~$\ancnode{\bar\variable_i'}$) by \reflem{lem:LiteralTransition},
    but there is
    a pebble on $\Region_{\bar f}$,
    a pebble on $\Region_f$,
    and
    a pebble on the other literal gadget
    $\literalConst{r_i}{\bar\variable_i}$
    (resp.~$\literalConst{r_i}{\variable_i}$);
    thus a transition cannot be done with at most $\gamma_i - 1$ pebbles over
    $\Region$.
    As such, there is a pebble on $\ancnode{\variable_i}$ and a pebble on
    $\ancnode{\bar\variable_i'}$ since time $t_1$.

    Because node $f_i$ must be visited before node $h_i$ can be pebbled, by
    \reflem{lem:Highway}, at some later time $t_2 > t_1$ there are
    $(\gamma_i - 6) + 2 = \gamma_i - 4$
    pebbles on the \highway from $f_i'$ to $f_i$.
    At time $t_2$ over $\Region$, there are at least
    \begin{equation}
      \underbrace{\gamma_i - 4}_{\Region_f \setminus \{ h_i \}}
    + \underbrace{\vphantom{\gamma_i}1}_{\Region_{\bar f}}
    + \underbrace{\vphantom{\gamma_i}2}_{\Ancnode{ \{ \variable_i, \bar\variable_i' \} }}
    = \gamma_i - 1
\label{eq:Clearance}
    \end{equation}
    pebbles, giving Item~(1).

    At time $t_2$, there is no pebble on $\{ h_i \} \union \Region_g$ where
    $\Region_g \defeq \ancnode{ g_i }
    \setminus \Ancnode{ \{ \variable_i, \bar\variable_i' \} }$,
    as all $\gamma_i - 1$ pebbles over $\Region$ are on $\Region_f \setminus
    \{ h_i \}$, $\Region_{\bar f}$ or $\Ancnode{ \{ \variable_i,
      \bar\variable_i' \}}$ by \refeq{eq:Clearance}.
    Over $\Region \setminus \Region_g$, there are at least
    \begin{equation*} \underbrace1_{\Region_f}
    + \underbrace1_{\Region_{\bar f}}
    + \underbrace2_{\Ancnode{ \{ \variable_i, \bar\variable_i' \} }}
    = 4\end{equation*}
    pebbles since time $t_2$,
    so at most $(\gamma_i - 1) - 4 = \gamma_i - 5$ pebbles over $\Region_g$.
    Because node $g_i$ must be visited before node $h_i$ can be pebbled, by
    \reflem{lem:Molding}, at some later time $t_3 > t_2$ there are $(\gamma_i -
    7) + 2 = \gamma_i - 5$ pebbles on the \highway from $q_{i-1}$ to $g_i$,
    including one on node $q_{i-1}$.
    Recall the region $\Region_1 = \Region \setminus \bigl(
    \ancnode{\variable_i} \union \ancnode{\bar\variable_i'} \bigr)$
    that is associated with the $(i-1)$-assignment $\partassign_1$.
    At time $t_3$, there is only one pebble over $\Region_1 \intersect
    \ancnode{q_{i-1}}$, which is on node $q_{i-1}$.
    The sub-strategy from time $t_2$ to $t_3$ persistently pebble node
    $q_{i-1}$ over the region $\Region_1 \intersect \ancnode{q_{i-1}}$
    using $\gamma_i - 5 = \gamma_{i-1}$ pebbles, so
    $\Persistentprice[\Region_1]{ \gadgetConst{\QBF_{i-1}} } \le
    \gamma_{i-1}$.
    By the assumption that \reflem{lem:QuantifierGadget} holds for $i-1$, we
    know $\QBF_{i-1}\restriction_{\partassign_1}$ is \true.

    We claim that at time $t_2$, node $\bar h_i$ has a pebble (by modifying
    the argument in the previous paragraph).
    Let $\Region_h \defeq \bigl( \{ h_i \} \union \ancnode{ g_i } \bigr)
    \setminus \ancnode{q_{i-1}}$
    be nodes properly in the \highway from $q_{i-1}$ to $g_i$ plus $h_i$,
    and
    $\Region_{\bar h} \defeq \bigl( \{ \bar h_i \} \union \ancnode{ \bar
    g_i } \bigr) \setminus \ancnode{q_{i-1}}$
    be nodes properly in the \highway from $q_{i-1}$ to $\bar g_i$ plus $\bar
    h_i$.

    \begin{claim}[Persistence]
      \label{cla:Persistence}
      At time $t_2$, node $\bar h_i$ has a pebble.
    \end{claim}
    \begin{proof}
    For otherwise,
    at time $t_2$, there is no pebble on $\Region_h$ or $\Region_{\bar h}$,
    as all $\gamma_i - 1$ pebbles over $\Region$ are on
    $\Region_f \setminus \{ h_i \}$,
    $\Region_{\bar f} \setminus \{ \bar h_i \}$
    or
    $\ancnode{ \{ \variable_i, \bar\variable_i' \}}$
    by \refeq{eq:Clearance}.
    Let $t_3$ (resp.~$t_4$) be the earliest time such that since time $t_3$
    (resp.~$t_4$) the region $\Region_h$ (resp.~$\Region_{\bar h}$) has
    pebble.

    \begin{claim}[No Double Persistence]
      \label{cla:NoDoublePersistence}
      Since $t_3$, there are at least two pebbles over $\Region_f \union
      \Region_h$.
    \end{claim}
    \begin{proof}
      Note that $\Region_f$ consists of the \highway from $f_i'$ to $f_i$ plus
      node $h_i$, and $\Region_h$ consists of nodes \emph{properly} in the
      \highway from $q_{i-1}$ to $g_i$ plus node $h_i$.

      Fix an induced subgraph $F \subseteq G$ (\eg $F =$ the \highway from
      $f_i'$ to $f_i$) having a unique sink.
      Say a pebbling configuration $\pconf$ is \introduceterm{\vislocked} on $F$ if
      $\persistentprice[F]{\pconf}=\visitprice{F}$; this is if the pebbles cannot
      be cleared without using $\visitprice{F}$ pebbles (when restricted to $F$).
      In particular, any configuration with a pebble on the sink of $F$ is
      \vislocked on $F$.
      Also, if a configuration is \vislocked on $F$, then there is a pebble on
      $F$.
      Given a pebbling configuration on $G$, we say that $f_i$ (resp.~$g_i$) is
      \introduceterm{\vislocked} if the configuration is \vislocked on the \highway from
      $f_i'$ to $f_i$ (resp.~from $q_{i-1}$ to $g_i$).

      Assume for contradiction that at some time $t_7 \ge t_3$ there is only
      one pebble over $\Region_f \union \Region_h$, which must be on $h_i$ by
      the inclusion-exclusion principle.
      Let $t_5$ be the earliest time before $t_7$ such that there is a pebble
      on node $h_i$ from $t_5$ to $t_7$.
      We know $t_0 < t_1 < t_3 < t_5 < t_7$.
      At time $t_5$ node $h_i$ is pebbled, so node $g_i$ and node $f_i$ each
      has a pebble, and both are \vislocked.
      At time $t_7$, nodes properly in the two \highway{}s have no pebbles, and
      nodes $g_i$ and $f_i$ are not \vislocked.
      Let $t_6$ be the earliest time after $t_5$ such that one of the \highway{}s
      is not \vislocked, then $t_5 < t_6 < t_7$.

      \begin{itemize}
        \item
          If the \highway from $f_i'$ to $f_i$ stops being \vislocked at time $t_6$, then
          at time $t_6 - 1$ there are $(\gamma_i - 6) + 2 = \gamma_i - 4$
          pebbles over the \highway from $f_i'$ to $f_i$ by \reflem{lem:Highway}.
          Over $\Region$, there are
          \begin{equation*}          
            \underbrace{\gamma_i - 4}_{\text{\highway from }f_i'\text{ to }f_i}
            + \underbrace{\vphantom{\gamma_i}1}_{\text{\highway from }q_{i-1}\text{ to }g_i}
            + \underbrace{\vphantom{\gamma_i}1}_{h_i}
            + \underbrace{\vphantom{\gamma_i}1}_{\Region_{\bar f}}
            + \underbrace{\vphantom{\gamma_i}2}_{\Ancnode{ \{ \variable_i, \bar\variable_i' \} }}
            = \gamma_i + 1
          \end{equation*}
          pebbles, contradicting that at most $\gamma_i - 1$ pebbles are used
          over $\Region$.
        \item
          If the \highway from $q_{i-1}$ to $g_i$ stops being \vislocked at time $t_6$,
          then at time $t_6 - 1$ there are $(\gamma_i - 7) + 2 = \gamma_i - 5$
          pebbles over the \highway from $q_{i-1}$ to $g_i$ by
          \reflem{lem:Highway}.
          Over $\Region$, there are
          \begin{equation*}
            \underbrace{\gamma_i - 5}_{\text{\highway from }q_{i-1}\text{ to
              }g_i}
            + \underbrace{\vphantom{\gamma_i}1}_{\text{\highway from }f_i'\text{ to }f_i}
            + \underbrace{\vphantom{\gamma_i}1}_{h_i}
            + \underbrace{\vphantom{\gamma_i}1}_{\Region_{\bar f}}
            + \underbrace{\vphantom{\gamma_i}2}_{\Ancnode{ \{ \variable_i, \bar\variable_i' \} }}
            = \gamma_i
          \end{equation*}
          pebbles, contradicting that at most $\gamma_i - 1$ pebbles are used
          over $\Region$.  \qedhere
      \end{itemize}
    \end{proof}

    Assume $t_3 < t_4$ by symmetry (in the rest of \refclaim{cla:Persistence}).
    At time $t_4 - 1$, there is no pebble on $\Region_{\bar h}$.
    By \reflem{lem:Highway}, at some later time $t_5 > t_4$ there are
    $(\gamma_i - 7) + 2 = \gamma_i + 5$
    pebbles on the \highway from $q_{i-1}$ to $\bar g_i$.
    Over $\Region$, there are at least
    \begin{equation*}
    \underbrace{\gamma_i - 5}_{\text{\highway from }q_{i-1}\text{ to }\bar
      g_i}
    + \underbrace{\vphantom{\gamma_i}2}_{\Region_f \union \Region_h}
    + \underbrace{\vphantom{\gamma_i}1}_{\Region_{\bar f}}
    + \underbrace{\vphantom{\gamma_i}2}_{\Ancnode{ \{ \variable_i, \bar\variable_i' \} }}
    = \gamma_i
    \end{equation*}
    pebbles, contradicting that at most $\gamma_i - 1$ pebbles are used over
    $\Region$.
    \end{proof}

    \refclaim{cla:Persistence} shows that there is a pebble on $\bar h_i$ at time
    $t_2$, hence there is no pebble on the \highway from $\bar f_i'$ to $\bar
    f_i$ by \refeq{eq:Clearance}.
    Let
    $t_3$ be the earliest time before $t_2$ such that the only pebble on
    $\Region_{\bar f}$ is on $\bar h_i$,
    and let
    $t_6$ be the earliest time before $t_3$ such that $\bar h_i$ has a pebble
    from time $t_6$ to $t_3$.
    At time $t_3$, there are pebbles on nodes $\variable_i'$ and
    $\bar\variable_i$ (so that node $\bar f_i'$ can be unpebbled at time
    $t_3 - 1$), hence literal $\variable_i$ is in \false position, and literal
    $\bar\variable_i$ is not in \tvundef position.
    Note that there is no transition of literals $\variable_i$ or
    $\bar\variable_i$ between time $t_6$ and $t_3 - 1$:
    to make a transition for literal $\variable_i$ (resp.~$\bar\variable_i$)
    takes $r_i = \gamma_i - 3$ pebbles over $\ancnode{\variable_i'}$
    (resp.~$\ancnode{\bar\variable_i'}$) by \reflem{lem:LiteralTransition},
    but there is
    a pebble on $\bar h_i$,
    a pebble on $\Region_{\bar f} \setminus \{ \bar h_i \}$,
    and
    a pebble on the other literal gadget
    $\literalConst{r_i}{\bar\variable_i}$
    (resp.~$\literalConst{r_i}{\variable_i}$);
    thus a transition cannot be done with at most $\gamma_i - 1$ pebbles over
    $\Region$.
    As such, there is a pebble on $\ancnode{\variable_i'}$ and a pebble on
    $\ancnode{\bar\variable_i}$ between time $t_6$ and $t_3$.

    At time $t_6$, node $\bar g_i$ and node $\bar f_i$ each has a pebble (so
    that node $\bar h_i$ can be pebbled at time $t_6 - 1$).
    At time $t_3$, the \highway from $\bar f_i'$ to $\bar f_i$ has no pebble.
    By \reflem{lem:Highway}, there is a time $t_4$ between $t_6$ and $t_3$ such
    that there are $(\gamma_i - 6) + 2 = \gamma_i - 4 $ pebbles on the
    \highway from $\bar f_i'$ to $\bar f_i$.
    We know $t_6 < t_4 < t_3 < t_2$.
    At time $t_4$ over $\Region$, there are at least
    \begin{equation}
      \underbrace{\gamma_i - 4}_{\Region_{\bar f} \setminus \{ \bar h_i \}}
    + \underbrace{\vphantom{\gamma_i}1}_{\bar h_i}
    + \underbrace{\vphantom{\gamma_i}2}_{\Ancnode{ \{ \variable_i', \bar\variable_i \} }}
    = \gamma_i - 1
\label{eq:BarClearance}
    \end{equation}
    pebbles.

    At time $t_4$, there is no pebble on $\Region_{\bar g}$ where
    $\Region_{\bar g} \defeq \ancnode{ \bar g_i }
    \setminus \Ancnode{ \{ \variable_i', \bar\variable_i \} }$,
    as all $\gamma_i - 1$ pebbles over $\Region$ are on $\Region_{\bar f}$ or
    $\ancnode{ \{ \variable_i, \bar\variable_i' \}}$ by
    \refeq{eq:BarClearance}.
    Over $\Region \setminus \Region_{\bar g}$, there are at least
    \begin{equation*} \underbrace1_{\Region_{\bar f} \setminus \{ \bar h_i \}}
    + \underbrace1_{\bar h_i}
    + \underbrace2_{\Ancnode{ \{ \variable_i', \bar\variable_i \} }}
    = 4\end{equation*}
    pebbles between time $t_6$ and $t_3 - 1$,
    so at most $(\gamma_i - 1) - 4 = \gamma_i - 5$ pebbles over $\Region_{\bar
      g}$.
    Because node $\bar g_i$ is visitied at time $t_6$ and the \highway from
    $q_{i-1}$ to $\bar g_i$ has no pebble at time $t_4$, by
    \reflem{lem:Molding}, at some time $t_5$ between $t_6$ and $t_4$ there are
    $(\gamma_i - 7) + 2 = \gamma_i - 5$ pebbles on the \highway from $q_{i-1}$
    to $\bar g_i$, including one on node $q_{i-1}$.
    Recall the region $\Region_0 = \Region \setminus \bigl(
    \ancnode{\variable_i'} \union \ancnode{\bar\variable_i} \bigr)$
    that is associated with the $(i-1)$-assignment $\partassign_0$.
    At time $t_5$, there is only one pebble over $\Region_0 \intersect
    \ancnode{q_{i-1}}$, which is on node $q_{i-1}$.
    The (reverse of the) sub-strategy from time $t_5$ to $t_4$ persistently
    pebble node $q_{i-1}$ over the region $\Region_0 \intersect
    \ancnode{q_{i-1}}$ using $\gamma_i - 5 = \gamma_{i-1}$ pebbles, so
    $\Persistentprice[\Region_0]{ \gadgetConst{\QBF_{i-1}} } \le
    \gamma_{i-1}$.
    By the assumption that \reflem{lem:QuantifierGadget} holds for $i-1$, we
    know $\QBF_{i-1}\restriction_{\partassign_0}$ is \true.

    As a result, $\forall \variable_i \QBF_{i-1}\restriction_\partassign\ =
    \QBF_i\restriction_\partassign$ is \true, giving Item~(2).
  \end{proof}

  \begin{lemma}[Universal Quantifier Gadget]
    \label{lem:UniversalQuantifierGadget}
    Assume that \reflem{lem:QuantifierGadget} holds for $i-1$.
    We have
    $\Persistentprice[\Region]{ \gadgetConst{\QBF_i} } = \gamma_i + \ib{ \QBF_i
      \restriction_\partassign\text{is \false}}$.
  \end{lemma} 
  \begin{proof}
    Since persistent price is one plus surrounding price
    (\refpr{lem:surrounding-eq-persistent}), it suffices to show that
    $\Surroundprice[\Region]{ \gadgetConst{\QBF_i} } = \gamma_i - 1 + \ib{
      \QBF_i \restriction_\partassign\text{is \false}}$.
    If $\QBF_i \restriction_\partassign$ is \true, then
    $\Surroundprice[\Region]{ \gadgetConst{\QBF_i} } = \gamma_i - 1$, as
    the upper bound (\reflem{lem:UniversalUpperBound})
    matches the lower bound (\reflem{lem:UniversalLowerBound}).
    If $\QBF_i \restriction_\partassign$ is \false, then
    $\Persistentprice[\Region]{ \gadgetConst{\QBF_i} } = \gamma_i$, as
    the upper bound (\reflem{lem:UniversalUpperBound})
    matches the lower bound (\reflem{lem:UniversalLowerBound}).
  \end{proof}

\section{Product Construction for Reversible Pebbling}
\label{sec:blow-up-reversible}

The part of the proof of Theorem~\ref{th:approx-both} that deals with
reversible pebbling uses as a black box the construction in
Theorem~\ref{th:blowup-both} for reversible pebbling. 
Now we state it again and we give its full proof.

\begin{theorem}
  \label{th:blowup}
  Given two graphs $\graphstd_1$ and $\graphstd_2$, 
  there is a \efficiently constructible graph 
  $\blowupConst{\graphstd_1}{\graphstd_2}$ 
  of size 
  $3\setsize{\graphstd_1}\cdot\setsize{\graphstd_2}$ 
  with reversible pebbling price
  $\Persistentprice{\blowupConst{\graphstd_1}{\graphstd_2}} = \persistentprice{\graphstd_1} + \persistentprice{\graphstd_2} + 1$.
\end{theorem}

We want to construct a graph $\blowup$ to inherit structures from two
graphs, which are called respectively the \introduceterm{exterior}
graph $G_1$ and the \introduceterm{interior} graph $G_2$.
Intuitively, for every node in $G_1$, we will construct a
\introduceterm{block} with the structure of $G_2$: in each such block,
for every node in $G_2$ we create a \introduceterm{cell} of three
nodes, where different cells are connected according to the exterior
graph $G_1$ and the interior graph $G_2$.

  \begin{figure}[ht]
    \begin{center}
      \begin{tikzpicture}[scale=.75]
        \ExampleBlowUpReversible
      \end{tikzpicture}
      \caption{Example of \refcon{con:BlowUp}: product of a pyramid of height 1 and a rhombus.}
      \label{fig:ExampleBlowUp}
    \end{center}
  \end{figure}

\begin{construction}[Product for reversible pebbling]
  \label{con:BlowUp}
  Fix two graphs $G_1$ and $G_2$, and denote $\sink_2$ as the unique sink of
  $G_2$.
  Construct a graph $\blowup \defeq \blowupConst{G_1}{G_2}$ as follows.
  For every node $(v_1, v_2) \in V(G_1) \times V(G_2)$, create three nodes
  $(v_1, v_2)\outblowup$, $(v_1, v_2)\extblowup$, $(v_1, v_2)\intblowup$.
  Add an edge from the exterior node to the output node, \ie
  from node $(v_1, v_2)\extblowup$ to $(v_1, v_2)\outblowup$;
  add an edge from the interior node to the output node, \ie
  from node $(v_1, v_2)\intblowup$ to $(v_1, v_2)\outblowup$.
  The exterior node supports the structure of the exterior graph $G_1$, in the
  sense that for every predecessor $w_1$ of $v_1$ in $G_1$, we create an edge
  from the sink the $w_1$-block of $G_2$, \ie
  from node $(w_1, z_2)\outblowup$ to node $(v_1, v_2)\extblowup$.
  The interior node supports the structure of the interior graph $G_2$, in the
  sense that for every predecessor $w_2$ of $v_2$ in $G_2$, we create an edge
  from the output node of $w_2$, \ie
  from node $(v_1, w_2)\outblowup$ to node $(v_1, v_2)\intblowup$.

  Formally,
  $V(\blowup) \defeq \bigl\{
    (v_1, v_2)\outblowup, \allowbreak
    (v_1, v_2)\extblowup, \allowbreak
    (v_1, v_2)\intblowup
    \where
    v_1 \in V(G_1), \allowbreak
    v_2 \in V(G_2)
  \bigr\}$,
  and
  $E(\blowup) \defeq
  E\outblowup \disjunion E\extblowup \disjunion E\intblowup$, where
  $E\outblowup \defeq \Bigl\{
    \bigl( (v_1, v_2)\extblowup, (v_1, v_2)\outblowup \bigr) , \allowbreak
    \bigl( (v_1, v_2)\intblowup, (v_1, v_2)\outblowup \bigr)
    \where \allowbreak
    v_1 \in V(G_1), v_2 \in V(G_2)
    \Bigr\}$,
  $E\extblowup \defeq \Bigl\{
    \bigl( (w_1, z_2)\outblowup, (v_1, v_2)\extblowup \bigr)
    \where \allowbreak
    v_1 \in V(G_1), \allowbreak
    v_2 \in V(G_2), \allowbreak
    w_1 \in \prednode[G_1]{v_1}
    \Bigr\}$, and
  $E\intblowup \defeq \Bigl\{
    \bigl( (v_1, w_2)\outblowup, (v_1, v_2)\intblowup \bigr)
    \where \allowbreak
    v_1 \in V(G_1), \allowbreak
    v_2 \in V(G_2), \allowbreak
    w_2 \in \prednode[G_2]{v_2}
    \Bigr\}$.
\end{construction}

Clearly, if $G_1$ and $G_2$ each has in-degree at most two and a unique sink,
then so does the resulting graph $\blowup$.
Note that the graph $\blowup$ partitions into $\setsize{V(G_1)}$ blocks,
namely, for each $v_1 \in V(G_1)$, the $v_1$-block is the subgraph of
$\blowup$ induced over the node set
$\bigl\{ (v_1, v_2)\outblowup, \allowbreak
  (v_1, v_2)\extblowup, \allowbreak
  (v_1, v_2)\intblowup
  \where \allowbreak
  v_2 \in V(G_2) \bigr\}$.
Each such block further partitions into $\setsize{V(G_2)}$ cells, namely,
in the $v_1$-block, for each $v_2 \in V(G_2)$, the $(v_1, v_2)$-cell is the
subgraph of the $v_1$-block induced over the node set
$\bigl\{ (v_1, v_2)\outblowup,\allowbreak
  (v_1, v_2)\extblowup,\allowbreak
  (v_1, v_2)\intblowup \bigr\}$.

Finally, given a configuration $\pconfBlowUp$ on $\blowup$, say the $v_1$-block
(resp.~the $(v_1, v_2)$-cell) is \emph{surrounded} if any exterior node of the
$v_1$-block (resp.~the interior node of the $(v_1, v_2)$-cell) is surrounded
in $\pconfBlowUp$.
Note that the $v_1$-block is surrounded iff \emph{every} exterior node of the
$v_1$-block is surrounded.

\begin{lemma}[Upper Bound]
  \label{lem:BlowUpUpperBound}
  $\Persistentprice{\blowupConst{G_1}{G_2}} \le \persistentprice{G_1} +
  \persistentprice{G_2} + 1$.
\end{lemma}
\begin{proof}
  Fix a persistent pebbling $\pebbling_1$ of $G_1$ using
  $\persistentprice{G_1}$ pebbles, and a persistent pebbling $\pebbling_2$ of
  $G_2$ using $\persistentprice{G_2}$ pebbles.
  We will construct a persistent pebbling $\pebblingBlowUp$ of
  $\blowupConst{G_1}{G_2}$ using $\persistentprice{G_1} +
  \persistentprice{G_2} + 1$ pebbles.
 
  We claim that the persistent price of each block of $\blowup$ is at most
  $\persistentprice{G_2} + 2$.
  For any $v_1 \in V(G_1)$, to persistently pebble the $v_1$-block (assuming
  the $v_1$-block is surrounded), simulate $\pebbling_2$ as follows:
  whenever $\pebbling_2$ pebbles a node $v_2 \in V(G_2)$, the simulating
  pebbling has a phase to persistently pebble the $(v_1, v_2)$-cell, and
  whenever $\pebbling_2$ unpebbles a node $v_2 \in V(G_2)$, then the simulating
  pebbling has a phase to persistently unpebble the $(v_1, v_2)$-cell.
  If the current configuration in $\pebbling_2$ is $\pconf$, and the
  configuration in the simulating pebbling at the end of a phase is $\pconfBlowUp$,
  then the simulating pebbling maintains the phase-invariant that $\pconfBlowUp =
  \bigl\{ (v_1, v_2)\outblowup \where v_2 \in \pconf \bigr\}$.
  Note that the simulating pebbling is legal:
  since the pebbling $\pebbling_2$ is legal,
  when $v_2$ is pebbled or unpebbled it is surrounded in the current
  configuration $\pconf$,
  so the $(v_1, v_2)$-cell is surrounded in the simulating configuration
  $\pconfBlowUp$ and the interior node can be pebbled or unpebbled;
  and we assume that the $v_1$-block is surrounded, so the exterior node can
  be pebbled or unpebbled.
  The simulating pebbling uses at most two more pebbles (on the exterior node
  and the interior node of each cell), for at most $\persistentprice{G_2} + 2$
  pebbles over the $v_1$-block.

  Then the resulting graph $\blowupConst{G_1}{G_2}$ can be persistently
  pebbled by simulating $\pebbling_1$ as follows.
  Whenever $\pebbling_1$ pebbles a node $v_1 \in V(G_1)$, the simulating
  pebbling $\pebblingBlowUp$ has a stage to persistently pebble the $v_1$-block;
  whenever $\pebbling_1$ unpebbles a node $v_1 \in V(G_1)$, the simulating
  pebbling $\pebblingBlowUp$ has a stage to persistently unpebble the $v_1$-block.
  If the current configuration in $\pebbling_1$ is $\pconf$, and the
  configuration in the simulating pebbling $\pebblingBlowUp$ at the end of a stage
  is $\pconfBlowUp$,
  then the simulating pebbling maintains the stage-invariant that $\pconfBlowUp =
  \bigl\{ (v_1, z_2)\outblowup \where v_1 \in \pconf \bigr\}$.
  Note that the simulating pebbling $\pebblingBlowUp$ is legal:
  since the pebbling $\pebbling_1$ is legal,
  when $v_1$ is pebble or unpebbled it is surrounded in the current
  configuration $\pconf$,
  so the $v_1$-block is surrounded in the simulating configuration $\pconfBlowUp$,
  so the $v_1$-block can be persistently pebbled or unpebbled in a stage of
  $\pebblingBlowUp$.
  When the $v_1$-block is pebbled or unpebbled in a stage of $\pebblingBlowUp$,
  there are at most $\persistentprice{G_1} - 1$ pebbles on other blocks of
  $\blowup$, and there are at most $\persistentprice{G_2} + 2$ pebbles in the
  $v_1$-block, for a total of $\persistentprice{G_1} + \persistentprice{G_2} +
  1$ pebbles.
\end{proof}

To prove the lower bound we extract simultaneous pebblings of $G_1$ and $G_2$ from any pebbling of $\blowup$ and use the known pebbling prices of $G_1$ and $G_2$ to argue that some configuration needs many pebbles. We do so by projecting a pebbling of $\blowup$ into $G_1$ and $G_2$ as the skeleton of a pebbling and then filling the gaps between configurations with a legal sequence of pebbling moves.

\mvcomment{I would suggest to use the notation $\mathbb{P}$ for $B$-configurations, $\mathbb{I}$ for $G_2$-configurations (so $\mathbb{I}^{w_1}(t) \defeq I(\mathbb{P}_t)$) and $\mathbb{E}$ for $G_1$-configurations.}

\begin{lemma}[Lower Bound]
  \label{lem:BlowUpLowerBound}
  $\Persistentprice{\blowupConst{G_1}{G_2}} \ge \persistentprice{G_1} +
  \persistentprice{G_2} + 1$.
\end{lemma}

\begin{proof}
  Fix any persistent pebbling
  $\pebblingBlowUp =
  (\pconfBlowUp_0, \pconfBlowUp_1, \dotsc, \pconfBlowUp_{\stoptime})$ of
  $\blowupConst{G_1}{G_2}$.
  For every $v_1 \in V(G_1)$, we are going to simulate a pebbling
  $\pebblingBlock{v_1}$ on $G_2$ based (essentially) on the configurations of
  $\pebblingBlowUp$ over the $v_1$-block.
  From the family of pebblings $\{ \pebblingBlock{v_1} \}_{v_1 \in V(G_1)}$,
  we then simulate a persistent pebbling $\pebbling$ on $G_1$.
  
  In more detail, for each $v_1\in V(G_1)$ we define a mapping
  $\funcdescr{\intmappingBlock{v_1}}{\blowup}{G_2}$ and we view the
  sequence of configurations $(\intmappingBlock{v_1}(\pconf_t'))_{\timet\in[0,\stoptime]}$ as the
  skeleton of a pebbling of $G_2$. We fill the gaps between
  configurations according to the algorithm described below to obtain a legal
  pebbling $\pebblingBlock{v_1}$. Similarly we define a mapping
  $\funcdescr{\formatfunctiontoset{Ext}}{\blowup}{G_1}$ to construct the pebbling
  $\pebbling$ on $G_1$.

  \mvcomment{Move projection and closure outside and give some basic properties and examples?}
  To describe the mappings we need some definitions.
  Given a configuration $\pconfBlowUp$ in $\pebblingBlowUp$ of
  $\blowupConst{G_1}{G_2}$, its projection to the output (resp.~exterior,
  interior) nodes of the $v_1$-block is
  $\outproj{v_1}{\pconfBlowUp} \defeq \bigl\{ v_2 \in V(G_2)
    \where \allowbreak
    (v_1, v_2)\outblowup \in \pconfBlowUp \bigr\}$
  (resp.~%
  $\extproj{v_1}{\pconfBlowUp} \defeq \bigl\{ v_2 \in V(G_2)
    \where \allowbreak
    (v_1, v_2)\extblowup \in \pconfBlowUp \bigr\}$,
  $\intproj{v_1}{\pconfBlowUp} \defeq \bigl\{ v_2 \in V(G_2)
    \where \allowbreak
    (v_1, v_2)\intblowup \in \pconfBlowUp \bigr\}$).
  
  The closure
  $\pebClosure{\pconfBlowUp} \subseteq
  V\bigl( \blowupConst{G_1}{G_2} \bigr)$
  is the smallest set of nodes containing $\pconfBlowUp$ that is closed
  under pebble placements on interior or output nodes;
  equivalently, $\pebClosure{\pconfBlowUp}$ can be generated by the
  following algorithm:
  Start with $\pconfBlowUp$, while there is a node
  $v \in V\bigl( \blowupConst{G_1}{G_2} \bigr)$ which is
  an interior node $v = (v_1, v_2)\intblowup$ or
  an output node $v = (v_1, v_2)\outblowup$ such that
  $v$ is surrounded by, but not in, the subset of nodes having pebbles,
  pebble $v$. Note that the closure of a $v_1$-block does not depend on
  other blocks as they are connected only through exterior nodes.

  For brevity, given a graph $G$ and a subset $U \subseteq V(G)$ of
  vertices, denote $\unsurrounded GU \defeq \{ v \in V(G) \where
  \prednode[G]{v} \not\subseteq U \}$ as the subset of nodes in $G$
  not surrounded by $U$.

  The block mapping is
  $\IntInvariantFunc{w_1}{\pconfBlowUp} \defeq
  \IntInvariant{w_1}{\pconfBlowUp}$

  We define the interior mapping $\intmappingBlock{w_1}(\pconfBlowUp)$ to be
  $\IntInvariantFunc{w_1}{\pconfBlowUp}$ if the $w_1$-block is
  surrounded, and $\IntInvariantFunc{w_1}{\pebClosure{\pconfBlowUp}} 
  $ if the $w_1$-block is
  not surrounded.

  Given a configuration $\pconfBlowUp$, 
  its \perprojText is $\perproj{\pconfBlowUp} \defeq
  \bigl\{
    v_1 \in V(G_1)
    \where \allowbreak
    \intmappingBlock{v_1}(\pconfBlowUp)
    \text{ is \perlocked}
  \bigr\}$, and
  \visprojText is $\visproj{\timet} \defeq
  \bigl\{
    v_1 \in V(G_1)
    \where \allowbreak
    \intmappingBlock{v_1}(\pconfBlowUp)
    \text{ is \vislocked}
  \bigr\}$.

  Finally the exterior mapping is $\extInvariantFunc{\pconfBlowUp} \defeq
  \ExtInvariant{\pconfBlowUp}$.

  We abuse the notation for mappings from configurations in
  $\pebblingBlowUp$ and write $f(t)$ to mean $f(\pconfBlowUp_t)$. In addition we
  define
  $\pconfBlockTime{v_1}{\timet}=\intmappingBlock{v_1}(\pconf_t')$ and
  $\pconfTime{\timet}=\extInvariantFunc{\pconf_t'}$.  Note that there
  can be multiple pebble moves between, say, $\pconfTime{\timet-1}$
  and $\pconfTime{\timet}$.

  We construct the pebblings $\pebblingBlocks$ and $\pebbling$
  according to \refalg{alg:reasonable}, which are legal by
  \refclaim{cla:Correctness}.
  Note that $\pebbling$ is a persistent pebbling of $G_1$:
  $\pconfBlockTime{w_1}{\stoptime}
  = \IntInvariantUnsurroundedFunc{w_1}{\stoptime}$, which is
  $\emptyset$ if $w_1$ is not the sink of $G_1$, and
  $\set{z_2}$ if $w_1$ is the sink of $G_1$, so
  $\pconfTime{\stoptime}
  = \extInvariantFunc{\stoptime} = \set{z_1}$, the
  sink of $G_1$.
  Then the lower bound follows from \refclaim{cla:LegalImpliesLowerBound}.
\end{proof}

  The algorithm to construct the remaining configurations in the
  pebblings $\pebblingBlocks$ and $\pebbling$, given
  $\pebblingBlowUp$, is essentially to insert and remove missing
  pebbles in topological order whenever two configuration are
  different, and make such a pebbling go through a specific
  configuration in the exterior case. We give an explicit description
  below and in \refclaim{cla:Correctness} we prove that it is equivalent
  to this implicit description.

\begin{definition}[\Reasonablepebbling]
  Given a configuration $\pconf$, a set of vertices to pebble
  $\toPebble$ and a set of vertices to unpebble $\toUnpebble$, the
  \introduceterm{\reasonablepebbling} is the following
  pebbling:
  \begin{itemize}
    \item start with $\pconf$;
    \item for each $v \in \toPebble$ in a topological order,
      pebble $v$;
    \item for each $v \in \toUnpebble$ in a reverse topological order,
      unpebble $v$.
  \end{itemize}
  Furthermore, the \reasonablepebbling between two configurations
  $\pconf_1$, $\pconf_2$ is the \reasonablepebbling with
  $\pconf=\pconf_1$, $\toPebble=\pconf_2 \setminus \pconf_1$, and
  $\toUnpebble=\pconf_1 \setminus \pconf_2$.
\end{definition}

\begin{claim}[Legality]
  \label{cla:Legality}
  Assume $\pconf, \toPebble, \toUnpebble \subseteq V(G)$ are given.
  If for every $v \in \toPebble \union \toUnpebble$, we have
  $\prednode[G]{v} \subseteq \pconf \union \toPebble$, then the
  \reasonablepebbling is legal.
\end{claim}
\begin{proof}
  For any $v \in \toPebble$, right before $v$ is pebbled,
  we know $\prednode[G]{v}$ have pebbles since
  $\prednode[G]{v} \subseteq \pconf \union \toPebble$ and the pebble
  placement on $\toPebble$ proceeds in a topological order.
  So all pebble placements on $\toPebble$ are legal.

  For any $v \in \toUnpebble$, right before $v$ is unpebbled,
  we know $\prednode[G]{v}$ have pebbles since
  $\prednode[G]{v} \subseteq \pconf \union \toPebble$ and the pebble
  placement on $\toUnpebble$ proceeds in a reverse topological order.
  So all pebble removals on $\toUnpebble$ are legal.
\end{proof}

\begin{corollary}[Legality]
  \label{cor:Legality}
  Assume $\pconf_1, \pconf_2 \subseteq V(G)$ are given.  Assume the
  sets $\toPebble \defeq \pconf_2 \setminus \pconf_1$ and $\toUnpebble
  \defeq \pconf_1 \setminus \pconf_2$ satisfy that for every $v \in
  \toPebble \union \toUnpebble = \pconf_1 \symmdiff \pconf_2$, we have
  $\prednode[G]{v} \subseteq \pconf_1 \union \toPebble = \pconf_1
  \union \pconf_2$, then the \reasonablepebbling over $\pconf_1$,
  $\toPebble$, $\toUnpebble$ is legal.
\end{corollary}
\begin{proof}
  Apply \refclaim{cla:Legality} on $\pconf_1, \toPebble, \toUnpebble$.
\end{proof}

\begin{algorithm}
  \label{alg:reasonable}
  As \refclaim{cla:Correctness} shows, we only need to consider the case when an output node $(v_1, v_2)\outblowup$
  is pebbled or unpebbled to get to $\pconfBTime{\timet}$ in
  $\pebblingBlowUp$. In this case we run the following steps in sequence:
  \begin{enumerate}[(a)]
  \item 
    \label{stepIIa}
    if the $v_1$-block is surrounded in $\pconfBTime{\timet}$, we insert
    the \reasonablepebbling between $\IntInvariantSurroundedFunc{v_1}{\timet-1}$
    and $\IntInvariantSurroundedFunc{v_1}{\timet}$ into $\pebblingBlock{v_1}$.
  \item if $v_2 = z_2$,
    \label{stepIIb}
    \begin{enumerate}[(i)]
    \item 
      \label{stepIIbi}
      if $(v_1, v_2)\outblowup$ is pebbled to get to
      $\pconfBTime{\timet}$ in $\pebblingBlowUp$, for each successor $w_1$ of
      $v_1$ such that the $w_1$-block is surrounded in
      $\pconfBTime{\timet}$, we insert the \reasonablepebbling between $\IntInvariantUnsurroundedFunc{w_1}{\timet-1}$ and $\IntInvariantSurroundedFunc{w_1}{\timet}$ into $\pebblingBlock{w_1}$.
    \item 
      \label{stepIIbii}
      if $(v_1, v_2)\outblowup$ is unpebbled to get to
      $\pconfBTime{\timet}$ in $\pebblingBlowUp$, for each successor $w_1$ of
      $v_1$ such that the $w_1$-block is surrounded in
      $\pconfBTime{\timet-1}$, we insert the \reasonablepebbling between $\IntInvariantSurroundedFunc{w_1}{\timet-1}$ and $\IntInvariantUnsurroundedFunc{w_1}{\timet}$ into $\pebblingBlock{w_1}$.
    \end{enumerate}
  \item
    \label{stepIIc}
    if $\extInvariantFunc{\timet-1} \ne \extInvariantFunc{\timet}$, let
    $\prePebbleUnpebble
    \defeq \visproj{\timet-1} \union \visproj{\timet}$.
    We insert the \reasonablepebbling with
    $\toPebble \defeq \prePebbleUnpebble \setminus
    \extInvariantFunc{\timet-1}$, and
    $\toUnpebble \defeq \prePebbleUnpebble \setminus
    \extInvariantFunc{\timet}$ into $\pebbling$.
  \end{enumerate}
\end{algorithm}

\begin{claim}[Correctness]
  \label{cla:Correctness}
  The pebblings $\pebblingBlocks$ and $\pebbling$ are legal.
\end{claim}

\begin{proof}
  At the beginning $t = 0$, we know $\pconfBTime{0} = \emptyset$, so
  for any $w_1 \in V(G_1)$ we have
  $\IntInvariantSurroundedFunc{w_1}{0} = \emptyset$ and
  $\IntInvariantUnsurroundedFunc{w_1}{0} = \emptyset$, matching
  $\pconfBlockTime{w_1}{0} = \emptyset$.
  Also $\perproj{0} = \emptyset = \visproj{0}$, and
  $\pconfTime{0} = \emptyset = \extInvariantFunc{0}$.

  For any two consecutive configurations we show that either they are the same or the algorithm inserted a \reasonablepebbling between them, which is legal by \refclaim{cla:InvariantImpliesLegal}.

  When $t > 0$, if an exterior node or an interior node is pebbled or
  unpebbled to get to $\pconfBTime{\timet}$ in $\pebblingBlowUp$, fix any $w_1
  \in V(G_1)$.  Note that the $w_1$-block is surrounded in
  $\pconfBTime{\timet-1}$ iff it is surrounded in $\pconfBTime{\timet}$.

  If the $w_1$-block is surrounded, then $\outproj{w_1}{\cdot}$ is the
  same at $\pconfBTime{\timet-1}$ and $\pconfBTime{\timet}$, and
  $\VisitPartInterior{w_1}{\cdot}$ is also the same at
  $\pconfBTime{\timet-1}$ and $\pconfBTime{\timet}$ because a node being
  pebbled or unpebbled must be surrounded, and the predecessors of an
  interior node are output nodes.  Hence
  $\pconfBlockTime{w_1}{\timet-1} = \IntInvariantSurroundedFunc{w_1}{\timet-1} =
  \IntInvariantSurroundedFunc{w_1}{\timet} = \pconfBlockTime{w_1}{\timet}$.

  Otherwise the $w_1$-block is not surrounded, then
  the pebble move is not on any exterior node in the $w_1$-block, hence
  $\pebClosure{\cdot}$ is the same at
  $\pconfBTime{\timet-1}$ and $\pconfBTime{\timet}$
  over the $w_1$-block, so
  $\IntInvariantFunc{w_1}{\pebClosure{\cdot}}$ is also the same at
  $\pconfBTime{\timet-1}$ and $\pconfBTime{\timet}$.

  Since $\extInvariantFunc{\cdot}$ only depends on $\pconfBlocksTime{\timet}$,
  it holds that $\pconf_{\timet-1}=\pconf_{\timet}$ as well.

  Focus on the case when an output node $(v_1, v_2)\outblowup$ is
  pebbled or unpebbled to get to $\pconfBTime{\timet}$ in $\pebblingBlowUp$,
  fix any $w_1 \in V(G_1)$.

  \begin{itemize}
  \item If $w_1 \ne v_1$ and either $w_1$ is not a successor of $v_1$
    or $v_2 \ne z_2$, note that the $w_1$-block is surrounded in
    $\pconfBTime{\timet-1}$ iff it is surrounded in $\pconfBTime{\timet}$.
    Since $\pconfBTime{\timet-1} = \pconfBTime{\timet}$ over the
    $w_1$-block, we have $\IntInvariantSurroundedFunc{w_1}{\timet-1} =
    \IntInvariantSurroundedFunc{w_1}{\timet}$ and
    $\IntInvariantUnsurroundedFunc{w_1}{\timet-1} =
    \IntInvariantUnsurroundedFunc{w_1}{\timet}$.

  \item If $w_1 = v_1$, then
    the $w_1$-block is surrounded in $\pconfBTime{\timet-1}$ iff it is
    surrounded in $\pconfBTime{\timet}$.
    If the $w_1$-block is not surrounded, since
    the pebble move is not on any exterior node in the $w_1$-block, we have
    $\pebClosure{\cdot}$ is the same at
    $\pconfBTime{\timet-1}$ and $\pconfBTime{\timet}$
    over the $w_1$-block, so
    $\IntInvariantFunc{w_1}{\pebClosure{\cdot}}$ is also the same at
    $\pconfBTime{\timet-1}$ and $\pconfBTime{\timet}$.
      
    Otherwise the $w_1$-block is surrounded and Step~\ref{stepIIa} inserts a
    \reasonablepebbling into $\pebblingBlock{w_1}$.

  \item If $w_1$ is a successor of $v_1$ and
    $v_2 = z_2$
    the $w_1$-block cannot be surrounded
    in both $\pconfBTime{\timet-1}$ and $\pconfBTime{\timet}$.

    If the $w_1$-block is unsurrounded
    in both $\pconfBTime{\timet-1}$ and $\pconfBTime{\timet}$,
    since
    $\pconfBTime{\timet-1} = \pconfBTime{\timet}$
    over the $w_1$-block, we have
    $\IntInvariantUnsurroundedFunc{w_1}{\timet-1}
    = \IntInvariantUnsurroundedFunc{w_1}{\timet}$.

    Otherwise $w_1$ is surrounded in exactly one of
    $\pconfBTime{\timet-1}$ and $\pconfBTime{\timet}$ and Step~\ref{stepIIb} inserts a
    \reasonablepebbling into $\pebblingBlock{w_1}$.
  \end{itemize}

Finally, Step~\ref{stepIIc} inserts a \reasonablepebbling into $\pebbling$ as needed.

\end{proof}

\begin{claim}[Exterior Symmetry]
  \label{cla:ExteriorSymmetry}
  In Step~\ref{stepIIc}, we have
  $\extInvariantFunc{\timet-1} \union \toPebble
  = \prePebbleUnpebble
  = \extInvariantFunc{\timet} \union \toUnpebble$.
\end{claim}
\begin{proof}
  Since $\perproj{\timet} \subseteq \visproj{\timet}$, we have
  $\extInvariantFunc{\timet}
  = \ExtInvariant{\timet} \subseteq \visproj{\timet}
  \subseteq \prePebbleUnpebble$ and likewise
  $\extInvariantFunc{\timet-1} \subseteq \visproj{\timet-1}
  \subseteq \prePebbleUnpebble$.
\end{proof}

\begin{claim}[Legality]
  \label{cla:InvariantImpliesLegal}
  All \reasonablepebbling{}s are legal.
\end{claim}
\begin{proof}
  At Step~\ref{stepIIa} 
  an output node $(v_1, v_2)\outblowup$ is pebbled or unpebbled to get to
  $\pconfBTime{\timet}$ in $\pebblingBlowUp$ and
  the $v_1$-block is surrounded, hence
  $\pconfBlockTime{v_1}{\timet-1}
  = \IntInvariantSurroundedFunc{v_1}{\timet-1}$.
  By \refcor{cor:Legality}
  it suffices to show that
  for $w_2 \in
  \IntInvariantSurroundedFunc{v_1}{\timet}
  \symmdiff
  \IntInvariantSurroundedFunc{v_1}{\timet-1}$, we have
  $\prednode[G_2]{w_2} \subseteq
  \IntInvariantSurroundedFunc{v_1}{\timet}
  \union
  \IntInvariantSurroundedFunc{v_1}{\timet-1}$.
  Recall
  $\IntInvariantFunc{v_1}{\cdot} = \IntInvariant{v_1}{\cdot}$.

      Assume $w_2 \in
      \IntInvariantSurroundedFunc{v_1}{\timet}
      \setminus
      \IntInvariantSurroundedFunc{v_1}{\timet-1}$, reversing the roles of $\timet-1$ and $\timet$ otherwise.
      If $w_2 \in \outproj{v_1}{\pconfBTimeFunc{\timet}}
      \setminus
      \IntInvariantSurroundedFunc{v_1}{\pconfBTimeFunc{\timet-1}}$, then
      $w_2 \in \outproj{v_1}{\pconfBTimeFunc{\timet}}
      \setminus
      \outproj{v_1}{\pconfBTimeFunc{\timet-1}}$, \ie
      $w_2 = v_2$ and
      $(v_1, w_2)\outblowup$ is being pebbled to get to $\pconfBTime{\timet}$ in
      $\pebblingBlowUp$.
      Since $\pebblingBlowUp$ is legal,
      $(v_1, w_2)\intblowup$ has a pebble in $\pconfBTime{\timet-1}$, \ie
      $w_2 \in \intproj{v_1}{\pconfBTimeFunc{\timet-1}}$, so
      $w_2$ is surrounded by
      $\outproj{v_1}{\pconfBTimeFunc{\timet-1}}
      \subseteq \IntInvariantSurroundedFunc{v_1}{\timet-1}$ as needed.
      
      Otherwise $w_2 \in
      \Bigl(
      \VisitPartInteriorTime{v_1}{\timet}
      \Bigr)
      \setminus
      \IntInvariantSurroundedFunc{v_1}{\timet-1}$, since
      $\intproj{v_1}{\pconfBTimeFunc{\timet}}
      =
      \intproj{v_1}{\pconfBTimeFunc{\timet-1}}$, we know
      $w_2 \in \UnsurroundedOutput{v_1}{\timet}$
      but
      $w_2 \notin \UnsurroundedOutput{v_1}{\timet-1}$, so
      $w_2$ is a successor of $v_2$ and
      $(v_1, v_2)\outblowup$ is being unpebbled to get to $\pconfBTimeFunc{\timet}$ in
      $\pebblingBlowUp$ and
      $\prednode[G_2]{w_2}
      \subseteq \outproj{v_1}{\pconfBTimeFunc{\timet-1}}
      \subseteq \IntInvariantSurroundedFunc{v_1}{\timet-1}$ as needed.

  At Step~\ref{stepIIb}\ref{stepIIbi}
  an output node $(v_1, v_2)\outblowup$ is pebbled to get to
  $\pconfBTime{\timet}$ in $\pebblingBlowUp$, and
  $w_1$ is a successor of $v_1$ such that
  the $w_1$-block is surrounded in $\pconfBTime{\timet}$.
  Note that
  the $w_1$-block is not surrounded in $\pconfBTime{\timet-1}$, so
  $\pconfBlockTime{w_1}{\timet-1}
  = \IntInvariantUnsurroundedFunc{w_1}{\timet-1}$.
  Since the pebble move to get to $\pconfBTime{\timet}$ in $\pebblingBlowUp$ is not in
  the $w_1$-block, we have
  $\IntInvariantUnsurroundedFunc{w_1}{\timet-1}
  =
  \IntInvariantUnsurroundedFunc{w_1}{\timet}$.
  By \refcor{cor:Legality}
  it suffices to show that
  for $w_2 \in
  \IntInvariantSurroundedFunc{w_1}{\timet}
  \symmdiff
  \IntInvariantUnsurroundedFunc{w_1}{\timet}$, we have
  $\prednode[G_2]{w_2}
  \subseteq \IntInvariantSurroundedFunc{w_1}{\timet}
  \union \IntInvariantUnsurroundedFunc{w_1}{\timet}$.
  \begin{itemize}
    \item
      Assume $w_2 \in
      \IntInvariantSurroundedFunc{w_1}{\timet}
      \setminus
      \IntInvariantUnsurroundedFunc{w_1}{\timet}$.
      Since
      $\outproj{w_1}{\pconfBTimeFunc{\timet}}
      \subseteq \outproj{w_1}{\pebClosure{\pconfBTimeFunc{\timet}}}
      \subseteq \IntInvariantUnsurroundedFunc{w_1}{\timet}$,
      we can assume
      $w_2 \in
      \Bigl( \VisitPartInteriorTime{w_1}{\timet} \Bigr)
      \setminus
      \IntInvariantUnsurroundedFunc{w_1}{\timet}$.
      Since
      $\intproj{w_1}{\pconfBTimeFunc{\timet}} 
      \subseteq \intproj{w_1}{\pebClosure{\pconfBTimeFunc{\timet}}}$, we know
      $w_2$ is surrounded by
      $\outproj{w_1}{\pebClosure{\pconfBTimeFunc{\timet}}} 
      \subseteq \IntInvariantUnsurroundedFunc{w_1}{\timet}$ as needed.
    \item
      Assume $w_2 \in
      \IntInvariantUnsurroundedFunc{w_1}{\timet}
      \setminus
      \IntInvariantSurroundedFunc{w_1}{\timet}$.
      We claim that
      $w_2 \in \intproj{w_1}{\pebClosure{\pconfBTimeFunc{\timet}}}$:
      if
      $w_2 \in
      \Bigl(
      \VisitPartInterior{w_1}{\pebClosure{\pconfBTimeFunc{\timet}}}
      \Bigr)
      \setminus
      \IntInvariantSurroundedFunc{w_1}{\timet}$, then
      we are done;
      otherwise
      $w_2 \in
      \outproj{w_1}{\pebClosure{\pconfBTimeFunc{\timet}}}
      \setminus
      \IntInvariantSurroundedFunc{w_1}{\timet}$, then
      $w_2 \notin
      \outproj{w_1}{\pconfBTimeFunc{\timet}}$ and thus
      by definition of $\pebClosure{\cdot}$ we have
      $w_2 \in \intproj{w_1}{\pebClosure{\pconfBTimeFunc{\timet}}}$ as claimed.
      Now if
      $w_2 \in \intproj{w_1}{\pconfBTimeFunc{\timet}}$, then
      $w_2$ is surrounded by
      $\outproj{w_1}{\pconfBTimeFunc{\timet}}
      \subseteq \IntInvariantSurroundedFunc{w_1}{\timet}$ as needed;
      otherwise
      $w_2 \notin \intproj{w_1}{\pconfBTimeFunc{\timet}}$, then from
      $w_2 \in \intproj{w_1}{\pebClosure{\pconfBTimeFunc{\timet}}}$ and
      by definition of $\pebClosure{\cdot}$ we have
      $w_2$ is surrounded by
      $\outproj{w_1}{\pebClosure{\pconfBTimeFunc{\timet}}}
      \subseteq \IntInvariantUnsurroundedFunc{w_1}{\timet}$ as needed.
  \end{itemize}

  To see that Step~\ref{stepIIb}\ref{stepIIbii} is legal,
  note that it is the reverse of Step~\ref{stepIIb}\ref{stepIIbi}.

  At Step~\ref{stepIIc}
  an output node $(v_1, v_2)\outblowup$ is pebbled or unpebbled to get to
  $\pconfBTime{\timet}$ in $\pebblingBlowUp$.
  By \refclaim{cla:Legality}
  it suffices to show that
  for $w_1 \in
  \bigl( \prePebbleUnpebble \setminus \extInvariantFunc{\timet-1} \bigr)
  \union
  \bigl( \prePebbleUnpebble \setminus \extInvariantFunc{\timet} \bigr)$,
  we have
  $\prednode[G_1]{w_1}
  \subseteq \extInvariantFunc{\timet-1} \union \toPebble
  = \prePebbleUnpebble$,
  where the last equality is due to \refclaim{cla:ExteriorSymmetry}.
  Recall that
  $\extInvariantFunc{\cdot} = \ExtInvariant{\cdot}$ and
  $\prePebbleUnpebble = \visproj{\timet-1} \union \visproj{\timet}$.

  Fix any $w_1 \in
  \bigl( \prePebbleUnpebble \setminus \extInvariantFunc{\timet-1} \bigr)
  \union
  \bigl( \prePebbleUnpebble \setminus \extInvariantFunc{\timet} \bigr)$.
  Swap the roles of $\timet$ and $\timet-1$ if necessary, we can assume
  $w_1 \in \prePebbleUnpebble \setminus \extInvariantFunc{\timet}$.
  If $w_1 \in \visproj{\timet} \setminus \extInvariantFunc{\timet}$, then
  $w_1$ is surrounded by
  $\perproj{\timet}
  \subseteq \extInvariantFunc{\timet}
  \subseteq \prePebbleUnpebble$ as needed.
  
  Otherwise $w_1 \in \visproj{\timet-1} \setminus \visproj{\timet}$, in particular
  $\pconfBlockTime{w_1}{\timet-1} \ne \pconfBlockTime{w_1}{t}$, so
  $w_1 \in \{ v_1 \} \union \succnode[G_1]{v_1}$ by design of the algorithm.
  Note that it suffices to show that the $w_1$-block is surrounded in
  $\pconfBTime{\timet-1}$ or in $\pconfBTime{\timet}$:
  if the $w_1$-block is surrounded in $\pconfB$,
  for any $u_1 \in \prednode[G_1]{w_1}$
  the sink $z_2$ of the inner graph $G_2$ satisfies
  $z_2 \in \outproj{u_1}{\pconfB}
  \subseteq \IntInvariantFunc{u_1}{\pconfB}
  \subseteq \intmappingBlock{u_1}(\pconfB)$, so
  $u_1 \in \visproj{\pconfB} \subseteq \prePebbleUnpebble$ as needed.
  If $w_1 = v_1$, then
  in Step~\ref{stepIIa} the $w_1$-block is surrounded in $\pconfBTime{\timet}$ as
  needed.
  Otherwise $w_1 \in \succnode[G_1]{v_1}$, then
  in Step~\ref{stepIIb}
  the $w_1$-block is surrounded in $\pconfBTime{\timet}$ (in Step~\ref{stepIIb}\ref{stepIIbi}) or
  in $\pconfBTime{\timet-1}$ (in Step~\ref{stepIIb}\ref{stepIIbii}) as needed.
\end{proof}

\begin{claim}[No Spurious Projections]
  \label{cla:NoMagicPebbles}
  For any $w_1 \in V(G_1)$ and $w_2 \in V(G_2)$, if $w_2$ has a pebble
  in any configuration of $\pebblingBlock{w_1}$ created between time
  $\timet-1$ and time $\timet$, \ie between
  $\pconfBlockTime{w_1}{\timet-1}$ and $\pconfBlockTime{w_1}{\timet}$,
  then the $(w_1, w_2)$-cell has a pebble both in $\pconfBTime{\timet-1}$
  and in $\pconfBTime{\timet}$ in $\pebblingBlowUp$.
\end{claim}

\begin{proof}
  Let us reword the claim.
  Define $\anyproj{w_1}{\pconfB}
  \defeq \outproj{w_1}{\pconfB}
  \union \extproj{w_1}{\pconfB}
  \union \intproj{w_1}{\pconfB}$, note that
  $\anyproj{w_1}{\pconfBTimeFunc{\timet-1}} = \anyproj{w_1}{\pconfBTimeFunc{\timet}}$, and
  we want to show
  $\pconfBlockTime{w_1}{\timet-1}\union\toPebble\block{w_1} \subseteq \anyproj{w_1}{\pconfBTimeFunc{\timet}}$.
  
  Note that for any $\pconfB$ and $w_1 \in V(G_1)$ we have
  $\IntInvariantFunc{w_1}{\pconfB}
  \subseteq \anyproj{w_1}{\pconfB}$.
  Moreover, we have
  $\IntInvariantFunc{w_1}{\pebClosure{\pconfB}}
  \subseteq \anyproj{w_1}{\pconfB}$:
  for any $w_2 \in \IntInvariantFunc{w_1}{\pebClosure{\pconfB}}$,
  if $w_2 \in \outproj{w_1}{\pebClosure{\pconfB}}$, then either
  $w_2 \in \outproj{w_1}{\pconfB}
  \subseteq \anyproj{w_1}{\pconfB}$ as needed, or
  $w_2 \notin \outproj{w_1}{\pconfB}$, then
  by definition of $\pebClosure{\cdot}$ we have
  $w_2 \in \extproj{w_1}{\pconfB}
  \subseteq \anyproj{w_1}{\pconfB}$ as needed.
  Otherwise $w_2 \in \VisitPartInterior{w_1}{\pebClosure{\pconfB}}$,
  then it follows that
  $w_2 \in \intproj{w_1}{\pconfB}
  \subseteq \anyproj{w_1}{\pconfB}$ as needed;
  for otherwise $w_2 \notin \intproj{w_1}{\pconfB}$, thus
  by definition of $\pebClosure{\cdot}$ we have
  $\prednode[G_2]{w_2}
  \subseteq \outproj{w_1}{\pebClosure{\pconfB}}$,
  which contradicts
  $w_2 \in \Unsurrounded{w_1}{\outproj{w_1}{\pebClosure{\pconfB}}}$.

  To prove the claim it is enough to recall that the set
  $\pconfBlockTime{w_1}{\timet-1}\union\toPebble\block{w_1}$ equals
$\IntInvariantFunc{w_1}{\pconf_{\timet-1}}
  \union \IntInvariantFunc{w_1}{\pconf_{\timet}}$ in Step~\ref{stepIIa} because the $w_1$-block is surrounded, it equals 
$\IntInvariantFunc{w_1}{\pebClosure{\pconf_{\timet-1}}}
  \union \IntInvariantFunc{w_1}{\pconf_{\timet}}$ in Step~\ref{stepIIb}\ref{stepIIbi}), and it equals
  $\IntInvariantFunc{w_1}{\pconf_{\timet-1}}
  \union \IntInvariantFunc{w_1}{\pebClosure{\pconf_{\timet}}}$ in Step~\ref{stepIIb}\ref{stepIIbii}.
  All three sets are contained in $\anyproj{w_1}{\pconfBTimeFunc{\timet}}$ by the previous paragraph.
\end{proof}

\begin{claim}[Legal implies Lower Bound]
  \label{cla:LegalImpliesLowerBound}
  Assume that
  the pebblings $\pebblingBlocks$ and $\pebbling$ are legal, and
  $\pebbling$ is a persistent pebbling of $G_1$.
  Then there is a configuration $\pconfBlowUpBottleneck$ in $\pebblingBlowUp$ of
  $\blowupConst{G_1}{G_2}$ having at least $\persistentprice{G_1} +
  \persistentprice{G_2} + 1$ pebbles.
\end{claim}
\begin{proof}
  Since $\pebbling$ is a legal persistent pebbling of $G_1$,
  there is a configuration $\pconfBottleneck$ having at least
  $\persistentprice{G_1}$ pebbles.
  By definition of the algorithm $\pconfBottleneck$ is created in Step~\ref{stepIIc} when
  an output node $(v_1, v_2)\outblowup$ is pebbled or unpebbled to get to
  $\pconfBTime{\timet}$ in $\pebblingBlowUp$, so the configuration $\pconfBottleneck$
  is between $\pconfTime{\timet-1}$ and $\pconfTime{\timet}$.
  Let $\pconfBlowUpBottleneck$ be $\pconfBTime{\timet}$ if
  the output node $(v_1, v_2)\outblowup$ is pebbled, and
  $\pconfBTime{\timet-1}$ if
  the output node $(v_1, v_2)\outblowup$ is unpebbled, then
  $\pconfBlowUpBottleneck = \pconfBTime{\timet-1} \union \pconfBTime{\timet}$.

  Note that for any $u_1 \in \pconfBottleneck$ we know that between
  $\pconfBlockTime{u_1}{\timet-1}$ and $\pconfBlockTime{u_1}{\timet}$,
  some configuration in $\pebblingBlock{u_1}$ is not empty, because
  $u_1 \in \extInvariantFunc{\timet-1} \union
  \extInvariantFunc{\timet} \subseteq D = \visproj{\timet-1} \union
  \visproj{\timet}$. It follows from \refclaim{cla:NoMagicPebbles} that
  there is at least one pebble in every $u_1$-block of
  $\pconfBlowUpBottleneck$.

  If $\perproj{\timet-1} \neq \perproj{\timet}$, that is if there is
  $w_1 \in \perproj{\timet-1} \symmdiff \perproj{\timet}$, then
  we are done. Since exactly one of
  $\pconfBlockTime{w_1}{\timet-1}$ and $\pconfBlockTime{w_1}{\timet}$ is \perlocked,
  there is a configuration between them having at least
  $\persistentprice{G_2}$ pebbles, and by \refclaim{cla:NoMagicPebbles} at least $\persistentprice{G_2}$ cells of the $w_1$-block have pebbles.
  Summing up, in $\pconfBlowUpBottleneck$, there are
  \begin{equation*}
  \underbrace{\persistentprice{G_2}}_{w_1\text{ block}}
  +
  \underbrace{\persistentprice{G_1} - 1}_{\text{other blocks}}
  \end{equation*}
  cells having pebbles, and
  among them the $(v_1, v_2)$-cell has three pebbles, for a total of
  $\persistentprice{G_1} + \persistentprice{G_2} + 1$ pebbles as needed.

  Assume that $\perproj{\timet-1}=\perproj{\timet}$ instead.
  Recall
  $\extInvariantFunc{\cdot} = \ExtInvariant{\cdot}$.
  In Step~\ref{stepIIc} we know
  $\extInvariantFunc{\timet-1} \ne \extInvariantFunc{\timet}$, therefore at least one of
  $\preToPebble
  \defeq \extInvariantFunc{\timet} \setminus \extInvariantFunc{\timet-1}$ and
  $\preToUnpebble
  \defeq \extInvariantFunc{\timet-1} \setminus \extInvariantFunc{\timet}$ is not empty.
  Swap $\timet$ and $\timet-1$ if necessary,
  by their symmetry in the rest of this lemma,
  assume $\preToUnpebble \ne \emptyset$.
  Fix $w_1 \in \preToUnpebble
  = \extInvariantFunc{\timet-1} \setminus \extInvariantFunc{\timet}$.

  Since $\perproj{\timet-1} \setminus \extInvariantFunc{\timet} \subseteq \perproj{\timet-1} \setminus \perproj{\timet} = \emptyset$, we can assume that $w_1 \in \bigl( \visitPartExterior{\timet-1} \bigr)
  \setminus \extInvariantFunc{\timet}$.

  In particular $w_1 \in \bigl( \visproj{\timet-1}\setminus\visproj{\timet} \bigr) \intersection \unsurrounded{G_1}{\perproj{\timet}}$.
  Let $u_1 \in \prednode[G_1]{w_1}$ be such that $u_1 \notin
  \perproj{\timet-1}=\perproj{\timet}$.

  Now
  $w_1 \in \visproj{\timet-1} \setminus \visproj{\timet}$, and
  as in the proof of \refclaim{cla:InvariantImpliesLegal}, we know that
  the $w_1$-block is surrounded
  in $\pconfBTime{\timet-1}$ or in $\pconfBTime{\timet}$.
This implies that the sink of the $u_1$-block $(u_1,z_2)\outblowup$ has a pebble, and so $z_2 \in \pconfBlockTime{u_1}{\timet-1}$ or $z_2 \in \pconfBlockTime{u_1}{\timet}$. But $u_1$ is not in the \perprojText at either time, so there is some other pebble other than $z_2$ in $\pconfBlockTime{u_1}{\timet-1}$ and $\pconfBlockTime{u_1}{\timet}$. By \refclaim{cla:NoMagicPebbles}, at least 2 cells of the $u_1$-block have pebbles.

  Since
  $w_1 \in \visproj{\timet-1} \setminus \visproj{\timet}$, exactly one of
  $\pconfBlockTime{w_1}{\timet-1}$ and $\pconfBlockTime{w_1}{\timet}$ is \vislocked so
  there is a configuration of $\pebblingBlock{w_1}$ between them having at least $\visitprice{G_2}$
  pebbles, and by 
 \refclaim{cla:NoMagicPebbles} at least $\visitprice{G_2}$ cells of the $w_1$-block have pebbles. Summing up, in $\pconfBlowUpBottleneck$, there are
  \begin{equation*}
  \underbrace{\visitprice{G_2}}_{w_1\text{ block}}
  + \underbrace{\vphantom{\visitprice{G_2}}2}_{u_1\text{ block}}
  + \underbrace{\persistentprice{G_1} - 2}_{\text{other blocks}} \ge
  \persistentprice{G_1} + \persistentprice{G_2} - 1
  \end{equation*}
  cells having pebbles, and
  among them the $(v_1, v_2)$-cell has three pebbles, for a total of
  $\persistentprice{G_1} + \persistentprice{G_2} + 1$ pebbles as needed.
\end{proof}

\mvcomment{Do we want the counterexample for black pebbling anymore?}

This product construction does not work for the \pstandard black
pebbling: there is no single constant $C$ such that for any two graphs
$G_1$ and $G_2$,
$\blackprice{\blowupConst{G_1}{G_2}}=\blackprice{G_1}+\blackprice{G_2}+C$,
therefore we need to develop a different construction in the next
section.

Indeed, if we take $G_1$ and $G_2$ to be the singleton graph, which
has pebbling price $1$, the product construction is the pyramid of
height 1, which has pebbling price $3$. This gives a value of $1$ for $C$.

If we take $G_1$ and $G_2$ to be the path of length $1$, which has
pebbling price $2$, the product construction has pebbling price $4$. This
gives a value of $0$ for $C$.

If we take $G_1$ and $G_2$ to be the pyramid of height $1$, which has
pebbling price $3$, the product construction has pebbling price
$5$. An optimal strategy is to pebble the sinks of the two $G_2$
copies corresponding to sources in $G_1$, then pebble all the exterior
nodes of the remaining copy of $G_2$, unpebble the sinks and finish
the pebbling. This gives a value of $-1$ for $C$.

\section{Product Construction for \Pstandard Pebbling}
\label{sec:blow-up-standard}

The part of the proof of Theorem~\ref{th:approx-both} that deals with
\pstandard pebbling uses as a black box the construction in
Theorem~\ref{th:blowup-both} for \pstandard pebbling. 
Now we state it again and we give its full proof.

\begin{theorem}
  \label{th:blowupblack}
  Given two graphs $\graphstd_1$ and $\graphstd_2$ of \pstandard
  pebbling price at least $3$, 
  there is a \efficiently constructible graph 
  $\blowupBlackConst{\graphstd_1}{\graphstd_2}$ 
  of size 
  $\setsize{{\graphstd_1}}(2\setsize{\graphstd_1}+\setsize{\graphstd_2})$ 
  with \pstandard pebbling price
  $\blackprice{\blowupBlackConst{\graphstd_1}{\graphstd_2}} = \blackprice{\graphstd_1} + \blackprice{\graphstd_2} - 1$.
\end{theorem}

For the rest of the section we fix $G_1$ and $G_2$ to be two single
sink directed acyclic graphs, with sinks $\sink_{1}$
and $\sink_{2}$, respectively. We set $p_{1}\defeq
\blackprice{G_{1}}$ and $p_{2}\defeq \blackprice{G_{2}}$, and we
assume that $p_{2}$ is at least $3$.

\begin{construction}[Product for \pstandard pebbling]
\label{con:BlowUpBlack}
A \introduceterm{centipede} of length $\ell$
is a
path of length $\ell$ where all nodes but the source have an extra
predecessor.
Formally, it is the graph with vertices $\set{r_0,\ldots, r_\ell}$,
$\set{s_1, \ldots, s_\ell}$ and edges
$\set{(r_{i-1},r_i):i\in[\ell]} \cup \set{(s_i,r_i):i\in[\ell]}$.

As the first step we define the graph $\pointygraph{G_{2}}$ from
$G_{2}$ as follows:
$\pointygraph{G_{2}}$ is the union of
$G_{2}$ and of a centipede of length $|G_{1}|$, where we identify the sink of
$G_{2}$ with the vertex $r_{0}$ of the centipede.
Observe that the pebbling price of $\pointygraph{G_{2}}$ is
$\max(3,p_{2})=p_{2}$.

The graph $\blowupBlackConst{G_{1}}{G_{2}}$ is as follows.
For every vertex
$v_1$ of $G_1$ we make a copy of $\pointygraph{G_2}$, which
we call the $v_1$-block. Then, for every edge $(u_1,v_1)$ in $G_1$, we
add edges from the sink of the $u_1$-block to all the sources of the
centipede in the $v_1$-block.
Formally, $\blowupBlackConst{G_{1}}{G_{2}}$ is the graph with vertices $\set{(v_1,v_2):v_1\in
  V(G_1),v_2\in V(\pointygraph{G_2})}$ and edges
$\set{((v_1,u_2),(v_1,v_2)):v_1\in V(G_1),(u_2,v_2)\in E(G_2)}$,
$\set{((u_1,r_{|G_{1}|}),(v_1,s_i)):(u_1,v_1)\in E(G_1),i\in[\setsize{G_1}]}$.
\end{construction}

\begin{figure}[ht]
  \begin{center}
    \begin{tikzpicture}[scale=.75]
      \BlowUpBlackExample
    \end{tikzpicture}
      \caption{Example of \refcon{con:BlowUpBlack}: product of a pyramid of height 1 and a rhombus.}
      \label{fig:ExampleBlowUpBlack}
  \end{center}
\end{figure}

\begin{lemma}
  The \pstandard pebbling price of $\blowupBlackConst{G_{1}}{G_{2}}$ is at most $p_{1}+p_{2}-1$.
\end{lemma}

\begin{proof}
  To pebble $\blowupBlackConst{G_{1}}{G_{2}}$ we simulate a strategy for pebbling $G_1$ with $p_1$ pebbles. When a
  pebble is placed in $v_1$, we pebble the sink of the $v_1$-block using
  a strategy for $\pointygraph{G_{2}}$ in $p_{2}$ pebbles. 
  When a pebble is removed from $v_1$, we remove the pebble from the sink
  of the $v_1$-block.
  When we put a pebble in some $v_{1}$-block there are at most $p_1-1$
  other non empty blocks and they all have exactly one pebble,
  therefore this strategy for $\blowupBlackConst{G_{1}}{G_{2}}$ is within the
  pebbling limit.
\end{proof}

\begin{lemma}
  \label{lmm:blowup-black-lowerbound}  
  The \pstandard pebbling price of $\blowupBlackConst{G_{1}}{G_{2}}$ is at least $p_{1}+p_{2}-1$.
\end{lemma}

\begin{proof}
  Given a pebbling 
  $\pebblingBlowUpBlack=(\pconfBlowUpBlack_{0},\pconfBlowUpBlack_{1},\ldots,\pconfBlowUpBlack_{\stoptime})$
  for the
  graph $\blowupBlackConst{G_{1}}{G_{2}}$ 
  we construct a pebbling $\pebbling$ for $G_{1}$
  in such a way that if the space of $\pebblingBlowUpBlack$ is less
  than $p_{1}+p_{2}-1$ then $\pebbling$ has space less than
  $p_{1}$, which is impossible.

  In particular for any configuration $\pconfBlowUpBlack_{t}$ in
  $\pebblingBlowUpBlack$, we build a sequence $\pebblingchunk{t}$ of
  pebbling configurations for $G_{1}$, such that the final pebbling
  $\pebbling$ of $G_{1}$ is the concatenation of
  $\pebblingchunk{0},\pebblingchunk{1},\ldots,\pebblingchunk{\stoptime}$.
  While we build these sequences of configurations, we say that a
  vertex $v_{1} \in V(G_{1})$ is \introduceterm{active} if in the last
  configuration built so far $v_{1}$ does not have a pebble while all
  of its predecessors do.

\mvcomment{You do not like surrounded?}

  If $\pconfBlowUpBlack_t$ follows from a pebbling removal after which some
  $v_{1}$-block of $\blowupBlackConst{G_{1}}{G_{2}}$ becomes empty
  then $\pebblingchunk{t}$ is the pebbling step that
  removes the pebble present on $v_{1}$, if any, otherwise
  $\pebblingchunk{t}$ is the empty sequence.

  If $\pconfBlowUpBlack_t$ is the result of a pebble placement after
  which some $v_{1}$-block of $\blowupBlackConst{G_{1}}{G_{2}}$
  contains $p_{2}$ pebbles, then $\pebblingchunk{t}$
  performs the following pebbling steps: place a pebble on each empty vertex
  $w_{1}\in V(G_{1})$ such that
  \begin{enumerate}
    \item\label{item:proj_full} the $w_{1}$-block of $\blowupBlackConst{G_{1}}{G_{2}}$
    contains a pebble in $\pconfBlowUpBlack_t$; and
    \item\label{item:proj_pred} $w_{1}$ is active.
  \end{enumerate}
  This process is repeated until there is no vertex in $V(G_{1})$
  that meets both conditions. In particular a pebbling placement in
  the pebbling $\pebblingBlowUpBlack$ can cause a long chain of
  pebbling placements in $\pebbling$.
 
  Pebbling $\pebbling$ is a legal \pstandard pebbling of $G_{1}$ since
  removals are always legal, and pebbling placements are only done on
  active vertices by construction.

  \begin{claim}
    \label{cla:faithful}
    Assume the pebbling $\pebblingBlowUpBlack$ uses at most
    $p_{1}+p_{2}-1$ pebbles.  If there is a pebble on $(v_{1},r_{|G_{1}|})$ in
    $\pconfBlowUpBlack_t$, then there is a pebble on $v_1$ at the end of
    $\pebblingchunk{t}$.
  \end{claim}
  \begin{proof}
    We prove the claim by induction over a topological order of $G_{1}$.
    Consider the earliest time $t_1$ such that
    $\pebblingBlowUpBlack$ has $p_2$ pebbles in the $v_1$-block and
    there is a pebble in the $v_1$-block during the whole interval
    $[t_1,t]$. 
    Such a time exists because the $v_{1}$-block is a copy of
    $\pointygraph{G_2}$, and $p_{2}$ pebbles are
    necessary to pebble its sink.  
    
    If $v_1$ is active at any point during the construction of
    $\pebblingchunk{t_1}$, then $v_{1}$ is pebbled in that sequence
    and it is not removed afterwards. This is always the case if
    $v_{1}$ is a source of $G_{1}$.
    
    If $v_{1}$ is not active we assume
    that Claim~\ref{cla:faithful} holds for all its predecessors.
    Time $t_1$ is the first time when there are $p_2$ pebbles in the
    $v_1$-block since it has been empty. Therefore none
    of the successors of $(v_1,\sink_2)$ in the $v_{1}$-block has a
    pebble.
    Also, since at most $p_1-1<\setsize{G_1}$ pebbles are outside the
    $G_2$ part of the $v_1$-block, some vertex $(v_1,s_i)$ in the
    centipede part of the $v_1$-block has no pebble.

    So far we discovered that at time $t_{1}$ there is a path
    $(v_1,s_i),(v_1,r_i),(v_1,r_{i+1})\ldots,(v_1,r_{\setsize{G_1}})$
    with no pebbles, and that at time $t$ vertex
    $(v_1,r_{\setsize{G_1}})$ has a pebble. 
    Then it must be the case that vertex $(v_1,s_i)$ is pebbled at some
    time $t_2$ where $t_1<t_{2}<t$,
    and furthermore at time $t_2$ there must be
    pebbles on $(u_1,r_{|G_{1}|})$ and
    $(w_1,r_{|G_{1}|})$, where $u_{1}$ and $w_{1}$ are the
    predecessors of $v_{1}$ in graph $G_{1}$. 
    By induction hypothesis $u_{1}$ and $w_{1}$ have a pebble at the end
    of $\pebblingchunk{t_{2}}$, so $v_{1}$ is active at that point
    and, since the $v_{1}$-block is not empty, it gets a pebble.
    Such pebble stays in place at least until the end of $\pebblingchunk{t}$.
  \end{proof}
  
We can finally prove Lemma~\ref{lmm:blowup-black-lowerbound}. 
Assume for the sake of contradiction that $\pebblingBlowUpBlack$ uses strictly less than $p_{1}+p_{2}-1$ pebbles.
Pebbling $\pebbling$ is a legal pebbling of $G_{1}$ which pebbles the
sink $\sink_{1}$, because of Claim~\ref{cla:faithful} and the fact
that $\pebblingBlowUpBlack$ pebbles vertex $(\sink_{1},r_{|G_{1}|})$.
Consider a configuration in which
$\pebbling$ reaches its maximum number of pebbles.
This configuration is at the end of a sequence $\pebblingchunk{t}$
corresponding to a pebble placement in $\pebblingBlowUpBlack$ that
causes some $v_{1}$-block to have $p_{2}$ pebbles,
since this is the only case in which a sequence $\pebblingchunk{t}$ adds pebbles.

The corresponding configuration $\pconfBlowUpBlack_{t}$ has $p_{2}$
pebbles in the $v_{1}$-block and at most $p_{1}-2$ other non empty
blocks by assumption. Empty blocks in $\pconfBlowUpBlack_{t}$ corresponds
to empty vertices in $\pebblingchunk{t}$ by construction, so there are
at most $p_{1}-1$ pebbles in all configurations in $\pebblingchunk{t}$. 
This contradicts the fact that $\blackprice{G_{1}}=p_{1}$.
\end{proof}

Observe that the only point where we used the fact that the length of a centipede is $\setsize{G_1}$ is to claim that there is one source without a pebble, so any length $u \geq \blackprice{G_1}$ would suffice. Since in general it is \pspace-hard to compute $\blackprice{G_1}$, we settle for the trivial upper bound $\setsize{G_1}$.

\section{Concluding Remarks}
\label{sec:conclusion}

In this paper, we study the pebble game first introduced
in~\cite{PH70Comparative} as well as the more restricted reversible
pebble game in~\cite{Bennett89TimeSpaceReversible}, where
by~\cite{Chan13JustAPebble} the latter game is also equivalent to the
\dtgametext~\cite{DT85Speedups} and the
\rmgametext~\cite{RM99Separation}.

We establish that it is \PSPACE-hard to approximate \pstandard and
reversible pebbling price up to any additive constant.  To the best of
our knowledge, these are the first hardness of approximation results
for such pebble games, even for polynomial time. It would be very
interesting to show stronger inapproximability results for pebbling
price under stronger assumptions. On the one hand, we are only able to
show additive hardness, but on the other hand our results hold for
arbitrary algorithms using a polynomial amount of memory.  It seems
reasonable to believe that the problem should become much harder for
algorithms restricted to polynomial time, but showing this seems like
a challenging task---in some sense, it appears that pebbling might be
so hard a problem that it is even hard to prove that it is hard.

Another challenging problem is to prove approximation hardness, or
even just \PSPACE-completeness, for the 
\introduceterm{black-white pebble game}~\cite{CS76Storage} modelling
nondeterministic computation.  This game is a strict generalization of
the \pstandard (black) pebble game, and so intuitively it should be at
least as hard, but the added option of placing nondeterministic white
pebbles anywhere in the graph completely destroys locality and makes
the reduction in \cite{GLT80PebblingProblemComplete} break down.
Hertel and Pitassi~\cite{HP10PspaceCompleteness} showed a
\PSPACE-completeness result in the nonstandard setting when unbounded
(and very large) fan-in is allowed. Essentially, the large fan-in
makes it possible to lock down almost all pebbles in one place at a
time (namely on the predecessors of a large fan-in vertex to be
pebbled) and to completely rule out any use of white pebbles, reducing
the whole problem to black pebbling (although this reduction, it
should be stressed, is far from trivial). This approach does not work
for bounded fan-in graphs, however, which is the standard setting
studied in the 1970s and 80s and the setting that could potentially
have interesting applications in, for instance, proof complexity.

We also show in this paper that \pstandard black pebbling is
asymptotically stronger than reversible pebbling by exhibiting
families of DAGs over $n$~vertices which have \pstandard pebblings in
space~$\pebsp$ but for which the reversible pebbling price 
\mbox{is $\bigomega{\pebsp \log n}$.}  Since any DAG on $n$~vertices
with \pstandard pebbling price~$\pebsp$ can be reversibly pebbled in
space $\bigoh{\pebsp^2 \log n}$, our separation is at most a linear
factor (in~$\pebsp \leq n$) off from the optimal. It would be
interesting to determine how large the separation can be. We do not
rule out the possibility that the separation we give might in fact be
asymptotically optimal.

\section*{Acknowledgements}

We are grateful to 
\mbox{Anna Gál}, 
\mbox{Yuval Filmus}, 
\mbox{Toniann Pitassi}, 
and 
\mbox{Robert Robere}
for stimulating discussions on the topic of pebble games.
A special thanks goes to 
\mbox{Mladen Mik\v{s}a}, 
who participated in the initial stages of this work but somehow
managed to avoid the pebbling addiction that seized the rest of us\ldots

\TheauthorSMC performed part of this work while at Princeton
University.
The second, third and fourth authors were funded by the
European Research Council under the European Union's Seventh Framework
Programme \mbox{(FP7/2007--2013) /} ERC grant agreement no.~279611.
\TheauthorJN 
was also supported by the
Swedish Research Council grants \mbox{621-2012-5645} and \mbox{2016-00782},
and by the Independent Research Fund Denmark grant \mbox{9040-00389B}.

\bibliographystyle{alpha}

\begin{thebibliography}{FNPW10}

\bibitem[AS15]{AS15HighParallel}
Jo\"{e}l Alwen and Vladimir Serbinenko.
\newblock High parallel complexity graphs and memory-hard functions.
\newblock In {\em Proceedings of the 47th Annual ACM Symposium on Theory of
  Computing ({STOC}~'15)}, pages 595\nobreakdash--603, June 2015.

\bibitem[Ben73]{Bennett73LogicalReversibility}
Charles~H. Bennett.
\newblock Logical reversibility of computation.
\newblock {\em {IBM} Journal of Research and Development},
  17(6):525\nobreakdash--532, November 1973.

\bibitem[Ben89]{Bennett89TimeSpaceReversible}
Charles~H Bennett.
\newblock Time/space trade-offs for reversible computation.
\newblock {\em SIAM Journal on Computing}, 18(4):766\nobreakdash--776, August
  1989.

\bibitem[BTV01]{BTV01Reversible}
Harry Buhrman, John Tromp, and Paul Vitányi.
\newblock Time and space bounds for reversible simulation.
\newblock {\em Journal of Physics A: Mathematical and general},
  34:6821\nobreakdash--6830, 2001.
\newblock Preliminary version in \emph{ICALP~'01}.

\bibitem[CFLS95]{CFLS95Probabilistically}
Anne Condon, Joan Feigenbaum, Carsten Lund, and Peter~W Shor.
\newblock Probabilistically checkable proof systems and nonapproximability of
  {PSPACE}\nobreakdash-hard functions.
\newblock {\em Chicago Journal of Theoretical Computer Science}, 1995, October
  1995.
\newblock Preliminary version in \emph{STOC~'93}.

\bibitem[Cha73]{Chandra73Efficient}
Ashok~K. Chandra.
\newblock Efficient compilation of linear recursive programs.
\newblock In {\em Proceedings of the 14th Annual Symposium on Switching and
  Automata Theory ({SWAT}~'73)}, pages 16\nobreakdash--25, 1973.

\bibitem[Cha13a]{Chan13JustAPebble}
Siu~Man Chan.
\newblock Just a pebble game.
\newblock In {\em Proceedings of the 28th Annual {IEEE }Conference on
  Computational Complexity ({CCC}~'13)}, pages 133\nobreakdash--143, June 2013.

\bibitem[Cha13b]{Chan13Thesis}
Siu~Man Chan.
\newblock {\em Pebble Games and Complexity}.
\newblock PhD thesis, University of California at Berkeley, 2013.

\bibitem[Coo74]{Cook74ObservationTimeStorageTradeOff}
Stephen~A. Cook.
\newblock An observation on time-storage trade off.
\newblock {\em Journal of Computer and System Sciences},
  9(3):308\nobreakdash--316, 1974.
\newblock Preliminary version in \emph{STOC~'73}.

\bibitem[CP14]{CP14Monotone}
Siu~Man Chan and Aaron Potechin.
\newblock Tight bounds for monotone switching networks via {F}ourier analysis.
\newblock {\em Theory of Computing}, 10:389\nobreakdash--419, October 2014.
\newblock Preliminary version in \emph{STOC~'12}.

\bibitem[CS76]{CS76Storage}
Stephen~A. Cook and Ravi Sethi.
\newblock Storage requirements for deterministic polynomial time recognizable
  languages.
\newblock {\em Journal of Computer and System Sciences},
  13(1):25\nobreakdash--37, 1976.
\newblock Preliminary version in \emph{STOC~'74}.

\bibitem[DL17]{DL17Inapproximability}
Erik~D. Demaine and Quanquan~C. Liu.
\newblock Inapproximability of the standard pebble game and hard to pebble
  graphs.
\newblock In {\em Proceedings of the 15th International Symposium on Algorithms
  and Data Structures ({WADS} '17)}, pages 313\nobreakdash--324, July 2017.

\bibitem[DNW05]{DNW05Pebbling}
Cynthia Dwork, Moni Naor, and Hoeteck Wee.
\newblock Pebbling and proofs of work.
\newblock In {\em Proceedings of the 25th Annual International Cryptology
  Conference ({CRYPTO}~'05)}, volume 3621 of {\em Lecture Notes in Computer
  Science}, pages 37\nobreakdash--54. Springer, August 2005.

\bibitem[DT85]{DT85Speedups}
Patrick~W. Dymond and Martin Tompa.
\newblock Speedups of deterministic machines by synchronous parallel machines.
\newblock {\em Journal of Computer and System Sciences},
  30(2):149\nobreakdash--161, April 1985.
\newblock Preliminary version in \emph{STOC~'83}.

\bibitem[FNPW10]{FNPW10Unpublished}
Yuval Filmus, Jakob Nordström, Toniann Pitassi, and Yu~Wu.
\newblock Unpublished note, 2010.

\bibitem[FPRC13]{FPRC13Average}
Yuval Filmus, Toniann Pitassi, Robert Robere, and Stephen~A Cook.
\newblock Average case lower bounds for monotone switching networks.
\newblock In {\em Proceedings of the 54th Annual {IEEE} Symposium on
  Foundations of Computer Science ({FOCS}~'13)}, pages 598\nobreakdash--607,
  November 2013.

\bibitem[GLT80]{GLT80PebblingProblemComplete}
John~R. Gilbert, Thomas Lengauer, and Robert~Endre Tarjan.
\newblock The pebbling problem is complete in polynomial space.
\newblock {\em SIAM Journal on Computing}, 9(3):513\nobreakdash--524, August
  1980.
\newblock Preliminary version in \emph{STOC~'79}.

\bibitem[HP10]{HP10PspaceCompleteness}
Philipp Hertel and Toniann Pitassi.
\newblock The {PSPACE}-completeness of black-white pebbling.
\newblock {\em SIAM Journal on Computing}, 39(6):2622\nobreakdash--2682, April
  2010.
\newblock Preliminary version in \emph{FOCS~'07}.

\bibitem[HPV77]{HPV77TimeVsSpace}
John Hopcroft, Wolfgang Paul, and Leslie Valiant.
\newblock On time versus space.
\newblock {\em Journal of the ACM}, 24(2):332\nobreakdash--337, April 1977.
\newblock Preliminary version in \emph{FOCS~'75}.

\bibitem[Kr{\'a}04]{Kralovic04TimeSpaceReversible}
Richard Kr{\'a}lovi{\v{c}}.
\newblock Time and space complexity of reversible pebbling.
\newblock {\em RAIRO -- Theoretical Informatics and Applications},
  38(02):137\nobreakdash--161, April 2004.

\bibitem[Lin78]{Lingas78PSPACE}
Andrzej Lingas.
\newblock A {PSPACE}-complete problem related to a pebble game.
\newblock In {\em Proceedings of the 5th Colloquium on Automata, Languages and
  Programming ({ICALP}~'78)}, pages 300\nobreakdash--321, 1978.

\bibitem[LMT00]{LMT00Reversible}
Klaus-Jörn Lange, Pierre McKenzie, and Alain Tapp.
\newblock Reversible space equals deterministic space.
\newblock {\em Journal of Computer and System Sciences},
  60(2):354\nobreakdash--367, April 2000.

\bibitem[LTV98]{LTV98Reversible}
Ming Li, John Tromp, and Paul Vit{\'a}nyi.
\newblock Reversible simulation of irreversible computation.
\newblock {\em Physica D: Nonlinear Phenomena}, 120(1--2):168\nobreakdash--176,
  September 1998.

\bibitem[LV96]{LV96Reversibility}
Ming Li and Paul Vitányi.
\newblock Reversibility and adiabatic computation: Trading time and space for
  energy.
\newblock {\em Proceedings of the Royal Society of London, Series A},
  452(1947):769\nobreakdash--789, April 1996.

\bibitem[Nor13]{Nordstrom13SurveyLMCS}
Jakob Nordström.
\newblock Pebble games, proof complexity and time-space trade-offs.
\newblock {\em Logical Methods in Computer Science},
  9(3):15:1\nobreakdash--15:63, September 2013.

\bibitem[Nor20]{Nordstrom09PebblingSurveyFTTCS}
Jakob Nordström.
\newblock New wine into old wineskins: A survey of some pebbling classics with
  supplemental results.
\newblock Manuscript in preparation. To appear in \emph{Foundations and Trends
  in Theoretical Computer Science}. Current draft version available at
  \url{http://www.csc.kth.se/~jakobn/research/}, 2020.

\bibitem[PH70]{PH70Comparative}
Michael~S. Paterson and Carl~E. Hewitt.
\newblock Comparative schematology.
\newblock In {\em Record of the Project MAC Conference on Concurrent Systems
  and Parallel Computation}, pages 119\nobreakdash--127, 1970.

\bibitem[Pip80]{Pippenger80Pebbling}
Nicholas Pippenger.
\newblock Pebbling.
\newblock Technical Report RC8258, IBM Watson Research Center, 1980.
\newblock In \emph{Proceedings of the 5th IBM Symposium on Mathematical
  Foundations of Computer Science}.

\bibitem[Pot10]{Potechin10Monotone}
Aaron Potechin.
\newblock Bounds on monotone switching networks for directed connectivity.
\newblock In {\em Proceedings of the 51st Annual {IEEE} Symposium on
  Foundations of Computer Science ({FOCS}~'10)}, pages 553\nobreakdash--562,
  October 2010.

\bibitem[PTC77]{PTC76SpaceBounds}
Wolfgang~J. Paul, Robert~Endre Tarjan, and James~R. Celoni.
\newblock Space bounds for a game on graphs.
\newblock {\em Mathematical Systems Theory}, 10:239\nobreakdash--251, 1977.

\bibitem[RM99]{RM99Separation}
Ran Raz and Pierre McKenzie.
\newblock Separation of the monotone {NC} hierarchy.
\newblock {\em Combinatorica}, 19(3):403\nobreakdash--435, March 1999.
\newblock Preliminary version in \emph{FOCS~'97}.

\bibitem[Seg07]{Segerlind07Complexity}
Nathan Segerlind.
\newblock The complexity of propositional proofs.
\newblock {\em Bulletin of Symbolic Logic}, 13(4):417\nobreakdash--481,
  December 2007.

\bibitem[Set75]{Sethi75CompleteRegisterAllocation}
Ravi Sethi.
\newblock Complete register allocation problems.
\newblock {\em SIAM Journal on Computing}, 4(3):226\nobreakdash--248, September
  1975.

\bibitem[SM73]{SM73WordProblems}
Larry~J. Stockmeyer and Albert~R. Meyer.
\newblock Word problems requiring exponential time (preliminary report).
\newblock In {\em Proceedings of the 5th Annual {ACM} Symposium on Theory of
  Computing ({STOC}~'73)}, pages 1\nobreakdash--9, 1973.

\bibitem[SS77]{SS77SpaceTimeFFT}
Sowmitri Swamy and John~E. Savage.
\newblock Space-time trade-offs on the {FFT}-algorithm.
\newblock Technical Report CS-31, Brown University, 1977.

\bibitem[SS79]{SS79SpaceTimeOblivious}
John~E. Savage and Sowmitri Swamy.
\newblock Space-time tradeoffs for oblivious interger multiplications.
\newblock In {\em Proceedings of the 6th International Colloquium on Automata,
  Languages and Programming ({ICALP}~'79)}, pages 498\nobreakdash--504, 1979.

\bibitem[SS83]{SS83SpaceTimeLinear}
Sowmitri Swamy and John~E. Savage.
\newblock Space-time tradeoffs for linear recursion.
\newblock {\em Mathematical Systems Theory}, 16(1):9\nobreakdash--27, 1983.

\bibitem[Tom78]{Tompa78TimeSpaceComputing}
Martin Tompa.
\newblock Time-space tradeoffs for computing functions, using connectivity
  properties of their circuits.
\newblock In {\em Proceedings of the 10th Annual ACM symposium on Theory of
  computing ({STOC}~'78)}, pages 196\nobreakdash--204, 1978.

\bibitem[VT89]{VT89NewPebbleGame}
H.~Venkateswaran and Martin Tompa.
\newblock A new pebble game that characterizes parallel complexity classes.
\newblock {\em {SIAM} Journal on Computing}, 18(3):533\nobreakdash--549, June
  1989.
\newblock Preliminary version in \emph{FOCS~'86}.

\bibitem[WAPL14]{WAPL14Inapproximability}
Yu~Wu, Per Austrin, Toniann Pitassi, and David Liu.
\newblock Inapproximability of treewidth and related problems.
\newblock {\em Journal of Artificial Intelligence Research},
  49:569\nobreakdash--600, April 2014.

\bibitem[Wil00]{Williams00Reversible}
Ryan Williams.
\newblock Space-efficient reversible simulations.
\newblock Technical report, Cornell University, 2000.
\newblock Available at \url{http://web.stanford.edu/~rrwill/spacesim9_22.pdf}.

\end{thebibliography}

\end{document}